\documentclass[acmsmall,screen,nonacm]{acmart}

\AtBeginDocument{%
  }

\setcopyright{acmlicensed}
\copyrightyear{2018}
\acmYear{2018}
\acmDOI{XXXXXXX.XXXXXXX}

\usepackage{wrapfig}
\usepackage{tabularx}
\usepackage{multirow}
\usepackage{soul}
\usepackage{subcaption}
\usepackage{physics}

\usepackage[ruled,vlined,linesnumbered,resetcount]{algorithm2e}
\allowdisplaybreaks

\providecommand{\Junyu}[1]{{\protect\color{brown}\noindent {\bf[Junyu]} \emph{#1} {\bf [/Junyu]}}}

\newcommand{\edge}[1]{\ensuremath{\langle #1 \rangle}}

\newcommand{\myCompilerName}{MarQSim}
\newcommand{\myCompilerNameSpace}{MarQSim }

\newtheorem{theorem}{Theorem}[section]
\newtheorem{corollary}{Corollary}[section]

\newtheorem{remark}{Remark}[section]
\newtheorem{definition}{Definition}[section]
\newtheorem{proposition}{Proposition}[section]
\newtheorem{example}{Example}[section]

\begin{document}

\title{MarQSim: Reconciling Determinism and Randomness in
Compiler Optimization for Quantum Simulation}

% \affiliation{%

\author{Xiuqi Cao}
\affiliation{%
  \institution{University of Pennsylvania}
  \city{Philadelphia}
  \country{USA}}
\email{}

\author{Junyu Zhou}
\email{junyuzh@seas.upenn.edu}
\affiliation{%
  \institution{University of Pennsylvania}
  \city{Philadelphia}
  \country{USA}}

\author{Yuhao Liu}
\affiliation{%
  \institution{University of Pennsylvania}
  \city{Philadelphia}
  \country{USA}}
\email{liuyuhao@seas.upenn.edu}

\author{YUnong Shi}
\affiliation{%
 \institution{AWS Quantum Technology}
 \city{New York}
 \country{USA}}

\author{Gushu Li}
\affiliation{%
  \institution{University of Pennsylvania}
  \city{Philadelphia}
  \country{USA}}
\email{gushuli@seas.upenn.edu}


\begin{abstract}
  Quantum simulation, fundamental in quantum algorithm design, extends far beyond its foundational roots, powering diverse quantum computing applications. However, optimizing the compilation of quantum Hamiltonian simulation poses significant challenges. Existing approaches fall short in reconciling deterministic and randomized compilation, lack appropriate intermediate representations, and struggle to guarantee correctness. Addressing these challenges, we present \myCompilerName, a novel compilation framework. \myCompilerNameSpace leverages a Markov chain-based approach, encapsulated in the Hamiltonian Term Transition Graph, adeptly reconciling deterministic and randomized compilation benefits. We rigorously prove its algorithmic efficiency and correctness criteria. Furthermore, we formulate a Min-Cost Flow model that can tune transition matrices to enforce correctness while accommodating various optimization objectives. Experimental results demonstrate \myCompilerName's superiority in generating more efficient quantum circuits for simulating various quantum Hamiltonians while maintaining precision.
\end{abstract}

\maketitle

\section{Introduction}

One of the most fundamental principles in quantum algorithm design is the concept of quantum simulation, often referred to as Hamiltonian simulation. Inspired initially by Feynman's visionary proposal for developing a quantum computer~\cite{feynman1982simulating}, this notion represents a pivotal application of quantum computing~\cite{lloyd1996universal, abrams1999quantum}.
Quantum simulation involves emulating quantum physical systems, a task of paramount importance that extends beyond simulation itself. Quantum algorithms have since harnessed the power of simulation to address an array of applications, such as solving linear system~\cite{harrow2009quantum}, quantum principal component analysis~\cite{lloyd2014quantum}, and quantum support vector machine~\cite{rebentrost2014quantum}. These algorithms necessitate the simulation of quantum systems carefully crafted to address specific computational problems, %further emphasizing the significance of quantum simulation. 
Notably, alongside digital quantum computing, the analog quantum simulation becomes a promising candidate to demonstrate quantum advantage since it provides excellent efficiency by directly mapping the simulated Hamiltonian onto the device Hamiltonian without the quantum circuit abstraction~\cite{peng2023simuq, daley2022practical}.

%. In this approach, quantum systems directly map the simulated Hamiltonian onto the device Hamiltonian, introducing a distinct avenue for quantum simulation.
%This convergence of digital and analog quantum simulation underscores the need for sophisticated quantum compilation techniques to optimize these processes, spanning both established and emerging quantum computing paradigms.

The quantum simulation process is usually described by the operator $e^{i\mathcal{H}t}$ where $\mathcal{H}$ is the Hamiltonian of the target quantum system and $t\in\mathbb{R}$ is the evolution time of the simulated system.
To implement this operator on a quantum computer, a common practice is to decompose the Hamiltonian into a weighted sum of some Hamiltonian terms $\mathcal{H} = \sum_jh_jH_j$ where $h_j \in \mathbb{R}$ is the weight, and $H_j$ is the Hamiltonian term. A typical selection of $H_j$ is the Pauli strings due to the simple implementation of $e^{i\theta_jH_j}$ ($\theta_j \in \mathbb{R}$ is a real parameter and $H_j$ is a Pauli string). 
Finally, implementing operator $e^{i\mathcal{H}t}$ is turned into combining the simulation of all these Hamiltonian terms $e^{i\theta_jH_j}$ and the compilation usually results in a very long quantum gate sequence.
Even for the analog quantum simulation without the quantum circuit abstraction, decomposing and manipulating the Hamiltonian terms are still required in its Hamiltonian compilation~\cite{peng2023simuq}.

Optimizing the compilation of quantum Hamiltonian simulation is imperative to harness the potential of quantum simulation fully.
However, the existing compilation approaches for quantum Hamiltonian simulation are still far from optimal, and we briefly summarize three major challenges.

%While existing quantum compilation techniques have primarily focused on low-level optimizations, they often overlook a significant portion of the optimization space. Some initial efforts have begun to exploit the algorithmic properties of quantum simulation. Nonetheless, these approaches tend to be static and lack the flexibility to encompass the full range of algorithmic optimization opportunities.

%The inherent challenge in optimization lies in the dilemma of whether to statically order the Hamiltonian terms or randomize their ordering. This quandary has yet to be effectively resolved, leaving substantial room for improvement.

%This paper introduces a novel approach to addressing the compilation problem. We formulate it as a Markov stochastic process, presenting a new intermediate representation founded on the Markov transition graph. By doing so, we transform program compilation into the task of sampling from this stochastic process. To substantiate our claims, we provide formal proofs of algorithmic efficiency and error bounds. Furthermore, we propose a solution for solving the transition matrix by means of a min-cost flow problem.

\textbf{Seemingly Incompatible Compilation Approaches}:
The first challenge in compilation for quantum simulation is that one of the potential optimization opportunities comes from two seemingly incompatible approaches: deterministic compilation and random compilation.
After the Hamiltonian is decomposed into the summation of Hamiltonian terms, the deterministic compilation will find a specific order to execute the simulation of each term and repeat this order many times to approximate the simulation of the entire Hamiltonian. 
The Hamiltonian terms can order for various optimization objectives, such as more gate cancellation~\cite{gui2020term,li2022paulihedral,hastings2015improving}, shorter circuit depth~\cite{li2022paulihedral}, reduced approximation error~\cite{gui2020term}, circuit simplification~\cite{cowtan2019phase, cowtan2020generic, de2020architecture, van2020circuit}, etc.
In contrast, the randomized approach will randomly sort the Hamiltonian terms~\cite{childs2019faster,ouyang2020compilation} or even randomly sample the Hamiltonian terms to assemble the final simulation circuit~\cite{campbell2019random}. Theoretical analysis shows that the randomized compilation can converge asymptotically faster with lower error bounds, but the benefits from deterministic compilation are mostly lost.
Previous works lie in either one of these two categories, as these two types of approaches are seemingly incompatible. To the best of our knowledge, none of them could simultaneously leverage the benefits of both approaches.

\textbf{Lack of Proper Intermediate Representation}:
Modern classical compilers %usually
employ multiple intermediate representations (IRs) (e.g., control flow graph, static single assignment) from high-level to low-level to accommodate different optimizations. % are applied on different IRs. % during compilation.
Today's quantum compilers~\cite{Qiskit,sivarajah2020t,amy2020staq,khammassi2020openql,mccaskey2021mlir}, on the other hand, are mostly built around low-level representations~\cite{cross2017open,smith2016practical,kissinger2019pyzx}. %, which makes it difficult.
Recently, several new languages, including the Pauli IR in Paulihedral~\cite{li2022paulihedral}, the phase gadgets~\cite{cowtan2019phase} in TKET~\cite{sivarajah2020t}, and the Hamiltonian Modeling Language in SimuQ~\cite{peng2023simuq}, are proposed to help with compiling quantum Hamiltonian simulation.
These languages capture the individual Hamiltonian terms and thus can benefit program analysis and optimization in deterministic compilation approaches. 
Yet, to the best of our knowledge, none of them can accommodate optimizing or tuning in the randomized compilation approaches.

%to extract high-level information about the semantics of the algorithm and discover non-commutative yet semantics-preserving re-orderings. The most recent version of open quantum assembly language (OpenQASM)~\cite{cross2021openqasm} recognizes the need for higher-level semantics such as control, inverse, and power operations, but is still incapable of representing Pauli-level semantics which are prevalent in quantum simulation kernels. As we have shown, our Pauli IR can carry high-level semantics through multiple optimization stages, encode all known algorithm constraints, and is compatible with further low-level optimizations by these tools.

\textbf{Correctness Guarantee for Quantum Simulation Compilation}:
Generic quantum compilation optimizations, such as circuit rewriting~\cite{soeken2013white}, gate cancellation~\cite{nam2018automated}, template matching~\cite{maslov2008quantum}, and qubit mapping~\cite{murali2019noise},  usually execute very small-scale quantum program transformations in each step and thus can be checked by evaluating the gate matrices or formally verified (e.g., VOQC~\cite{hietala2023a}, Giallar~\cite{tao2022giallar}, VyZX~\cite{lehmann2022vyzx}).
In contrast, compiling the quantum simulation from the input Hamiltonian to the circuit implementation is a large-scale program transformation involving many gates and qubits.
Evaluating the correctness or approximation error requires calculating the exponentially large gate matrix of the entire Hamiltonian simulation circuits.
Previous deterministic compilation approaches are governed by the same trotterization product formula~\cite{trotter1959product}. As a result, they cannot break its approximation error bound.
The randomized compilation approaches~\cite{childs2019faster,ouyang2020compilation,campbell2019random}, on the other hand, rely on sophisticated ad-hoc algorithm performance analysis. 
The correctness guarantee for any new quantum simulation compilation involving randomization is non-trivial.

%Rigorous correctness guarantee and error bounds are required \todo{hmmm?}

\begin{wrapfigure}{r}{0.5\textwidth}
    \begin{center}
        \includegraphics[width=0.5\textwidth]{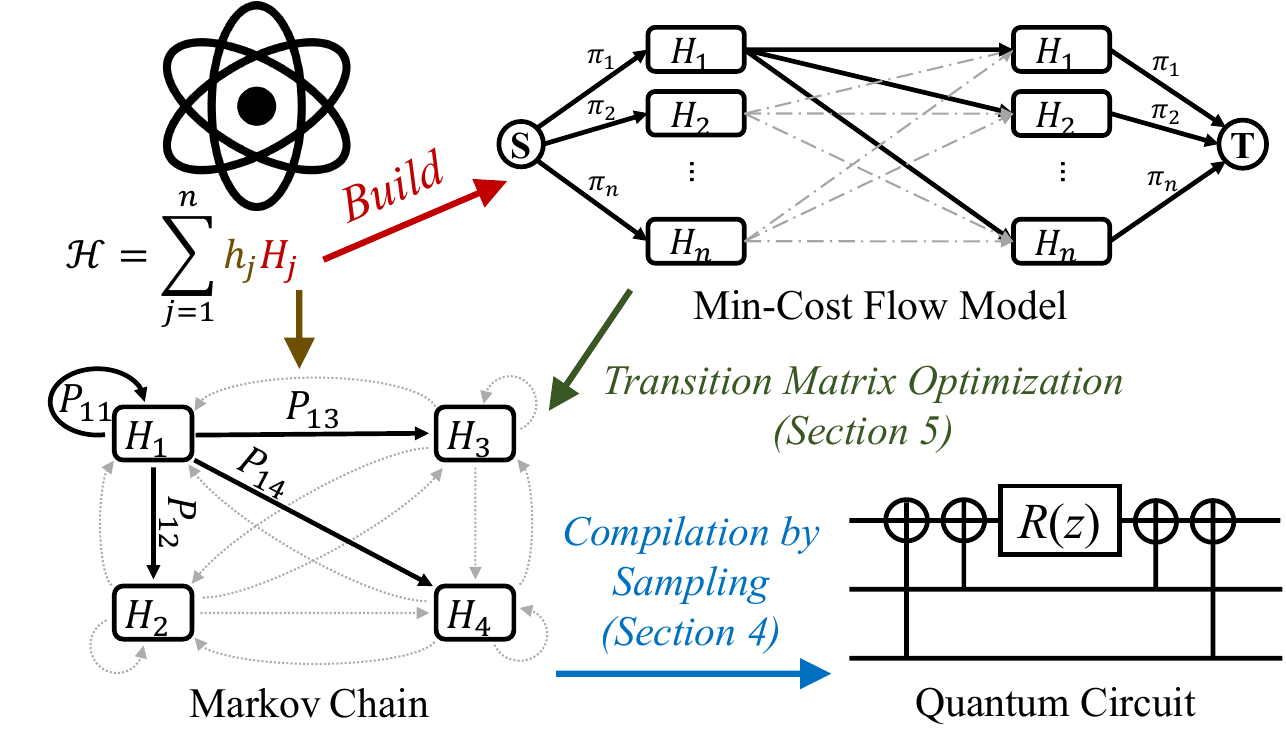}
    \end{center}
    \vspace{-10pt}
    \caption{Overview of \myCompilerNameSpace compiler framework}
    \label{fig:markqsim-overview}
\end{wrapfigure}
To this end, we propose \myCompilerName, a new framework that can advance the compilation for quantum Hamiltonian simulation programs by reconciling the deterministic and randomized compilation. An overview of \myCompilerNameSpace is illustrated in Fig.~\ref{fig:markqsim-overview}.
\textbf{First}, \myCompilerNameSpace formulates the compilation process into sampling from a Markov chain. Sampling from a random process can naturally inherit the benefits from randomized compilation. The probability transition matrix in the Markov chain can be tuned to change the sampling tendency and thus incorporate the benefits from deterministic compilation.
\textbf{Second}, \myCompilerNameSpace will convert an input Hamiltonian into a newly designed IR, the Hamiltonian Term Transition Graph, to encode the Hamiltonian information into a homogeneous Markov chain. The compilation process is to sample from a Markov chain whose state transition graph is our IR. As the key parameter, we rigorously prove the sufficient conditions for the Markov chain transition matrix to guarantee the overall correctness and algorithmic advantage of \myCompilerNameSpace compilation.
\textbf{Third}, in addition to the sufficient condition, we provide a practical approach to obtain a valid and optimized transition matrix by solving a Min-Cost Flow problem.
By carefully designing the flow architecture, the conditions for correctness and other optimization objectives can be intuitively encoded into the cost and capacity functions in the flow network. The solved flow can then be turned back into a transition matrix.
\textbf{Finally}, we show that different optimization objectives represented by multiple transition matrices can be easily reconciled in \myCompilerNameSpace by a normalized weighted summation of those individual transition matrices. 
We also propose to evaluate the convergence speed of the sampling process by analyzing the spectra of different transition matrices.
In this paper we demonstrate the combination of randomized compilation, CNOT gate cancellation, and random perturbation while extending to more optimization objectives is also possible.

The major contributions of this paper can be summarized as follows:
\begin{enumerate}
    \item We propose \myCompilerName, a new compilation framework that can reconcile the benefits from previous deterministic and randomized compilation techniques, even if they seem incompatible, and thus opens up new optimization space for quantum simulation.
    \item We designed a new IR, the Hamiltonian Term Transition Graph, for \myCompilerNameSpace to naturally formulate the compilation process into sampling from a Markov chain and a corresponding compilation algorithm. We rigorously proved the sufficient conditions to guarantee the correctness and the algorithmic advantage of \myCompilerNameSpace compilation.
    \item To further incorporate more optimization objectives from deterministic compilation, we provide a Min-Cost Flow model to turn the Markov transition matrix for different ordering tendencies while the correctness conditions are enforced.
    \item Experimental results show that \myCompilerNameSpace can outperform the baseline in compiling a wide range of quantum Hamiltonian simulations. The circuits generated by \myCompilerNameSpace contain significantly fewer gates when maintaining similar approximation accuracy.
\end{enumerate}

\section{Preliminary}

%This section introduces the necessary background to help understand this article (including the section-\ref{sec:rwork}). 
In this section, we introduce the necessary preliminary about quantum computing, quantum simulation, Markov process, and minimum-cost flow.
%It is impossible to cover every essential part of quantum computing. For more basic concepts (e.g., qubit, gate, linear operator, circuit), we recommend~\cite{nielsen2010quantum} for more details.

\subsection{Qubit and Quantum State}
%\Junyu{
%A quantum state is usually constructed of several qubit states in quantum computing.
The basic information processing unit in quantum computing is the qubit.
The state of a single qubit is defined by a vector $\vert \psi \rangle= \alpha \vert 0 \rangle + \beta \vert 1 \rangle$ with $|\alpha|^2 + |\beta|^2 = 1$ in the two-dimensional Hilbert space spanned by the $\vert 0 \rangle$ and $\vert 1 \rangle$ basis. This representation can be perceived as the qubit simultaneous being in $\vert 0 \rangle$ and $\vert 1 \rangle$ basis with amplitude $\alpha$ and $\beta$, known as the superposition. After a measurement, the superposition state collapses to $\vert 0 \rangle $ or $\vert 1 \rangle$, with probability $|\alpha|^2$ and $|\beta|^2$, respectively.%}

% \Junyu{
% As described, the $N$-qubit state falls on a $2^N$-dimensional Hilbert space, in which the basis is described by a string of binary numbers with length equal to $N$: $\vert x_0x_1\cdots x_N  \rangle\ $, $ x_i \in \{0,1 \}$. Then, the $N$-qubit state is superposition among this basis with the probability sum to 1. Each basis in multi-qubit system can be viewed as tensor product of $\{\vert 0 \rangle\ ,\vert 1 \rangle\ \}$ in single qubit state, for two-qubit system: $\vert 01 \rangle\ = \vert 0 \rangle\ \bigotimes \vert 1 \rangle\ $. Also, we can construct 
%  a multi-qubit state by combining single-qubit states using tensor product in this manner $\vert \psi \rangle\ = \vert \psi_1 \rangle\ \bigotimes \vert \psi_1 \rangle\ \bigotimes \cdots \bigotimes \vert \psi_N \rangle\ $. However, the reverse is not true, the EPR pair $\frac{\vert 00 \rangle\ + \vert 11 \rangle\ }{ \sqrt{2}} $ can not be decomposed into any tensor product of single-qubit state, this phenomenon is also known as quantum entanglement.
% }

%\yhliu{
To extend this idea, the $N$-qubit state is defined by a vector in a $2^N$-dimensional Hilbert space. The basis of this Hilbert space could be constructed via the Kronecker product of $N$ 2-dimensional Hilbert space basis as $\ket{x_1x_2\dots x_N}=\ket{x_1}\otimes\ket{x_2}\otimes\dots\otimes\ket{x_N}$ where $\ket{x_i}\in\{\ket{0},\ket{1}\}$. Then, an $N$-qubit state can be a superposition among this basis. % with the probability sum to 1. 
%If we fix the basis order, 
Once the basis is fixed,
the basis's superposition coefficients form the \textit{state vector} of a $N$-qubit state, which is also a vector in $2^N$-dimensional Hilbert space. 
The rest of this paper will always use this $\vert 0 \rangle$ and $\vert 1 \rangle$ basis. Also, we can construct a multi-qubits state by combining single-qubit states using the Kronecker product in this manner $\ket\psi= \ket\psi_1\otimes\ket\psi_2\otimes\dots\otimes\ket\psi_N$. However, it does not entirely hold reversely. For example, the EPR pair state $(\ket{00}+\ket{11})/\sqrt{2}$ can not be decomposed into the Kronecker product of any two single-qubit states. This phenomenon is also known as entanglement.%}

\subsection{Quantum Program and Circuit}
%\Junyu{
A quantum program (circuit) usually consists of a set of quantum gates (operators) and measurements. %constructed by basic quantum operators (gates). 
%Operators act on the quantum state in the Schr\"{o}dinger picture. 
A quantum gate is a unitary matrix that applies to the state vector of the qubits.
%Such operation could be viewed directly via the matrix representation of the quantum gate and the state vector of qubits under $\ket{0},\ket{1}$ basis.%}
%\begin{example}
For example, the $X$ gate is a single-qubit gate. Applying the $X$ gate to a single-qubit state $\alpha\ket{0}+\beta\ket{1}$ follows matrix multiplication:
%Take the $X$ gate as an example. Let it operates on the state $\alpha\ket{0}+\beta\ket{1}$ using the matrix representation:
\begin{equation*}
\footnotesize
    X = \mqty(0 & 1 \\ 1 & 0)\Rightarrow X\mqty(\alpha \\ \beta)=\mqty(0 & 1 \\ 1 & 0)\mqty(\alpha \\ \beta)=\mqty(\beta \\ \alpha)
\end{equation*}
Thus, gate $X$ transforms the state $\alpha\ket{0}+\beta\ket{1}$ to $\beta\ket{0}+\alpha\ket{1}$.
%\end{example}
%\yhliu{To extend, 
A Quantum gate applied on $N$ qubits can always be represented by a $2^N\times 2^N$ complex unitary matrix. For example, the two-qubit $CNOT$ gate has the following matrix representation:
\begin{equation*}\footnotesize
    CNOT=\begin{pmatrix}
    1 & 0 & 0 & 0 \\
    0 & 1 & 0 & 0 \\ 
    0 & 0 & 0 & 1\\ 
    0 & 0 & 1 & 0
    \end{pmatrix}
\end{equation*}

%\Junyu{

\begin{wrapfigure}{R}{0.35\textwidth}
    \centering
    \vspace{-5pt}
    \includegraphics[width=0.25\textwidth]{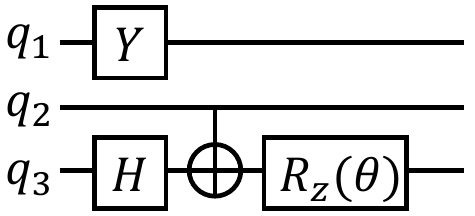}
    \vspace{-5pt}
    \caption{Quantum circuit example}
    \label{fig:quantum-program-representation}
\end{wrapfigure}
In this paper, the quantum programs are represented in the quantum circuit diagram. Fig.~\ref{fig:quantum-program-representation} shows an example of a 3-qubit circuit.
In the quantum circuit, each \textit{horizontal line} represents a qubit. Fig~\ref{fig:quantum-program-representation} has three qubits: $q_1$, $q_2$, and $q_3$. 
The \textit{square blocks and $\oplus$ symbols} on the horizontal lines represent gates applied on the qubits associated with the horizontal lines. Fig~\ref{fig:quantum-program-representation} has three single-qubit gates: one $Y$ gate on $q_1$, one $H$ gate on $q_3$, and a $R_z$ gate on $q_3$. Besides, a $CNOT$ gate is applied on $(q_2,q_3)$ between the $H$ gate and the $R_z$ gate.
%A $Y$ gate acting on $q_1$, an $H$ gate acting on $q_3$, a $CNOT$ gate acting on $(q_2,q_3)$, and finally, a rotation gate $R_z(\theta)$ on $q_3$.
%Quantum programs could be represented in several equivalent forms. Let's consider a system with three qubits $q_1$, $q_2$ and $q_3$. The circuit is constructed by a $Y$ gate acting on $q_1$, a $Hadamard$ gate acting on $q_3$, a $CNOT$ gate acting on $(q_2,q_3)$, and finally, a rotation gate $R_z(\theta)$ on $q_3$. The circuit could be represented by, in Figure~\ref{fig:quantum-program-representation}, (a) Quantum Circuit Graph (left): the $q_1,q_2,q_3$ and three horizontal wires stand for the initial quantum state and their evolution through time. The boxes on the wires represent quantum gates and their parameters; in our example, one $Y$ gate, one $H$ gate, one $CNOT$, and a $R_z$ gate. %(b) Intermediate Presentation (middle): the text intermediate representation specifies the order and position of each quantum gate. (c) Matrix Representation (right): it shows that the gates in the same layer are combined by tensor product, and matrices in different layers are joined by matrix multiplication.
%}

\subsection{Quantum Simulation} \label{quantumsimulation}

%In this subsection, we first show the basic knowledge of quantum simulation, and then we introduce the hierarchies of the quantum simulation operator in quantum circuits.

%\subsubsection{Introduction to Quantum Simulation}
%\ 

%Let us consider a general quantum simulation problem, which involves determining the state of a quantum system described by the wavefunction $\left | \phi_i \right \rangle$ at a specific evolutionary time $t$, and computing the value of a desired physical quantity. For simplicity, we focus on time-independent Hamiltonians denoted as $H$ (a complex matrix of high dimension).

%The Schrödinger equation which is vital in the quantum mechanics:
%\begin{equation}
%\frac{d}{dt} \left | \phi_i \right \rangle = iH\left |\phi_i \right \rangle,
%\end{equation}
%%has the solution which is given by,
%\begin{equation}
%\left | \phi_i(t) \right \rangle = e^{iHt} \left | \phi_i(0) \right \rangle,
%\end{equation}
%where (the exponential of a matrix is defined as)
%\begin{equation}
%e^{iHt} = I + \sum_{j=1}^{\infty} (iHt)^{j},
%\end{equation}
%with $I$ being the identity matrix. 
%Numerically computing the initial state $\left | \phi_i(0) \right \rangle$ is computationally challenging for classical computers. 
%However, quantum algorithms executed on quantum computers offer a potential solution to this problem.

%\subsubsection{ Trotter-Suzuki decomposition }

%\yhliu{

Quantum Hamiltonian simulation involves implementing the gate $e^{i\mathcal{H}t}$ on a quantum computer where $\mathcal{H}$ is the Hamiltonian of the simulated system. 
To implement the gate $e^{i\mathcal{H}t}$ with basic single- and two-qubit gates, a common practice is to decompose the Hamiltonian into a weighted sum of an array of Hamiltonian terms $\mathcal{H} = \sum_jh_jH_j$ where $h_j \in \mathbb{R}$ is the weight and $H_j$ is the Hamiltonian term. A typical selection of $H_j$ is the Pauli strings due to the simple implementation of $e^{i\theta_jH_j}$ ($\theta_j \in \mathbb{R}$ is a real parameter). 
%Finally, implementing operator $e^{i\mathcal{H}t}$ is turned into combining the time evolution of all these Hamiltonian terms $e^{i\theta_jH_j}$

%Quantum simulation takes the Hamiltonian operator $\mathcal{H}$ of a system. It performs simulated time evolution factor $e^i\mathcal{H}t$ on quantum computers without involving the actual system. 
%Quantum computers show their capacity by being capable of simulating large-scale systems, whereas classical computers fail due to the overwhelming complexity. 
%This section will briefly discuss some critical technical details in quantum simulation that are involved in our work.%}

%\textbf{Trotter-Suzuki decomposition}
%Trotter-Suzuki decomposition \cite{suzuki1990fractal, Berry2006, Suzuki1991} is a way to decompose the $e^{i\mathcal{H}t}$ operator into more basic operators. \cite{Whitfield_2011} shows the error bound of the decomposition:
%\begin{equation}
%    e^{i\mathcal{H}t}=(e^{iH_1\Delta t}e^{iH_2\Delta t}\ldots e^{iH_n\Delta t})^{\frac{t}{\Delta t}}+O(t\Delta t),
%\end{equation}
%where
%\begin{equation}
%\label{equ: decomposition of H}
%    \mathcal{H}=\sum_{i=1}^{n}H_i.
%\end{equation}
%In practice, $H_i$ in \eqref{equ: decomposition of H} is always considered a weighted Pauli string (which will be discussed next). 

% {R}{0.5\textwidth}

\begin{wrapfigure}{R}{0.4\textwidth}
    \centering
    \vspace{-5pt}
    \includegraphics[width=0.38\textwidth]{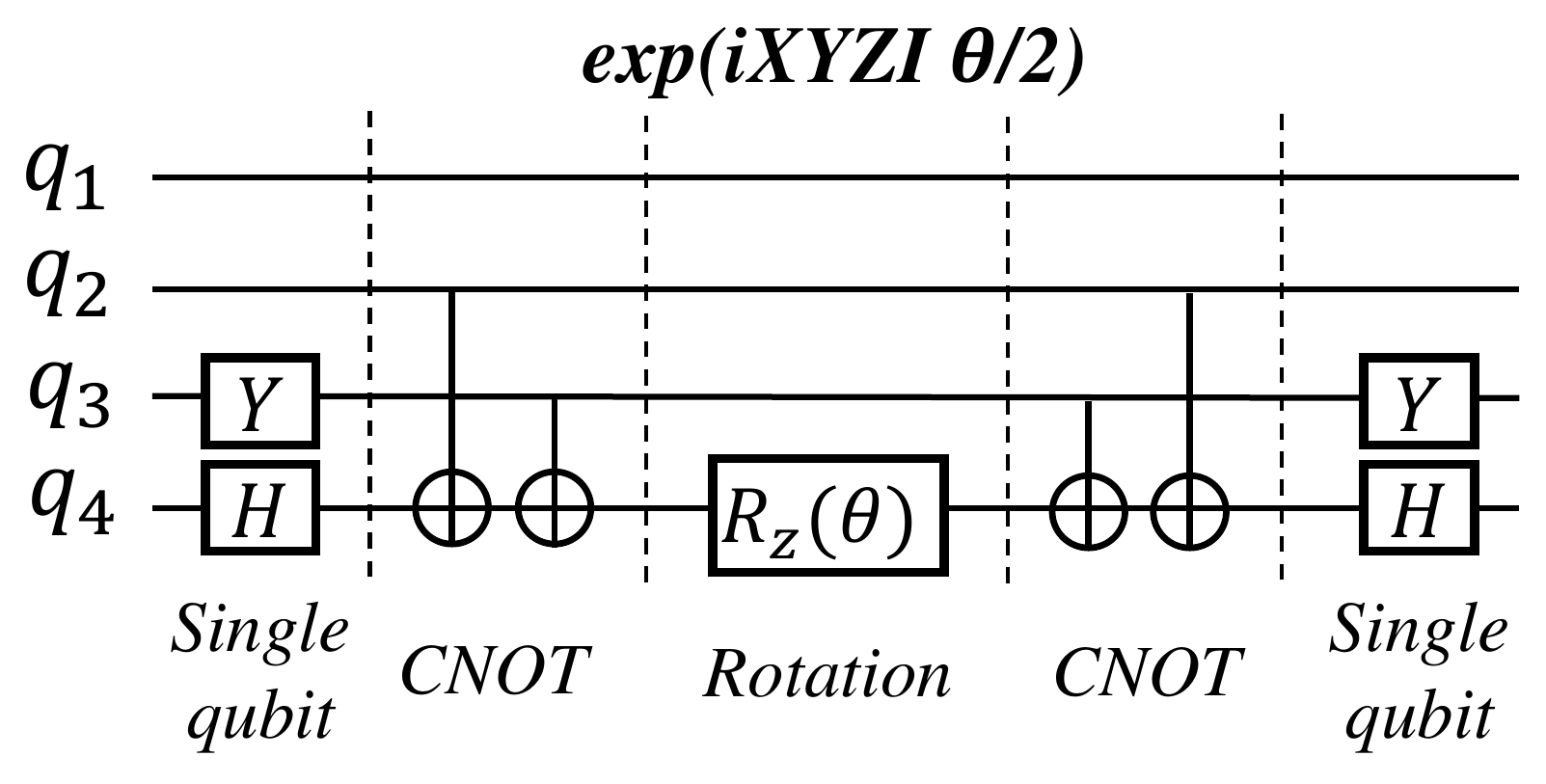}
    \vspace{-10pt}
    \caption{Decomposition of ${\rm exp}({\rm i}X_4Y_3Z_2I_1\frac{\theta}{2})$}
    \vspace{-5pt}
    \label{fig:PSsynthesis}
\end{wrapfigure}

%Pauli strings are fundamental elements of a Hamiltonian (i.e., a Hamiltonian is the sum of weighted Pauli strings). Mathematically, a Hamiltonian can be represented as $\mathcal{H} = \sum_{j} w_j P_j$, where $w_j \in \mathbb{R}$ and each $P_j$ is a Pauli string.

\textbf{Pauli String}: For an $n$-qubit system, %(i.e., a composite system of $n$ qubits), 
an $n$-qubit Pauli string is defined as $P = \sigma_{n} \otimes \sigma_{n-1} \otimes \cdots \otimes \sigma_{1}$ (or $\sigma_{n} \sigma_{n-1} \cdots \sigma_{1}$ in abbreviation), where $\sigma_i \in \{I, X, Y, Z\}$ for $i \in \{1,2,\ldots,n\}$. 
Here, $\otimes$ represents the Kronecker product, $X, Y$, and $Z$ are Pauli operators, and $I$ denotes the identity operator (see Appendix~\ref{sec:pauli-matrix}). The operator $\sigma_i$ acts on the $i$-th qubit.

\textbf{From Pauli Strings to Circuit}:
A critical property of a Pauli string $P$ is that for any $\theta \in \mathbb{R}$ the operator ${\rm exp}({\rm i}P\frac{\theta}{2})$ can be decomposed into basic gates easily. 
We introduce the decomposition process using an example shown in Fig.~\ref{fig:PSsynthesis}, which illustrates the synthesis of ${\rm exp}({\rm i}X_4Y_3Z_2I_1\frac{\theta}{2})$. First, two identical layers of single-qubit gates are added at the beginning and the end of the circuit snippet. In these two layers, the $\rm H$ or $\rm Y$ gates are applied on the qubits whose corresponding operators in the Pauli string are $\rm X$ (i.e., $q_4$) or $\rm Y$ (i.e., $q_3$), respectively.
Next, one qubit whose corresponding operator in the Pauli string is not the identity is selected as the root qubit.
Then, a group of CNOT gates are applied between all other qubits whose corresponding operator is not the identity and this root qubits. 
In this example, $q_4$ is the root qubit. There are CNOT gates on $(q_2, q_4)$ and $(q_3, q_4)$.
%construct a left $\rm CNOT$ tree. The left tree can be generated differently if its non-I nodes (representing qubits whose operators in the Pauli string are not $\rm I$) are connected in order. We use the "star" implementation~\cite{gui2020term} to put the target qubits of CNOT gates on the same qubit.
Then, the $\rm R_z(\theta)$ gate is applied on the root qubit.
Another group of CNOT gates is applied after the $\rm R_z(\theta)$ gate. The CNOT gates are applied on the same qubit pairs as the first group, but the order is reversed.
Note that there are other ways to generate the CNOT gate. The approach in Fig.~\ref{fig:PSsynthesis} is selected to benefit from CNOT gate cancellation optimization introduced in~\cite{gui2020term}.

After the $e^{i\theta_jH_j}$'s are implemented with basic gates, a compiler will further assemble them to realize $e^{i\mathcal{H}t}$. We will introduce various options for this step later in Section~\ref{sec:previouswork}.
%of the left $\rm CNOT$ tree. And the right $\rm CNOT$ tree includes the same $\rm CNOT$ gates as the left tree, but in reverse order.

\iffalse
The left tree can be generated differently if its non-I nodes (representing qubits whose operators in the Pauli string are not $\rm I$) are connected in order. 
The lower one-third of Figure \ref{fig:PSsynthesis} shows three different ways to generate the $\rm CNOT$ tree circuits and their corresponding tree graphs. The tree's $\rm CNOT$ gates should connect all nodes from the leaf nodes to the root node in order. 
Any non-I node in the tree can serve as the root node.
\fi

%\item 

%\end{enumerate}
%The Markov compiler utilizes this algorithmic flexibility in synthesis to increase gate cancellation and reduce mapping overhead.

\subsection{Markov Stochastic Process}
\label{subsec:markov}

Our compiler optimization technique is based on formulating the instruction scheduling and code synthesis into a Markov stochastic process.
A discrete stochastic process $X$ is defined as an array of random variables  $ \{ X_t \mid t \in T \}$ indexed by $t\in T = \{ 0, 1, 2, \ldots \}$. 
Suppose a finite set $S = \{s_1, s_2, s_3, \dots, s_{|S|}\}$ is the state space.
Each random variable $X_t$ in the set $\{ X_0, X_1, X_2, \ldots \}$ takes values in a finite set $S = \{s_1, s_2, s_3, \dots, s_{|S|}\}$ known as the state space. 
%\subsubsection{The definition of Markov Chain}
%\label{subsub:The definition of Markov Chain}
%\ 
\begin{definition}[Markov Chain]
    A stochastic process $X$ is a Markov chain if it satisfies the following condition:
$$\forall n \ge 1,\ \forall s,x_0,x_1,\ldots,x_{n-1}\in S, \
 \Pr[X_n = s | X_{n-1} = x_{n-1}]
= \Pr[X_n = s | X_0 = x_0, \ldots, X_{n-1} = x_{n-1}].$$
The probability of each random variable in a Markov chain is determined solely by the state observed in the preceding one.
\end{definition}

%In our discussion, we select $T$ to be $\{ 0, 1, 2, \ldots \}$ and assume that each $X_i$ in the set $\{ X_0, X_1, X_2, \ldots \}$ takes values in a finite set $S$ known as the state space. 

%We can then provide a definition of a Markov chain:
%\begin{quote}
%The process $X$ is classified as a Markov chain if it satisfies the Markov property, which is defined as:

%\begin{equation}
%\begin{align}
%& \forall \; n \ge 1 \text{ and } s,x_0,x_1,\ldots,x_{n-1}\in S , \notag\\
%& \Pr[X_n = s | X_{n-1} = x_{n-1}] \notag\\
%&= \Pr[X_n = s | X_0 = x_0, \ldots, X_{n-1} = x_{n-1}]
%\end{align}

%\end{quote}

%To meet the necessary conditions, forecasts for $X_n$—which suggest that we have knowledge of the sampling outcomes for $X_0,X_1,\cdots X_n$ and are investigating the distribution of variable $X_{n+1}$—can be composed solely on the basis of the current state $X_{n}$. These predictions are in accordance with the predictions that could be made if all preceding states were known. Hence, the probability of each event is solely determined by the state achieved in the preceding event. In less formal terms, this concept can be articulated as, "The future is determined solely by the present state of affairs."

\begin{definition}[Homogeneous Chain]
   A Markov chain is a homogeneous chain if it satisfies the condition:
$$\forall n \in T,\ \forall s_i, s_j\in S, \
  \Pr[X_{n+1}=s_j | X_n=s_i]=\Pr[X_1=s_j|X_0=s_i].$$
This means the probability of transitioning from one state to another is consistent across all steps.
\end{definition}
%\textbf{Homogeneity} 
In the rest of this paper, all Markov chains discussed are homogeneous.
%\begin{align}
%    & \forall n \in T\text{ and } \forall i,j\in S, \notag\\
%  &  \Pr[X_{n+1}=j | X_n=i]=\Pr[X_1=j|X_0=i]
%\end{align}
%This means that the probability of transitioning from one state to another is consistent across all steps. In other words, the probability of a state change remains the same throughout the process.

%In the subsequent discussion, all Markov chains are assumed to be homogeneous.

%\subsubsection{Transition Matrix}
%\ 
%As we have discussed in subsubsection-\ref{subsub:The definition of Markov Chain}, any Markov chain we discuss follow the assumption that for Markov chain $ \{ X_t \mid t \in T \}$, $T=\{0,1,2,\ldots\}$, the value of $X_i$ takes values in a finite set $S$, and the Markov chain is homogeneous.

%\textbf{Transition Matrix}

\begin{definition}[Transition Matrix]\label{transitionmatrix}
    The transition matrix $\mathbf{P}=(p_{ij})$ of a homogeneous Markov chain can be defined by the transition probabilities $p_{ij}=\Pr[X_{n+1}=s_j|X_n=s_i]$ to represent the probabilities of changing among all states.
It is easy to prove that for all $i,j \in \left \{1,2,..., |S|\right \} $,  $0 \le p_{ij}\le 1$ and $\sum_{j}p_{ij}=1$.
Naturally, the matrix $\mathbf{P}^{(n)}$ (the $n$-th power of matrix $\mathbf{P}$) where $n \in \left \{1,2,3,... \right \}$ represents the probability of changing from one state to another after exact $n$ steps.
Thus, the element of the matrix $p_{ij}^{(n)}$ is defined as $p_{ij}^{(n)} = \Pr[X_n=s_j|X_0=s_i]$.
\end{definition}

\begin{definition}[State Transition Graph]
       A directed graph $G(V, E)$ can be constructed to represent a finite-state homogeneous Markov chain with a transition matrix $\mathbf{P}$. 
The vertex $i$ in the graph represents a state $s_i$  for all $1\leq i \leq |S|$ and there are $|S|$ nodes. A directed edge from state $s_i$ to $s_j$ exists with weight $p_{ij}$ if $p_{ij} > 0$. Self-edges are allowed.
\end{definition}
\begin{definition}[Recurrent Class]
 
For two vertices $i$ and $j$ (states $s_i$ and $s_j$), if there exists a path from $i$ to $j$ in the constructed graph, we say that state $s_i$ communicates with state $s_j$, denoted as $i\rightarrow j$.
We use the notation $i \leftrightarrow j$ if states $i$ and $j$ can communicate with each other (i.e., $i\rightarrow j$ and $j\rightarrow i$). 
Note that ``$\leftrightarrow$'' signifies an equivalence relationship, which is reflexive, symmetric, and transitive~\cite{StochasticProcess}.  Consequently, we can categorize all the recurrent states into recurrence classes $R_1, R_2,...$. For any two states $i$ and $j$ in the same recurrence class, there exists a finite integer $n \in \{0,1,2,\ldots \}$ such that the probability of transitioning from state $i$ to state $j$ after exactly $n$ steps is non-zero. %\todo{check, introduce transition graph}
\end{definition}

\begin{definition}[Stationary Distribution]
\label{def:stationarydist}

Let  $\pi$ denote a probability distribution of all states of a Markov chain, where $\pi_i$ denotes the probability of being in state $s_i$. 
$\pi$ is a \textit{stationary distribution} if, provided $\Pr[X_0 = s_i] = \pi_i$ for all $i \in \{0, 1, 2, \ldots\}$ at $t=0$, the condition $\Pr[X_t = s_i] = \pi_i$ is satisfied for all time points $t \in \{1, 2, 3, \ldots\}$ and states $i \in \{0, 1, 2, \ldots\}$.
%Consider a probability distribution $\pi$ on all states of a Markov chain, where $\pi_i$ denotes the probability of being in state $i$. A distribution $\pi$ is termed a stationary distribution if, provided $\Pr[X_0 = i] = \pi_i$ for all $i \in \{0, 1, 2, \ldots\}$, the condition $\Pr[X_n = i] = \pi_i$ is satisfied for all time points $n \in \{1, 2, 3, \ldots\}$ and states $i \in \{0, 1, 2, \ldots\}$.
It can be established that the distribution $\pi$ is a stationary distribution if and only if it fulfills the subsequent conditions:
\begin{enumerate}
    \item $\sum_i \pi_i =1$. A stationary distribution should be normalized.
    \item $ \forall i, \pi_i = \sum_{k} p_{ki}\pi_k \ge 0$. A stationary distribution is not changed after one one-step state transition.
\end{enumerate}
%\begin{equation}
%\label{equ:stationary distribution}
%\left\{
  %           \begin{array}{lr}
 %    \forall i, \pi_i = \sum_{k} p_{ki}\pi_k \ge 0, \text{ and } \sum_i \pi_i =1
    % \sum_i \pi_i =1,
    % \forall i,\; \pi_i=\sum_{k} p_{ki}\pi_k 
   %          \end{array}
%\right.
%\end{equation}
\end{definition}
\begin{remark}[Properties of Stationary Distribution]
    Certain properties~\cite{norris1998markov} of stationary distributions, which are instrumental in the proposed compiler optimization, are listed as follows: %\todo{where}
% algorithm to be introduced in this article.
%Subsequently, we will explore certain properties of stationary distributions, which are instrumental in the algorithm to be introduced in this article. We will omit the proofs of these properties within this article. 
\begin{enumerate}
    \item A stationary distribution exists and is unique if there is only a single recurrence class. 

    \item If the chain possesses a unique recurrence class, then for $\pi$, the corresponding unique stationary distribution, the mixing property $\lim_{t \to \infty} p_{ij}^{(t)} = {\pi}_j$ holds for all $i$ and $j$. As the number of steps $t$ approaches infinity, the probability of transitioning from state $i$ to state $j$ after $t$ steps converges to the corresponding entry $\pi_j$ of the stationary distribution $\pi$.

\end{enumerate}

\end{remark}

%\textbf{An Example of the Concepts about the Markov Chain}
%\begin{example}

All the Markov chains used in the rest of this paper are homogeneous. All their state transition graphs are strongly connected, so there is always only one recurrent class for each Markov chain. Each of the Markov chains has one unique stationary distribution.
An example of such a Markov chain is in the following.

\begin{example}

Fig.~\ref{fig:TransitionModel} shows an example of such Markov chains.
On the left is the state transition graph of this chain.
This Markov chain has four possible states, so this graph has four vertices.
On the right is the transition matrix of this Markov chain, and the transition probabilities are also labeled as the weights on the edges of the state transition graph. It can be easily verified that for all vertices $ i,j\in\{1,2,3,4\}$, $i\leftrightarrow j$. 
Thus, there is only one recurrent class, which is $\{1,2,3,4\}$.
The stationary distribution of this Markov chain is unique and satisfies the conditions in Definition~\ref{def:stationarydist}:
\begin{equation}
    \left\{
             \begin{array}{lr}
             \setlength{\arraycolsep}{2pt}
            \begin{pmatrix}
                \pi_0&\pi_1&\pi_2&\pi_3
            \end{pmatrix}
            =
            \begin{pmatrix}
                \pi_0&\pi_1&\pi_2&\pi_3
            \end{pmatrix}\mathbf{P}
            \\
            \pi_0+\pi_1+\pi_2+\pi_3=1
             \end{array}
\right.
\end{equation}
The unique stationary distribution is:
\begin{equation}
    \pi=\begin{pmatrix}
        0.29&0.24&0.29&0.18
    \end{pmatrix}.
\end{equation}

\begin{figure}[h]
    \centering
    \vspace{-5pt}
    \includegraphics[width=0.6\columnwidth]{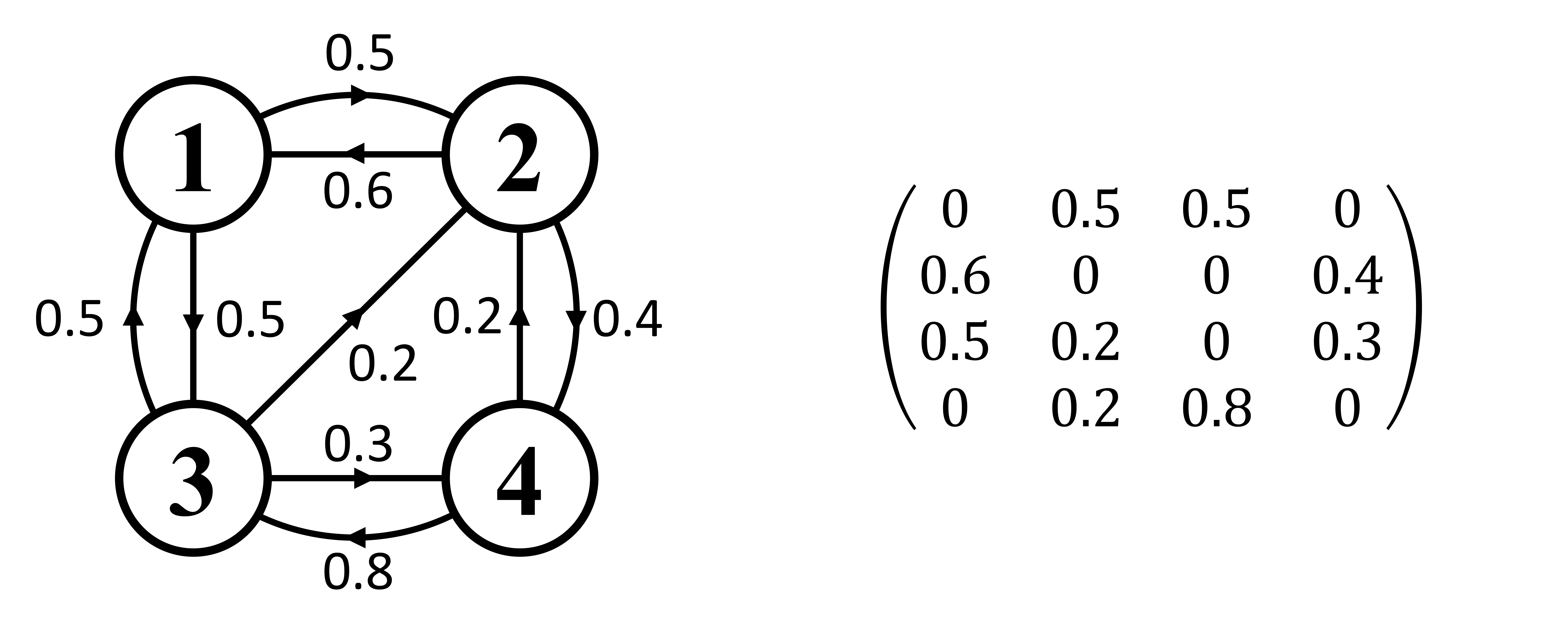}
    \vspace{-10pt}
 %   \text{Left: the graph model of the Markov chain}
  %  \text{Right: the transition matrix of the Markov chain}
    \caption{A Markov chain example. Left: state transition graph. Right: state transition matrix.}
   % \vspace{-5pt}
    \label{fig:TransitionModel}
\end{figure}

\end{example}

% \begin{equation}
%     \mathbf{P}=\begin{pmatrix}
%         0&1&0&0&0&0\\
%         0&0&0.5&0.5&0&0\\
%         0&0.8&0&0&0.2&0\\
%         0&0.3&0&0&0.7&0\\
%         0&0&0.75&0.25&0&0\\
%         0&0&0&0&1&0\\
%     \end{pmatrix}.
% \end{equation}

% \begin{equation}
%     \mathbf{P}=\begin{pmatrix}
%         0&1&0&0&0&0\\
%         0&0&\frac{1}{2}&\frac{1}{2}&0&0\\
%         0&\frac{4}{5}&0&0&\frac{1}{5}&0\\
%         0&\frac{3}{10}&0&0&\frac{7}{10}&0\\
%         0&0&\frac{3}{4}&\frac{1}{4}&0&0\\
%         0&0&0&0&1&0\\
%     \end{pmatrix}.
% \end{equation}

% \begin{figure}
%     \centering
%     \includegraphics[width=0.50\textwidth]{fig/P_1.pdf}
%     \text{The first number of an edge is its capacity}
%     \text{The second number of the edge is its weight}
%         \caption{The graph model gets from the example in subsubsection-\ref{subsub:example}
%     }
%     \label{fig:P_1}
% \end{figure}

% \review{new example}
%\end{example}

%These properties of stationary distributions are vitally important in understanding the long-term behavior and equilibrium of Markov chains.

\subsection{Minimum-Cost Flow Problem}
\label{sec:mincost-flow-problem}

A Minimum-Cost Flow Problem (MCFP)~\cite{enwiki:1184915760} is an optimization and decision problem to send a certain amount of flow through a flow network with the lowest cost. We first define a flow network, and an MCFP is defined upon a flow network.
%transporting a specific amount of flow through a flow network. 
%The Minimum-Cost Flow Problem (MCFP) is an optimization and decision-making challenge tasked with identifying the least expensive method of transporting a specific amount of flow through a flow network. %This formulation is highly practical, and many mathematical problems can be conveniently transformed into an MCFP.
%The MCFP is defined as follows:

\begin{definition}[Flow Network]
    A flow network $(V, E, c, s, t)$ is a directed graph in which each edge possesses a capacity and receives a flow. This network, represented by a directed graph $G = (V, E)$, is assigned a non-negative capacity function $c(\cdot)$ for each edge, ensuring $c(e)>0$ for all $e\in E$. Edges sharing the same source and target nodes are excluded. Two nodes within $G$ are distinctly identified - one as the source $s$ and the other as the sink $t$. The flow $f$ is a function that targets the edges of the graph (commonly extended to directed edges from every pair of different nodes) and returns a real number. The flow network is subject to the following constraints:
\begin{equation}
\label{equ: flow constraints 0}
\left\{
             \begin{array}{lr}
     | f(u,v) | \le c(u,v) \\
     f(u,v)=-f(v,u) \\
     \forall u \notin\{ s,t\} ,\sum_{v\in V}f(u,v) = 0 
    \end{array}
\right.
\end{equation}
It should be noted that $f(u,v)=f(\langle u,v \rangle)$ and $c(u,v)=c(\langle u,v \rangle)$ as is conventionally accepted. 
Consequently, $\sum_{v\in V}f(s,v)=\sum_{v\in V}f(v,t)$ can be verified.
\end{definition}
\begin{definition}[Minimum-Cost Flow Problem]
    In the context of MCFPs, we establish a required amount  $d$ of flow in a given flow network $(V, E, c, s, t)$. Then the flow $f$ needs to satisfy
\begin{equation}
\label{equ: flow constraints 1}
\left\{
             \begin{array}{lr}
     \sum_{v\in V}f(s,v)=d \\
     \sum_{v\in V}f(v,t)=d 
             \end{array}
\right.
\end{equation}

MCFPs also incorporate a cost (weight) function  $w(\cdot)$ for each edge. The cost associated with transporting this flow along an edge $(u,v)$ is calculated as $f(u,v) \cdot w(u,v)$. The problem mandates a flow $f$ of amount $d$ to be transmitted from source $s$ to target $t$. The optimization objective is to minimize the total cost of the flow across all edges:
\begin{equation}
    \min \enspace \sum_{(u,v)\in E} w(u,v) \cdot f(u,v)
    \label{equ: MCFP}
\end{equation}
 subject to all the constraints in the Equations (\ref{equ: flow constraints 0}) and (\ref{equ: flow constraints 1}) above.
\end{definition}

%\color{blue}

\begin{example}

\begin{figure}[h]
    \centering
 %   \vspace{-5pt}
    \includegraphics[width=0.7\columnwidth]{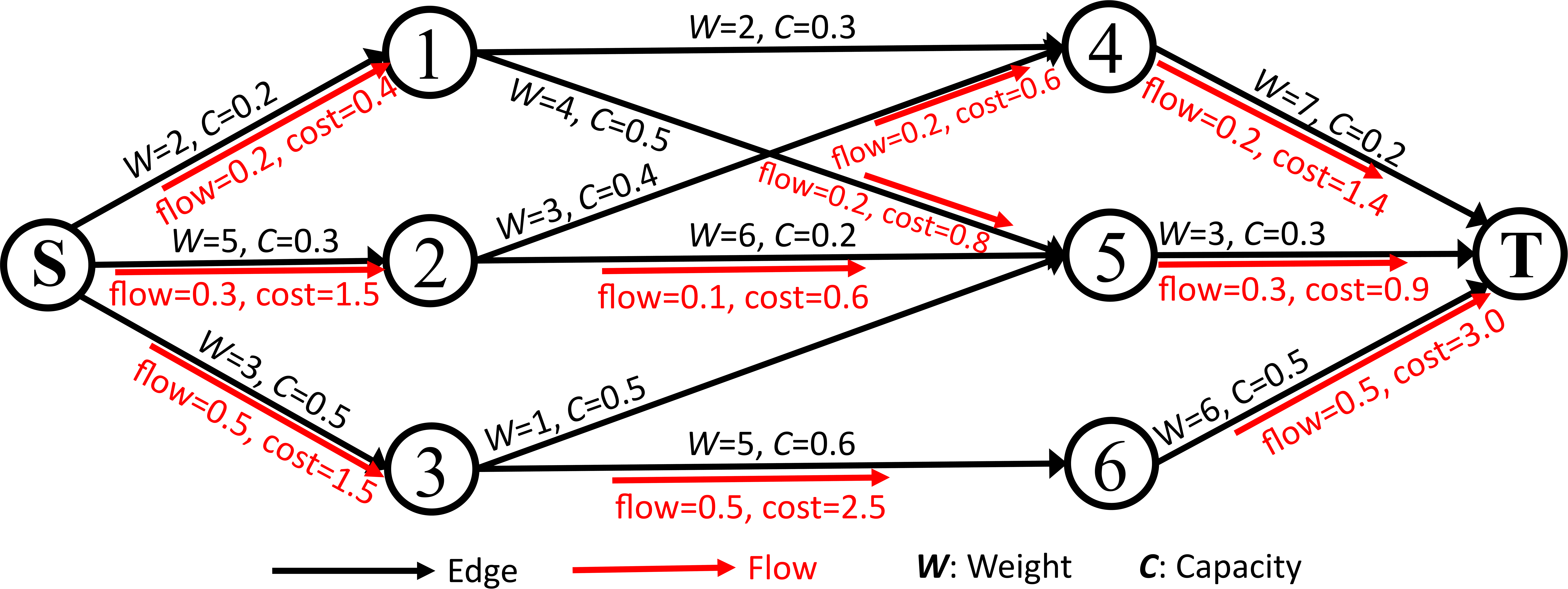}
   % \vspace{-20pt}
 %   \text{Left: the graph model of the Markov chain}
  %  \text{Right: the transition matrix of the Markov chain}
    \caption{An MCFP example. The minimum cost solution with flow amount 1 is obtained and the actual flow on each edge is denoted by the red arrows attached to the edges.}
   % \vspace{-5pt}
    \label{fig:MCFP}
\end{figure}

Fig.~\ref{fig:MCFP} shows an example of a Minimum-Cost Flow Problem. This graph has 8 vertices with one source $s$ and one sink $t$. On each edge, we list the weight and capacity, denoted by $W$ and $C$, respectively. We set the constraint $d=1$ in the Equation (\ref{equ: flow constraints 1}) and solve the minimization problem (\ref{equ: MCFP}). The resulting network flow and cost are also listed on each edge. The red arrows show the flow's direction. The total cost of the flow is $13.2$, which is the minimum cost for Equation (\ref{equ: MCFP}) for this example. It can be easily verified that the constraints in the Equations (\ref{equ: flow constraints 0}) and (\ref{equ: flow constraints 1}) are satisfied.
\end{example}

\section{Previous Works and Motivation}
\label{sec:previouswork}

%\Junyu{begin}

In this section, we introduce the previous works on optimizing the compilation of Hamiltonian simulation. 
We summarize the previous works in two major categories: the deterministic compilation and the randomized compilation. %The first type of approach 
%optimizes the basic gate number by static ordering the Hamiltonian terms and adapting local gate cancellation. 
%The second type takes advantage of randomly sampling the circuit and achieves better asymptotic.  
We briefly discuss the pros and cons of these two types of approaches and the dilemma of combining them, followed by the motivation of this work.

%In the following, the Hamiltonian is considered as 
%\begin{equation}
%\label{equ:ham}
%\mathcal{H} = \sum_{i=1}^m h_i H_i , h_i\in \mathbb{R}
%\end{equation}
%decomposed into a sum of $H_j$ multiplied with weight $h_j$, 
%where each $H_i$ is assumed as a Pauli string. Our target is to compute $e^{i\mathcal{H}t}$ for given $t$, $h_i$ and $H_i$. 
%To make it possible, many works aim to reduce the error of the circuit, which we will discuss as follows. 

\subsection{Deterministic Compilation with Fixed Hamiltonian Term Order}\label{staticorder}
To compile the operator $e^{i\mathcal{H}t}$ into basic gates, a commonly-used approach is the Trotter-Suzuki decomposition~\cite{suzuki1990fractal, Berry2006, Suzuki1991}.
%As introduced in Section~\ref{quantumsimulation},
Given a Hamiltonian with decomposition $\mathcal{H} = \sum_{j=1}^n h_jH_j$, the Trotter product formula~\cite{Whitfield_2011} to approximately implement $e^{i\mathcal{H}t}$ is: \begin{equation}\label{equation:productformula}
    e^{i\mathcal{H}t}=(e^{ih_1H_1\Delta t}e^{ih_2H_{2}\Delta t}\ldots e^{ih_nH_{n}\Delta t})^{\frac{t}{\Delta t}}+O(t\Delta t),
\end{equation}
The Hamiltonian terms after decomposition should be easily implemented with basic gates. Pauli strings are a widely used decomposition basis, and the circuit implementation of a Pauli string simulation has been introduced in Section~\ref{quantumsimulation}.

%uses the same Hamiltonian term order across every step.
Note that this formula does not require any specific order of the Hamiltonian terms within one trotterization step.
The first type of approach is to find a specific order for the Hamiltonian terms $\{e^{ih_1H_1\Delta t}$, $e^{ih_2H_{2}\Delta t}$, $\dots$,  $e^{ih_nH_{n}\Delta t}\}$ with some benefits, and then repeat this fixed order $\frac{t}{\Delta t}$ times to approximate the overall Hamiltonian simulation $e^{i\mathcal{H}t}$. 
In this approach, the ordering of the Hamiltonian terms in the entire circuit is deterministic.

\begin{wrapfigure}{R}{0.5\textwidth}
    \centering
    \vspace{-5pt}
    \includegraphics[width=0.5\columnwidth]{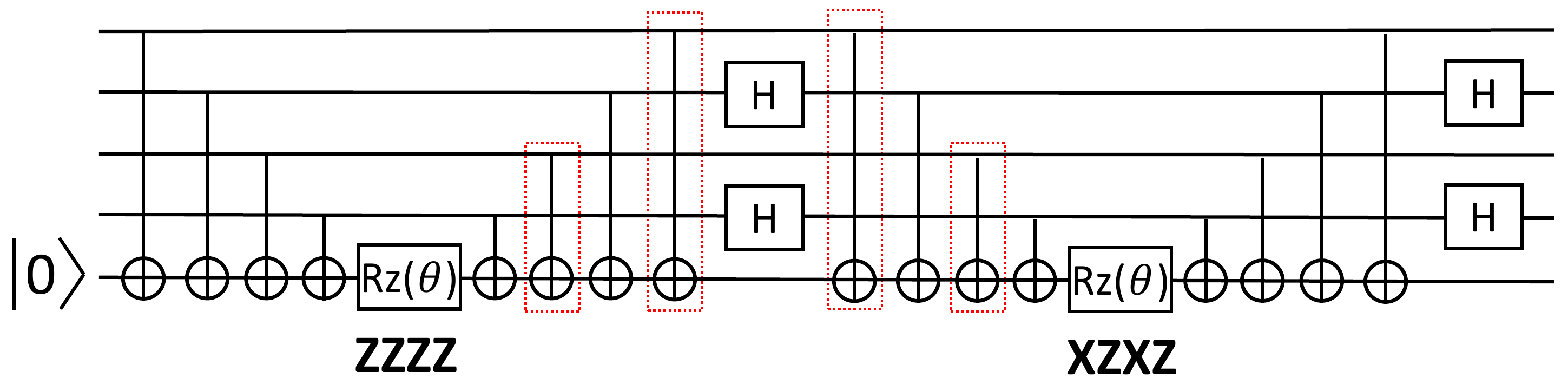}
    \vspace{-20pt}
    \caption{CNOT gate cancellation example}\vspace{-5pt}
    \label{fig:gatecancellationexample}
\end{wrapfigure}

One typical benefit of Hamiltonian term ordering is gate cancellation. 
When the Hamiltonian terms are Pauli strings, and two consecutive terms share the same Pauli operator at some qubits, there is a potential to cancel some gates.
Fig.~\ref{fig:gatecancellationexample} shows an example when applying the gate cancellation technique in~\cite{gui2020term}.
There are two Pauli string simulation circuits for string $ZZZZ$ and $XZXZ$, and they share the same $Z$ operators on two qubits. Adding one ancillary qubit can cancel all the CNOT gates associated with these two qubits (the gates marked in rectangles in Fig.~\ref{fig:gatecancellationexample}).
To maximize the matched operators for more gate cancellation, Hastings et al.~\cite{hastings2015improving}, Gui et al.~\cite{gui2020term}, and the Paulihedral compiler~\cite{li2022paulihedral} explored various lexical-based Hamiltonian term ordering.

Benefits beyond gate cancellation are also possible. Paulihedral compiler~\cite{li2022paulihedral} provides another ordering option for circuit depth reduction. Some works~\cite{cowtan2019phase, cowtan2020generic, de2020architecture, van2020circuit} groups the mutually commutative Pauli strings in the ordering to apply simultaneous diagonalization for circuit simplification.
Gui et al.~\cite{gui2020term} also find that grouping the mutually commutative Pauli strings can reduce the trotterization error and provide better approximation.

%Various benefits of Hamiltonian term ordering have been investigated.
%Hastings et al.~\cite{hastings2015improving} orders the Hamiltonian terms by the lexical ordering of the Pauli strings to achieve gate cancellation for molecule Hamiltonian simulation.
%The Paulihedral compiler~\cite{li2022paulihedral} provides two ordering options optimized for gate cancellation or circuit depth reduction.
%Some works~\cite{cowtan2019phase, cowtan2020generic, de2020architecture, van2020circuit} groups the mutually commutative Pauli strings in the ordering to apply simultaneous diagonalization for circuit simplification.
%Gui et al.~\cite{gui2020term} also find that grouping the mutually commutative Pauli strings can reduce the trotterization error and provide better approximation in practice.
This approach optimizes the Hamiltonian term ordering within each trotterization step and repeats one optimized order across the entire simulation circuit. Their algorithmic approximation errors are still bounded by the original product formula (Equation (\ref{equation:productformula})).

\subsection{Randomized Compilation with Sampling the Hamiltonian Terms}

Recently, it has been observed that randomized compilation can usually yield asymptotically better performance compared with the deterministic approaches.
Different from the deterministic compilation where the order of the Hamiltonian terms is fixed, the randomized compilation approaches generate the quantum circuit via random sampling from the Hamiltonian terms.
Childs et al.~\cite{childs2019faster} find that randomly ordering the Hamiltonian terms within each trotterization step and using different randomly generated orders across the circuit generation can provide lower approximation error.
The SparSto algorithm~\cite{ouyang2020compilation} downsamples the Hamiltonian terms for each trotterization step and then randomly orders the terms within each step.

Especially, the qDrift randomized compiler~\cite{campbell2019random} breaks the framework of trotterization and outperforms the deterministic compilation with an asymptotically better error bound in theory and better numerical results in practice. % compared with deterministic compilation via trotterization. 
%This randomized compilation is called the qDrift and 
Given a Hamiltonian $\mathcal{H} = \sum_{j=1}^n h_jH_j$ and let $\lambda = \sum_{j=1}^n |h_j|$.
The qDrift compiler will sample from an array of operators $e^{i\frac{\lambda t}{N}H_1}$, $e^{i\frac{\lambda t}{N}H_2}$, $\dots$, $e^{i\frac{\lambda t}{N}H_n}$ with a respective probability distribution $ \frac{|h_1|}{\sum_{j}|h_j|}$, $ \frac{|h_2|}{\sum_{j}|h_j|} $, $\dots$, $ \frac{|h_n|}{\sum_{j}|h_j|} $ and $N$ is the number of samples.
The compiled circuit is generated by cascading the randomly sampled operators. 
This circuit can approximate  $e^{i\mathcal{H}t}$ with an error bound $\Delta  \lesssim \frac{ 2 \lambda^2 t^2}{N}$.
This algorithmic error bound can outperform the first-order trotterization (Equation (\ref{equation:productformula})) in all cases and even outperform high-order trotterization most of the time in a common setup for quantum Hamiltonian simulations~\cite{campbell2019random}.
%\textbf{Random Compiler}~\cite{campbell2019random}
%construct its circuit by sampling the Pauli string $H_j$ according to the strength of its weight $|h_j|$ for a given N times. By setting the evolution time equal to $\tau:=t \lambda /N$ for each step with $\lambda = \sum_{i=1}^n |h_i|$, the complete circuit takes the form:
%\begin{equation}
%\label{Vj_unitary}
%V = \prod_{k=1}^N e^{ i \tau H_{i_k} },
%%\text{with probability}
%P =\lambda^{-N} \prod_{k=1}^N h_{i_k}.
%\end{equation}

%It can deliver an asymptotically better error bound in theory and better numerical results in practice against the trotterization-based deterministic compilation.

However, the benefits from deterministically ordering the Hamiltonian terms are lost in the randomized compilation. Because the Hamiltonian terms are randomly cascaded, there is not much control over how we can order the terms.

\subsection{Comparison of the Two Types of Approaches and Our Motivation}

%In summary, 
These two types of approaches have their own pros and cons. 
Deterministic compilation approaches can leverage various benefits from carefully ordering the Hamiltonian terms, but the overall algorithmic approximation accuracy is not as good as that of randomized compilation.
Randomized compilation benefits from better algorithmic approximation accuracy while largely losing the benefits from Hamiltonian term ordering.
Naturally, to further optimize the compilation for quantum simulation, our research question is:
\begin{center}
\fbox{\parbox{3.5in}{\centering
Can we design a compiler that can leverage the benefits of both deterministic compilation and random compilation?}}
\end{center}

Intuitively, these two types of approaches are incompatible because of their different built-in deterministic and random nature. 
Previously, to the best of our knowledge, it is not known how we can naturally combine these types of approaches and reconcile the benefits from both deterministic compilation and randomized compilation.
This paper aims to systematically overcome this challenge and deliver a new compiler for quantum Hamiltonian simulation where the pros of both deterministic and randomized compilation can be leveraged simultaneously.

\section{MarQSim Intermediate Representation and Compilation Algorithm}
\label{sec:Our Technique}

\iffalse
In this section, we present our method, called "Markov simulation", which combines the co-optimization of circuit optimization and the classical oracle function optimization and achieves a trade-off between the physical error and the algorithmic error. The overviews of each subsection are as follows:
\begin{enumerate}
    \item 
Subsection-\ref{sub: The Overall Framework of the Markov Simulation} provides an overview of the algorithm's framework but omits some specific details, particularly the construction of the transition matrix. 
\item 
Subsection-\ref{subsec:tranmat} outlines the steps involved in constructing the transition matrix and offers a theoretical analysis of the algorithm's effectiveness. 
\item 
Subsection-\ref{sub: An Upgrade Method} introduces a new method that builds upon the techniques presented in the previous subsections. 
As these two methods share many similarities, the description of the new method is presented more concisely.

\end{enumerate}
\fi

In this section, we introduce \myCompilerName, a new compilation framework that can reconcile the benefits of deterministic and random compilation.
As discussed above, there is a built-in contradiction between these two types of compilation approaches, and combining them requires non-trivial efforts. % and a new compilation framework.
%the static Hamiltonian term ordering compilation methods and the dynamic random sampling method, a new compilation framework is naturally invoked to reconcile 
\myCompilerNameSpace tackles this challenge by starting from randomized compilation but introducing controllability in the random sampling.
That is, \myCompilerNameSpace formulates the compilation process into sampling from a Markov stochastic process, and the controllability can be implemented by tuning the transition matrix. 
In the rest of this section, we will introduce a new intermediate representation (IR), the Hamiltonian Term Transition Graph (HTT Graph), to naturally formulate and encode the Markov chain for simulating a target Hamiltonian in our \myCompilerNameSpace framework, followed by our compilation algorithm to convert an HTT Graph into a quantum circuit.
We rigorously prove that the compilation algorithm in \myCompilerNameSpace can generate the correct quantum circuit without losing the algorithmic advantage of the randomized compilation for quantum simulation.

\subsection{Hamiltonian Term Transition Graph as Intermediate Representation}

As introduced in Definition~\ref{subsec:markov}, a state transition graph can express a homogeneous Markov chain. 
To formulate the compilation into sampling from such a Markov chain, \myCompilerNameSpace employs a new intermediate representation (IR) to encode the input Hamiltonian into the state transition graph of a Markov chain. 
%We present a new intermediate representation (IR) to encode the input Hamiltonian information and the follow-up optimizations. 
This IR is called the Hamiltonian Term Transition Graph (HTT Graph).
We first introduce how to convert an input Hamiltonian in its corresponding HTT Graph IR.

\begin{definition}[Hamiltonian Term Transition Graph (HTT Graph)]
\label{def:htt-graph}
Given a Hamiltonian $\mathcal{H}=\sum_{j=1}^n h_j H_j$ that is decomposed into the summation of $n$ terms, the Hamiltonian Term Transition Graph $G_{HTT}(V, E)$ is defined as a Markov transition graph, where:
\begin{itemize}
    \item \textbf{Vertices}: The vertex set $V = \{v_1, v_2, \dots, v_n\}$ has $n$ vertices to represent the $n$ Hamiltonian terms after decomposition. There is a one-to-one correspondence between the vertices and the Hamiltonian terms. %correspond to Hamiltonian terms, which indicates 
    A vertex $v_j$ represents the term $H_j$. %There are $n$ vertices.
    \item \textbf{Edges}: A directed edge from vertex $v_i$ to $v_j$ exists with weight $p_{ij}$ when $p_{ij} > 0$.
    %The edge set contains weighted directed edges for arbitrary two vertices in both directions.
    Self-edges are allowed.
   % Especially if an edge has weight 0, 
    \item \textbf{Weights}: %The weight of an edge from vertex $v_i$ to $v_j$ is denoted as $p_{ij}$. 
    We require that all weights $p_{ij}$'s satisfy $0 \leq p_{ij} \leq 1$.
    Especially, $p_{ij}$ is set to be $0$ when there is no edge from vertex $v_i$ to $v_j$.
    Additionally, the summation of the weights of any vertex's outgoing edges should be $1$. That is,
        $
        \sum_{1\leq j \leq n} p_{ij}=1
    $
    
    %from $H_i$ to $H_j$ is attached with a probability $p_{ij}$ representing the probability to sample $H_j$ while currently sampled $H_i$. The normalization of probability is guarded by:

\end{itemize}
\end{definition}

The graph constructed by the procedure above can immediately become a state transition graph of a homogenous Markov chain.
The state space of this Markov chain is the vertex set, or equivalently, the set of all Hamiltonian terms after decomposition due to the one-to-one correspondence. 
The transition matrix of this chain is the weights of the edges  $\mathbf{P}=(p_{ij})$ as the weights are designed to satisfy the requirements of a probability transition matrix (see Definition~\ref{transitionmatrix}). 
The weight $p_{ij}$ on the edge from $v_i$ to $v_j$ represents the transition probability from the state of $H_i$ to the state of $H_j$.

%The probability attached on each edge could be concluded by the transition matrix $\mathbf{P}$, where $\mathbf{P}_{ij}=p_{ij}$. The normalization property is then:
%$$
%    \sum_{j=1}^n\mathbf{P}_{ij}=1
%$$

%\begin{figure}[ht]
\begin{wrapfigure}{R}{0.34\textwidth}
    \centering
    \vspace{-5pt}
    \includegraphics[width=0.34\textwidth]{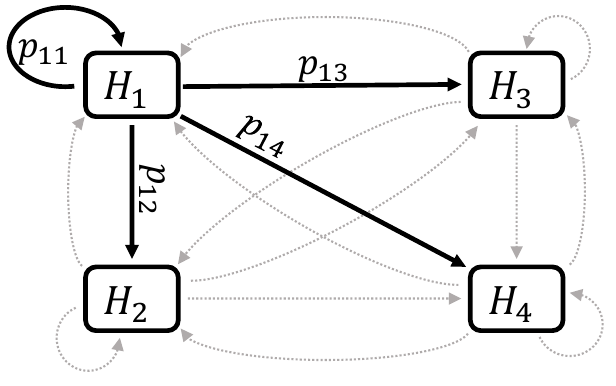}
    \vspace{-20pt}
    \caption{HTT Graph example}
    \vspace{-10pt}
    \label{fig:transition-graph}
    \end{wrapfigure}
%\end{figure}

%\begin{example}

%\yhliu{Figure~\ref{fig:transition-graph} shows an example transition graph of four Hamiltonian terms $H_1$ to $H_4$. Note how each element in the transition matrix is related to the edges in the graph.} 
%\end{example}
Fig.~\ref{fig:transition-graph} shows an example of the HTT Graph for a Hamiltonian with four terms  ($H_1$ to $H_4$) after decomposition. 
This graph has four vertices to represent the four Hamiltonian terms.
%transition graph of four Hamiltonian terms $H_1$ to $H_4$. Note how each element in the transition matrix is related to the edges in the graph.
The state space of the constructed Markov chain is $\{H_1, H_2, H_3, H_4\}$.
The edges have weights to represent the corresponding state transition probability.
For example, the value of $p_{13}$ represents the probability of sampling $H_3$ as the next state when the current state is $H_1$.
The probability distribution should be normalized so we have $p_{11}+p_{12}+p_{13}+p_{14} = 1$. This means when the current state is $H_1$, the summation of the probabilities for sampling all possible next states should always be $1$.

\textbf{How to find the transition probabilities?} Note that the vertices in the HTT Graph IR have clear meanings as they are associated with the Hamiltonian terms. However, the Definition~\ref{transitionmatrix} only requires the weights of the edges to satisfy the constraints of becoming a Markov chain transition matrix. The actual values are not yet determined and will be resolved later. 
In the rest of this section and the next section, we will introduce the sufficient conditions of a desired transition matrix and how to optimize it, respectively. 
%How to optimize this transition matrix 

\begin{algorithm}[t]
    \SetAlgoLined
    \KwIn{\begin{enumerate}
        \item Hamiltonian $\mathcal{H}=\sum_{j}h_jH_j$, where $H_j$ is a Pauli string.% multiplied with $\pm 1$.
        \item Evolution time $t$.
        \item Classical oracle function $Sample(p)$ that gives $i$ with the probability of $p_i$ from the distribution $p$.
        \item Initial distribution $p$ where $p_i=|h_i|/(\sum_{j}|h_j|)$.
        \item Target precision $\epsilon$.
    \end{enumerate}}
    \KwOut{Quantum circuit that simulate $e^{i\mathcal{H}t}$.}
    Construct an HTT Graph with transition matrix $\mathbf{P}$ for the Hamiltonian $\mathcal{H}$\;
    \tcp{How to obtain a transition matrix $\mathbf{P}$ is in Section~\ref{subsec:tranmat}}
    $\lambda \gets \sum_{j}|h_j|$; $N\gets\left\lceil\frac{2{\lambda}^2t^2}{\epsilon}\right\rceil$\; 
    
    $i\gets 0$; \tcp{Last sampled Hamiltonian term index}
    
    $circuits\gets\{\}$; \tcp{Empty quantum circuit}
    \For{$n=1\dots N$}{ % sample n times        
        $i\gets Sample(p)$\;
        
        $p=\mathbf{P}[i]$;  \tcp{$\mathbf{P}[i]$ is the ith row of $\mathbf{P}$}
        Append $e^{i\lambda t H_i/N}$ to $circuits$\;
    }
    \Return $circuits$
    \caption{Compilation As Sampling from Markov Process}
    \label{alg:ARQSC_FT}
\end{algorithm}

\subsection{Compiling the Hamiltonian Term Transition Graph to Quantum Circuit}\label{sec:compilationalgorithm}

By converting an input Hamiltonian into an HTT Graph, we can immediately obtain a Markov chain represented by this HTT Graph IR. % down to the quantum circuit.
The next step is to compile the HTT Graph to the quantum circuit.
The compilation is implemented by sampling from the Markov chain of an HTT Graph.
%The compilation of this HTT Graph IR can be considered as sampling from the Markov chain it represents.

The compilation algorithm in \myCompilerNameSpace is shown in Algorithm~\ref{alg:ARQSC_FT}.
This algorithm has several inputs, including the Hamiltonian $\mathcal{H}=\sum_{j}h_jH_j$ to be simulated, the desired evolution time $t$, a $Sample(p)$ function that can be sampled from a distribution $p$, an initial probability distribution obtained from the Hamiltonian decomposition, and the targeted approximation error $\epsilon$.
The output is a quantum circuit that approximates $e^{i\mathcal{H}t}$ with guaranteed approximation error below $\epsilon$.
%The first input is 
%Provided with the HTT graph of a given Hamiltonian $\mathcal{H}=\sum_{j}h_jH_j$, Algorithm~\ref{alg:ARQSC_FT} illustrates the steps to convert it into the desired simulation quantum circuit.

The first step of the compilation algorithm is to generate an HTT Graph for the Hamiltonian $\mathcal{H}$ (by Definition~\ref{def:htt-graph}). 
How to obtain the transition matrix $\mathbf{P}$ will be introduced later, and we assume that we already have the matrix $\mathbf{P}$ for now.
We then need to prepare a parameter $\lambda$, which is the sum of the absolute values of the coefficients of all the Hamiltonian terms in the decomposition.
We also need the number of sampling steps $N=\left\lceil2{\lambda}^2t^2/\epsilon\right\rceil$.
Finally, we will sample $N$ steps from the Markov chain associated with the HTT graph.
In the first sampling step, we sample from the set of all Hamiltonian terms $\{H_1, H_2, \dots, H_n\}$ based on the initial probability distribution $p$ where $p_i=|h_i|/(\sum_j|h_j|)$.
After the first step, the remaining $N-1$ sampling steps are from a Markov chain where the current sampled term determines the probability distribution of sampling the next term. 
When the current sampled term is $H_i$, the probability distribution of the sampling the next term is the $i$-th row of the transition matrix $\mathbf{P}$, $[p_{i1}, p_{i2}, \dots, p_{in}]$. And $p_{ij}$ is the probability of sampling $H_j$ as the next term when the current term is $H_i$. This follows that \myCompilerNameSpace is formulating the compilation into sampling from a Markov chain.
After a term $H_i$ is sampled in one step, we will append the operator $e^{i\lambda t H_i/N}$ to the final generated circuit.
How to convert the $e^{i\lambda t H_i/N}$ into basic single- and two-qubit gates has been introduced earlier in Section~\ref{quantumsimulation}.

%In each sampling step of totally $N=\left\lceil2{\lambda}^2t^2/\epsilon\right\rceil$ steps, the next sampled Hamiltonian term $H_j$ is selected based on a probability distribution $p$, which could be obtained via two approaches: 1) the initial distribution $p$ where $p_i=|h_i|/(\sum_j|h_j|)$, 2) the following distribution is obtained from the transition matrix $\mathbf{P}$ based on previous selected term. 

\textbf{Correctness of Algorithm~\ref{alg:ARQSC_FT}}: The immediate and most critical question regarding our compilation Algorithm~\ref{alg:ARQSC_FT} is that whether the circuit generated by Algorithm~\ref{alg:ARQSC_FT} can correctly approximate the quantum Hamiltonian simulation $e^{i\mathcal{H}t}$. If so, what is the error bound of the circuits sampled by Algorithm~\ref{alg:ARQSC_FT}?
The following theorem answers these questions and provides sufficient conditions that can guarantee the correctness and bounded approximation error in the final circuit generated from Algorithm\ref{alg:ARQSC_FT}.
%Hence, the critical part of the algorithm is to construct a dedicated transition matrix $\mathbf{P}$ from the given Hamiltonian $\mathcal{H}$ to allow best properties of the sampling result. By imposing certain constraints on the constructed matrix $\mathbf{P}$, we are able to prove the correctness of our algorithm, which the final generated circuit do approximate the Hamiltonian with the provided error bound $\epsilon$.

\begin{theorem}[Correctness and Approximation Error Bound]\label{theorem:correctness}
Given a Hamiltonian $\mathcal{H}=\sum_{j=1}^n h_j H_j$, the quantum circuit compiled by Algorithm~\ref{alg:ARQSC_FT} correctly approximates the operator $e^{i\mathcal{H}t}$ if the HTT graph and the corresponding transition matrix $\ \mathbf{P}$ satisfies the following two conditions:
\begin{enumerate}
    \item \textbf{Strong Connectivity}: The HTT graph is a strongly connected state transition graph, meaning that only one unique recurrence class exists, and it contains all possible states. %\todo{check}
    \item \textbf{Stationary Distribution Preservation}: The distribution $\pi_i=|h_i|/(\sum_{j}|h_j|)$ is stationary under the transition matrix $\mathbf{P}$ (this is the initial distribution in our algorithm):
    $$
    \pi = \pi\mathbf{P} \mathrm{\quad where\quad} \pi =  \begin{pmatrix}
    \frac{|h_1|}{\sum_{j}|h_j|} & \frac{|h_2|}{\sum_{j}|h_j|} & \cdots & \frac{|h_n|}{\sum_{j}|h_j|} 
    \end{pmatrix}
    $$
  %  This make the 
\end{enumerate}

The approximation error $\epsilon$ is bounded by $$\epsilon\lesssim\frac{2\lambda^2 t^2}{N}$$ where $\lambda$ is the sum of the absolute values of the Hamiltonian term coefficients, $t$ is the simulated evolution time, and $N$ is the number of sampling steps. (defined in Algorithm~ \ref{alg:ARQSC_FT}, line 2). 
\end{theorem}

\begin{proof}
Postponed to Appendix \ref{proof_lem: MarQSim}.
\end{proof}

Theorem~\ref{theorem:correctness} guarantees that the quantum circuit generated by the \myCompilerNameSpace compilation algorithm correctly approximates the desired quantum simulation process and the error bound same as qDrift~\cite{campbell2019random} once we find a transition matrix $\mathbf{P}$ that can satisfy the sufficient conditions. Actually, the vanilla qDrift can be implemented in \myCompilerNameSpace with a special transition matrix.

\begin{corollary}[qDrift is a Special Case in \myCompilerName]\label{corollary:qdrift}
The original qDrift algorithm~\cite{campbell2019random} is a special case in our Algorithm~\ref{alg:ARQSC_FT}, where the transition matrix $\mathbf{P}$ is constructed directly by the desired stationary distribution. For a given Hamiltonian $\mathcal{H}=\sum_{j=1}^n h_j H_j$, qDrift can be implemented in \myCompilerNameSpace with the following Markov chain transition matrix:
$$
    \mathbf{P} = \begin{pmatrix}
        \pi_1 & \pi_2 & \cdots & \pi_n \\
        \pi_1 & \pi_2 & \cdots & \pi_n \\
        \cdots & \cdots & \cdots & \cdots \\
        \pi_1 & \pi_2 & \cdots & \pi_n 
    \end{pmatrix}=
    \begin{pmatrix}
        \frac{|h_1|}{\sum_{j}|h_j|} & \frac{|h_2|}{\sum_{j}|h_j|} & \cdots & \frac{|h_n|}{\sum_{j}|h_j|} \\
        \frac{|h_1|}{\sum_{j}|h_j|} & \frac{|h_2|}{\sum_{j}|h_j|} & \cdots & \frac{|h_n|}{\sum_{j}|h_j|}\\
           \cdots & \cdots & \cdots & \cdots \\
        \frac{|h_1|}{\sum_{j}|h_j|} & \frac{|h_2|}{\sum_{j}|h_j|} & \cdots & \frac{|h_n|}{\sum_{j}|h_j|} \\
    \end{pmatrix}
$$
This transition matrix indicates that the probability to sample the next term does not depend on the previously sampled term and always selects from the stationary distribution $\pi$. It can be easily verified that this matrix satisfies the two conditions in Theorem~\ref{theorem:correctness}. The fact that all $|h_j|$'s are strictly larger than $0$ implies that the transition graph is a complete graph. A complete directed graph is strongly connected. $\pi = \pi\mathbf{P}$ can be checked directly.
\end{corollary}

We denote the transition matrix presented by qDrift as $\mathbf{P_{qd}}$. A simple numerical example of how $\mathbf{P_{qd}}$ is constructed is provided in the following.

\begin{example}\label{eg:example-hamiltonian}
Consider $\mathcal{H}=1.0\cdot IIIZ+ 0.5\cdot IIZZ+0.4\cdot XXYY + 0.1\cdot ZXZY$. In this example, $(h_1,h_2,h_3,h_4) = (1.0,0.5,0.4,0.1)$. So we have the stationary distribution $(\pi_1,\pi_2,\pi_3,\pi_4)=(0.5,0.25,0.2,0.05)$ and the transition matrix that implements the qDrift algorithm in our \myCompilerNameSpace framework is:
$$
    \mathbf{P_{qd}}=
    \begin{pmatrix}
        0.5 & 0.25 & 0.2 & 0.05 \\
        0.5 & 0.25 & 0.2 & 0.05 \\
        0.5 & 0.25 & 0.2 & 0.05 \\
        0.5 & 0.25 & 0.2 & 0.05 \\
    \end{pmatrix}
$$
\end{example}

\section{Compiler Optimization via Min-cost Flow}
\label{subsec:tranmat}

In the previous section, we introduced the HTT Graph IR and the circuit compilation algorithm in the \myCompilerNameSpace framework, while the question of finding a transition matrix $\mathbf{P}$ remains unsolved.
We have demonstrated how to implement the baseline qDrift randomized compilation with a special transition matrix in \myCompilerName.
However, the benefits of static Hamiltonian term ordering, such as the gate cancellation, are still not incorporated in the qDrift transition matrix. 
Fortunately, the \myCompilerNameSpace framework provides great flexibility by allowing us to tune the transition matrix for different purposes. 
Note that once the transition matrix satisfies the two conditions in Theorem~\ref{theorem:correctness}, the overall simulation algorithmic efficiency and error bound are guaranteed. There exist many more matrices beyond the baseline qDrift transition matrix that can satisfy these two conditions. 
In this section, we introduce how \myCompilerNameSpace can find and optimize the transition matrix by formulating it in an optimization problem, the \textit{Min-Cost Flow Problem~(MCFP)}.
By carefully encoding the cost and capacity functions of the edges in a flow network, \myCompilerNameSpace can generate transition matrices towards different Hamiltonian term ordering tendencies, and we demonstrate how to optimize for more CNOT gate cancellation by constructing a new transition matrix. %, the $\mathbf{P}_{gc}$ (`gc' denotes gate cancellation).

\subsection{Min-Cost Flow Problem Formulation with Correctness Guarantee}

%\yhliu{$\downarrow$ And the following section $\downarrow$}

We select the flow network to model our Markov chain sampling because the probability transition can be considered as probability flows that go through the vertices of the state transition graph.
To model the concept of transition in the HTT graph, the overall structure of the flow network has a bipartite graph structure.
A bipartite graph has two sets of vertices. Edges only exist between two vertices when they are in different sets.
This can naturally represent the sampling transition from one Hamiltonian term to another Hamiltonian term.

\subsubsection{Flow Network Structure from HTT Graph}
In our flow network, we have two vertices set $Prev$ and $Next$. %$$Prev=\{H^{prev}_i\}$ and $Next=\{H^{next}_i\}$.
Suppose there are $n$ Hamiltonian terms in the HTT graph after decomposition.
Then each of the $Prev$ and $Next$ sets will have $n$ vertices to represent the $n$ Hamiltonian terms, denoted as $Prev=\{H^{prev}_i\}$ and $Next=\{H^{next}_i\}$. 
$Prev=\{H^{prev}_i\}$ represents the previous sampled Hamiltonian terms. 
$Next=\{H^{next}_i\}$ represents the next Hamiltonian terms to be sampled.
Then an edge from $H^{prev}_i$ to $H^{next}_j$ can represent sampling $H_i$ in the previous step and then sampling $H_j$ in the next step.
We only allow edges from $Prev$ to $Next$.
A flow network also has a source vertex $S$ and a sink vertex $T$.
We have edges from the source vertex $S$ to each vertex in  $\{H^{prev}_i\}$ and from each vertex in $\{H^{next}_i\}$ to the sink node $T$.
Overall, the flow goes through  $S\rightarrow Prev \rightarrow Next \rightarrow T$. The graph of the constructed flow network is defined as:
\begin{equation}
    \left\{\begin{array}{ll}
        V=&Prev\cup Next\cup\{S,T\} \\
        E=&\{\edge{H^{prev}_i,H^{next}_j}|H^{prev}_i\in Prev\land H^{next}_j\in Next\}\cup \\
        &\{\edge{S,H^{prev}_i}|H^{prev}_i\in Prev\}\cup \\
        &\{\edge{H^{next}_j,T}|H^{next}_j\in Next\}
    \end{array}\right.
\end{equation}

Fig.~\ref{fig:flowgraph-overview} shows an example of generating the structures of the flow network (on the right) based on an HTT graph (on the left). Apart from the source node and sink node, the subgraph $(\{H^{prev}\}\cup\{H^{next}\},\edge{H^{prev},H^{next}})$ is a direct shifting of the HTT graph, where a node $H_i$ is split into $H^{prev}_i$ and $H^{next}_j$ and edge $\edge{H_i,H_j}$ is correspondent to $\edge{H^{prev}_i,H^{next}_j}$.

\iffalse
We introduce the \textit{flow model} for stationary distribution to explore the gate cancellation chances while preserving the stationary distribution. As we previously discussed, the new transition matrix $\mathbf{P}_{gc}$ has to maintain the same stationary distribution $\pi$ such that $\pi=\pi\mathbf{P}_{gc}$ for $\pi =  (
    \frac{|h_1|}{\sum_{i}|h_i|} \ \frac{|h_2|}{\sum_{i}|h_i|} \ \cdots \ \frac{|h_n|}{\sum_{i}|h_i|} 
    )$ to guarantee the simulation correctness and accuracy of the simulation while also satisfying the probability normalization $\sum_{j=1}^n\mathbf{P}_{ij}=1$.
\fi

%The transition matrix could be defined from a graph with \textit{flow} representing the probability attached to each edge. Based on this idea, we could consider a complete, directed bipartite graph where vertices are Hamiltonian terms. To accord with the concept of \textit{transition}, two vertices set $Prev=\{H^{prev}_i\}$ and $Next=\{H^{next}_i\}$ represents the previous sampled and the next sampling Hamiltonian term. The $\mathbf{P}_{gc}$ correspondent transition graph $G_{\mathbf{P}_1}$ is defined as:

\begin{figure}[t]
    \centering
    \includegraphics[width=\textwidth]{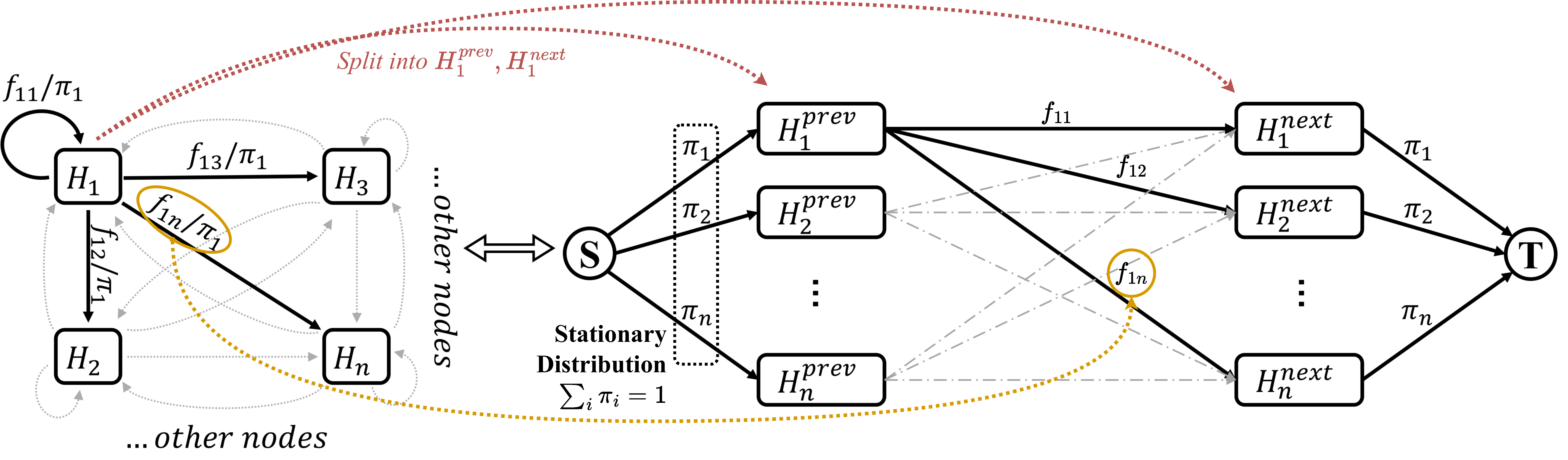}
%    \vspace{-20pt}
    \caption{Constructing the flow network from the HTT Graph. The vertex $H_i$ in the HTT graph is converted and split into $H_i^{prev}$ and $H_i^{next}$ in the flow graph. The transition probability $p_{ij}$ in the HTT graph corresponds to the flow $f_{ij}$ from $H_i^{prev}$ to $H_j^{next}$ where $p_{ij} = f_{ij} / \pi_i$.}
    \label{fig:flowgraph-overview}
\end{figure}

\subsubsection{From Flow to Probability Transition Matrix}
After constructing the structure of the flow network, we then introduce how to generate a probability transition matrix from the flow on the edges. A flow function $f(\edge{u,v})$ is attached to each edge.
Consider a vertex $H_i^{prev}$ in the $Prev$ set. 
The only incoming flow of this vertex is the $ f(\edge{S, H^{prev}_i})$ from $S$. Its outgoing flows contain $f(\edge{H^{prev}_i,H^{next}_j})$ for all $H^{next}_j$ vertices in the $Next$ set.
A flow network requires that the incoming flow should be equal to the outgoing flow for all vertices except the source and sink (see the third equation in Equation (\ref{equ: flow constraints 0})). So we have 
\begin{equation}
    f(\edge{S, H^{prev}_i}) = \sum_jf(\edge{H^{prev}_i,H^{next}_j})
\end{equation}
Similarly, we have the following for the vertex $H_i^{next}$ in the $Next$ set.:
\begin{equation}
    f(\edge{H^{next}_i,T}) = \sum_jf(\edge{H^{prev}_j,H^{next}_i})
\end{equation}
When the flow network graph is constructed, the edge $\edge{H^{prev}_i, H^{next}_j}$ already represents sampling $H_i$ in the previous step and then sampling $H_j$ next and the flow in the network models the probability flow.
We can use the value of $$p_{ij} = \frac{f(\edge{H^{prev}_i,H^{next}_j})}{ f(\edge{S, H^{prev}_i})}$$ to represent the probability of transiting from $H_i$ to $H_j$ and $\sum_jp_{ij} = 1$ is immediately satisfied.

From the procedure above, a probability transition matrix $\mathbf{P}=(p_{ij})$ for the HTT graph can be extracted from the flow network. We still need the transition matrix $\mathbf{P}$ to satisfy the conditions in Theorem~\ref{theorem:correctness} to guarantee the correctness and error bound.
The following theorem provides a sufficient condition that can make the transition matrix $\mathbf{P}$ satisfy the second required condition in Theorem~\ref{theorem:correctness}, the \textbf{Stationary Distribution Preservation} condition. The first required condition (strongly connected graph) will be discussed later in Section~\ref{combination}.

\begin{theorem}[Preserving Stationary Distribution in the Flow Network]
\label{theorem:flow-constraint}
The transition matrix $\mathbf{P}$ extracted from a flow network mentioned above satisfies the second condition in Theorem~\ref{theorem:correctness} once the following equations hold:
\begin{equation}
f(\edge{S,H^{prev}_i})=f(\edge{H^{next}_i,T})=\pi_i 
%\left\{\begin{array}{lr}
%     f(\edge{S,H^{prev}_i})=\pi_i \\
%     f(\edge{H^{next}_j,T})=\pi_j 
%\end{array}\right.
\label{TH5.1eq}
\end{equation}
for all $1\leq i\leq n$.
\end{theorem}

\begin{proof} 
The general proof is in Appendix~\ref{proof_lem:tranflow}. 

\iffalse
\begin{align}    
         (\mathbf{\pi} \mathbf{P})(i) &=\sum_{j} \mathbf{\pi_j} \mathbf{P}(j,i) \notag\\ 
        & =\sum_{j} \mathbf{\pi_j} \frac{f(\edge{H^{prev}_j,H^{next}_i})}{ f(\edge{S, H^{prev}_j})}  \notag\\ 
        & = \sum_{j} f(\edge{H^{prev}_j,H^{next}_i}) \notag\\ 
        & = f(\edge{H^{next}_i,T}) \notag\\ 
        & = \mathbf{\pi_i} \notag
\end{align}
\fi

\end{proof}
%A flow function $f(\edge{u,v})$ is attached to each edge. The flow involves the starting node $S$ and termination node $T$ is fixed as:
%$$
%\left\{\begin{array}{lr}
%%     f(\edge{S,H^{prev}_i})=\pi_i \\
%     f(\edge{H^{next}_j,T})=\pi_j 
%\end{array}\right.
%$$

%The flow function satisfies the conservation principle:
%\begin{principle}[Flow Conservation]
%For each vertex (except $S$ and $T$), the sum of incoming flow equals the sum of outgoing flow and is a fixed number. Formally:
%\begin{equation}
%\label{eqn:flow-constraint}
%    \sum_j f(\edge{i,j})=\sum_j f(\edge{j,i})=\pi_i
%\end{equation}
%\end{principle}

%Figure~\ref{fig:flowgraph-overview} shows a basic overview of the graph $G_{{\mathbf{P}_1}}$. Apart from the starting and termination node, the subgraph $(\{H^{prev}\}\cup\{H^{next}\},\edge{H^{prev},H^{next}})$ is a direct shifting of the HTT graph, where a node $H_i$ is split into $H^{prev}_i$ and $H^{next}_j$, edge $\edge{H_i,H_j}$ is correspondent to $\edge{H^{prev}_i,H^{next}_j}$.

%Under such a sense, the flow we defined characterizes the transition probability to sample term $H^{next}_j$ when the last sampled term is $H^{prev}_i$ equals $f(\edge{H^{prev}_i,H^{next}_j})/\pi_i$. The transition matrix $\mathbf{P}_1$ is defined as:
%\begin{equation}
%    (\mathbf{P}_1)_{ij}=f(\edge{H^{prev}_i,H^{next}_j})/\pi_i
%\end{equation}

%In such a sense, there should be a strong connection between graph $G_{\mathbf{P}_1}$ and the HTT graph. 

\subsubsection{Encoding the Min-Cost Flow Problem}

After constructing the flow network structure and discussing the desired probabilities of the flow, we finalize the MCFP formulation (see Section~\ref{sec:mincost-flow-problem}) by attaching the capacity $c$ and cost $w$ function onto each edge and specifying the overall required amount of flow.

\textbf{Hard Constraints} The hard constraints that must be satisfied to ensure a correct transition matrix have been introduced in the Theorem~\ref{theorem:flow-constraint} above and can be naturally encoded in the capacity on the corresponding edges:
\begin{equation}
     c(\edge{S, H^{prev}_i})=c(\edge{H^{next}_i,T})=\pi_i 
\end{equation}   
where $\{\pi_i\}$ is the desired stationary distribution, forcing the amount of required flow to be $\bf 1$. Since we have $\sum_i\pi_i = 1$, these constraints will ensure the capacity of these edges is fully utilized and then forcing the flows satisfy the requirements in Theorem~\ref{theorem:flow-constraint}.

\textbf{Soft Transition Matrix Tuning} The transition matrix can be tuned by changing the cost function $w$ on the edges from $\{H^{prev}_i\}$ to $\{H^{next}_i\}$. For example, if we hope to increase the probability of sampling $H_j$ after sampling $H_i$, we simply just reduce the cost on the corresponding edge $w(\edge{H^{prev}_i, H^{next}_j})$ in the flow network.
Then the MCFP framework will tend to increase the flow $f(\edge{H^{prev}_i, H^{next}_j})$ through this edge if possible, naturally increasing the entry $p_{ij}$ in the transition matrix and the probability of sampling $H_j$ after $H_i$.

\textbf{Others} The cost and capacity on other edges can also be specified. In this paper, they are not used and we have the following default setup:
\begin{gather*} 
w(\edge{S, H^{prev}_i})=w(\edge{H^{next}_j,T})=0 \\
     c(\edge{H^{prev}_i,H^{next}_j})=\infty
\end{gather*}
This represents the cost function is not effective on the edges from $S$ to $\{H^{prev}_i\}$ and those from $\{H^{next}_i\}$ to $T$. And we do not place capacity constraints on the edges from $\{H^{prev}_i\}$ to $\{H^{next}_i\}$.

\subsection{Optimizing for CNOT Gate Cancellation} \label{sec:opt_for_CNOT}

Now we illustrate how we can optimize the transition matrix for gate cancellation with the MCFP framework. % \Junyu{The resulted transition matrix will be denoted as $\mathbf{P_{gc}}$.}
To encode the gate cancellation optimization, we can let the weight function of the edge $\edge{H^{prev}_i,H^{next}_j}$ represent the number of gates after applying the CNOT gate cancellation between the Hamiltonian terms:
\begin{gather*}
    w(\edge{H^{prev}_i,H^{next}_j})=CNOT\_count(H^{prev}_i,H^{next}_j)
\end{gather*}
The weight function $CNOT\_count(H^{prev}_i,H^{next}_j)$ calculates the number of CNOT gates between $R(z)$ gate of the compiled term $e^{i\lambda tH^{prev}_i/N}$ and $e^{i\lambda tH^{next}_j/N}$. 
In this paper, the Hamiltonian terms are Pauli strings and the gate cancellation between Pauli string simulation circuits has been widely studied in previous static Hamiltonian term ordering works (introduced in Section~\ref{staticorder}, Fig.~\ref{fig:gatecancellationexample}). Here, we adopt the gate cancellation method in~\cite{gui2020term}, and the weight function $CNOT\_count(H^{prev}_i, H^{next}_j)$ can be easily obtained based on the number of matched Pauli operators in the two Pauli strings.
%A circuit optimizer gives this number after performing gate cancellations. 
Given the weight function, the flow function $f$ could be solved via an MCFP in the following:
\begin{equation}
\label{eqn:min-cost-flow-problem}
    \min_f \mathcal{W}=\min_f \sum_{\edge{i,j}\in E}f(\edge{i,j})w(\edge{i,j})
\end{equation}%}

%However, the flow does not capture the gate cancellation opportunities we are discussing. In the following sections, we will construct a min-cost flow model based on this graph with \textit{weight} and \textit{capacity} attached to each edge that not only solves \textit{flow} that complies with Equation~\eqref{eqn:flow-constraint} to make the transition matrix $\mathbf{P}$  preserves the stationary distribution, but also reflects the expectation of the number of CNOT gates along with the \textit{flow}, which will provide a more thorough representation of the desired circuit optimization.

%\yhliu{%Now, each edge is not only attached with a \textit{flow} $f$, but also a \textit{weight} $w$ defined as:
%\begin{gather*}
%    w(\edge{S, H^{prev}_i})=w(\edge{H^{next}_j,T})=0 \\
%    w(\edge{H^{prev}_i,H^{next}_j})=CNOT\_number(H^{prev}_i,H^{next}_j)
%\end{gather*}

%\yhliu{Problem~\eqref{eqn:min-cost-flow-problem} is similar to the min-cost problem in graph theory described in Section~\ref{sec:mincost-flow-problem}. However, in actual graph min-cost flow problems, there should be a) a required input flow, which is always \textbf{1} in our problem; b) a capacity $C(\edge{u,v})$ binds to each edge, where the flow has to satisfy $f(\edge{u,v})\leq C_{\edge{u,v}}$. It is easy to attach a pseudo capacity in our problem as the flow is specially defined, where:
%\begin{gather*}
%    C(\edge{S, H^{prev}_i})=C(\edge{H^{next}_i,T})=\pi_i \\
%    C(\edge{H^{prev}_i,H^{next}_j})=\infty
%\end{gather*}}

\begin{figure}[h]%{R}{0.6\textwidth}
    \centering
     \vspace{-5pt}
    \includegraphics[width=0.7\textwidth]{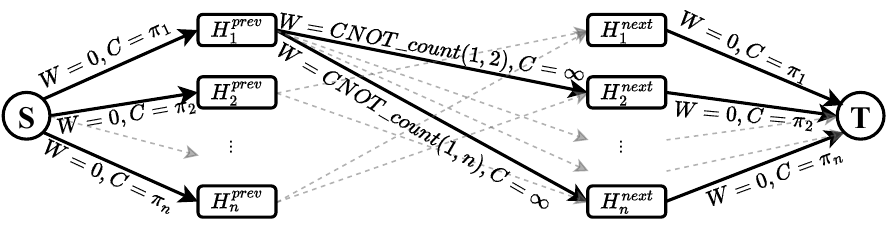}
    \caption{Overview of the Min-Cost Flow Problem for CNOT gate cancellation}
    \vspace{-5pt}
    \label{fig:flowgraph-problem}
\end{figure}

%\yhliu{
Fig.~\ref{fig:flowgraph-problem} shows an overview of the MCFP to optimize for CNOT gate cancellation. Each edge is annotated with a tuple of $(Weight, Capacity)$. 
Note that $CNOT\_count(H^{prev}_i, H^{next}_i)\equiv0$ because the two same terms can be directly combined by duplicating the time parameter without adding gates. To prevent the min-cost flow yields a trivial solution where $\mathbf{P_{gc}}=\mathbb{I}$ (the transition matrix is the identity where all flow choose path $H^{prev}_i\rightarrow H^{next}_i$ since its cost is $0$), we remove the edges of $\edge{H^{prev}_i, H^{next}_i}$ for all $1\leq i \leq n$.
An MCFP can be solved efficiently with various existing solvers~\cite{kiraly2012efficient}.
%has been already widely studied effectively (see Section~\ref{sec:mincost-flow-problem})
The overall procedure of obtaining a new transition matrix via MCFP is summarized in Algorithm~\ref{alg:P_1}. %\Junyu{what is P1 in this algorithm}

\begin{example}\label{exampleP1}
\begin{figure}[ht]
    \centering
    \includegraphics[width=0.7\textwidth]{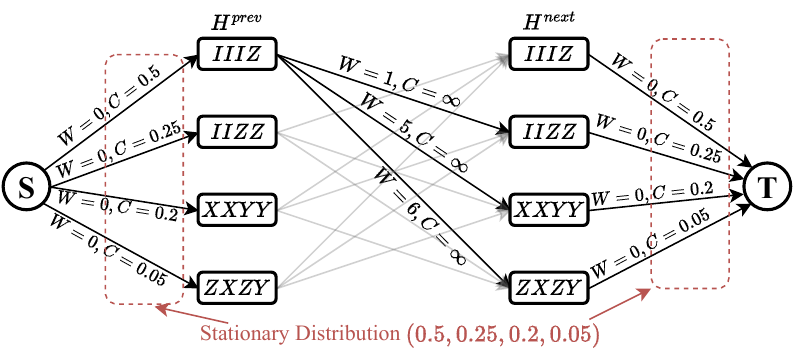}
    \caption{Min-Cost Flow Problem to CNOT gate cancellation in Example~\ref{eg:example-hamiltonian}}
    \label{fig:example-mincost-graph}
\end{figure}

%\yhliu{
Consider the Hamiltonian defined in Example~\ref{eg:example-hamiltonian}. %The cost $gate\_count$  can be obtained from work in Section~\ref{subsec:cnot_cancellation_trick} (?????). 
%The corresponding flow network is thus constructed as Fig.~\ref{fig:example-mincost-graph}.
Fig.~\ref{fig:example-mincost-graph} shows a concrete example of the MCFP formulated for optimizing the Hamiltonian in Example~\ref{eg:example-hamiltonian} (in Section~\ref{sec:compilationalgorithm}). 
The actual cost and capacity are labeled on each edge.
%}.
Solving the min-cost flow problem gives $f(\cdot)$ (for short, $f_{ij}=f(\edge{H^{prev}_i,H^{next}_j})$:
\begin{equation}
    \begin{pmatrix}
        f_{11}&f_{12}&f_{13}&f_{14}\\
        f_{21}&f_{22}&f_{23}&f_{24}\\
        f_{31}&f_{32}&f_{33}&f_{34}\\
        f_{41}&f_{42}&f_{43}&f_{44}\\
    \end{pmatrix}
    =
    \begin{pmatrix}
        0.0&0.25&0.2&0.05\\
        0.25&0.0&0.0&0.0\\
        0.2&0.0&0.0&0.0\\
        0.05&0.0&0.0&0.0\\
    \end{pmatrix}
\end{equation}

The transition matrix $\mathbf{P_{gc}}$ given by Algorithm~\ref{alg:P_1}, is:
\begin{equation}
    \mathbf{P_{gc}}=
    \begin{pmatrix}
        0.0 & 0.5 & 0.4 & 0.1 \\
        1.0 & 0.0 & 0.0 & 0.0 \\
        1.0 & 0.0 & 0.0 & 0.0 \\
        1.0 & 0.0 & 0.0 & 0.0
    \end{pmatrix}
    \label{gc_eg}
\end{equation}
\end{example}

\begin{algorithm}[t]
\SetAlgoLined
\KwIn{\begin{enumerate}
    \item Hamiltonian $\mathcal{H}=\sum_{k=1}^{n}h_kH_k$.
    \item Classic oracle function $CNOT\_count(H_i,H_j)\rightarrow \mathbb{N}$.
    \item Classic subroutine $solve\_mincost\_flow\_problem(G,w,C)\rightarrow f$.
\end{enumerate}}
\KwOut{$n\times n$ Transition Matrix $\mathbf{P_{gc}}$}
%$\mathbf{P}_1\gets \mathbf{0}$ \\
$V\gets\{S,T\}$; $E\gets\emptyset$\;
\For{$i\in\{1,\dots,n\}$}{
    $V\gets V\cup\{H^{prev}_i,H^{next}_i\}$; $E\gets E\cup\{\edge{S,H^{prev}_i},\edge{H^{next}_i,T}\}$\;
    $c(\edge{S,H^{prev}_i}) = c(\edge{H^{next}_i,T}) = |h_i|/(\sum_{k}|h_k|)$;  $w(\edge{S,H^{prev}_i}) = w(\edge{H^{next}_i,T}) = 0$\;
}
\For{$i\in\{1,\dots,n\}$}{
    \For{$j\in\{1,\dots,n\}$}{
        \uIf{$i \ne j$}{
            $E\gets E\cup\{\edge{H^{prev}_i,H^{next}_j}\}$\;
            $c(\edge{H^{prev}_i,H^{next}_j}) = \infty$; $w(\edge{H^{prev}_i,H^{next}_j}) = CNOT\_count(H^{prev}_i,H^{next}_j)$\;
        }
    }
}

$G_{flow}\gets(V,E)$\;
$f\gets solve\_mincost\_flow\_problem(G_{flow},w,C)$\;
% Get the minimum-cost flow $f$ of graph G (the required flow is $1.0$) \; 
%\tcp{If $\exists j,|h_j|/(\sum_{k}|h_k|) > 0.5$, there is no solution. We can decompose $h_jH_j$ into $0.5h_jH_j+0.5h_jH_j$ and start over. Since it is hard to find $j$ satisfies $|h_j|/(\sum_{k}|h_k|) > 0.5$ in practice, we would ignore such case here.}

\For{$i\in\{1,\dots,n\}$}{
    \For{$j\in\{1,\dots,n\}$}{
        $\mathbf{P_{gc}}_{ij}\gets\frac{f(\edge{H^{prev}_i,H^{next}_j})}{|h_i|/(\sum_{k}|h_k|)}$\;
    }
}

\Return $\mathbf{P_{gc}}$\;
\caption{Construct $\mathbf{P_{gc}}$ out of the Min-Cost Flow Problem}
\label{alg:P_1}
\end{algorithm}

% For the example with Hamiltonian defined in \eqref{equ:exampledef}, we are going to demonstrate the process of constructing the matrix $\mathbf{P_1}$. 
% First using the trick in \ref{subsec:cnot_cancellation_trick} (the function in \eqref{equ:CNOTnum}), we can find that the CNOT gate number between the $Rz$ gates of $H_0$ and $H_1$ (or $H_1$ and $H_0$), $H_0$ and $H_2$ (or $H_2$ and $H_0$), $H_1$ and $H_2$ (or $H_2$ and $H_1$), are $1,5,6$, respectively. And then calculate the value of $|h_j|/\sum_{i}|h_i|$. Now we have got the parameters of the graph model. And we can get the graph (the $1$-$17$th lines of the Algorithm-\ref{alg:P_1}) which is shown in Fig-\ref{fig:P_1}. In the picture, the first number of an edge is its capacity and the second number is its weight.

%\subsection{Correctness and Properties of $\mathbf{P_1}$}
%\label{sec:property-correctness-P1}

Actually, we can prove that if we compile the quantum Hamiltonian simulation circuit by executing Algorithm~\ref{alg:ARQSC_FT} with the transition matrix $\mathbf{P_{gc}}$, the expectation of the number of CNOT gates in the generated circuit is minimized and the potential from the gate cancellation (depending on the $CNOT\_count$ function) is fully utilized.
%\yhliu{We state that the total cost we defined previously captures the expectation of the number of CNOT gates in the final circuit based on compilation via the transition matrix $\mathbf{P}_1$.}
\begin{proposition}[Cost is the Expectation of CNOT Gates Count]
\label{theorem:cost-is-expectation}
Given the MCFP framework in Fig.~\ref{fig:flowgraph-problem}, when the cost function $w(\edge{i,j})$ in the flow network can represent the number of CNOT gates required when transiting between the Hamiltonian terms $e^{i\lambda tH_i/N}$ and $e^{i\lambda tH_j/N}$,
%Given $G_{\mathbf{P_1}}$, $C(\edge{i,j})$, $w(\edge{i,j})$, any flow $f(\edge{i,j})$ that satisfies Theorem~\ref{theorem:flow-constraint}, and the generated transition matrix $\mathbf{P_{gc}}$, 
the total cost $\mathcal{W}=\sum_{\edge{i,j}\in E}f(\edge{i,j})w(\edge{i,j})$ is the mathematical expectation number of CNOT gates between circuit snippets of $H^{prev}$ and $H^{next}$ if the circuit is compiled via $\mathbf{P_{gc}}$.
\end{proposition}

\begin{proof}
See Appendix~\ref{sec:proof-cost-is-expectation}
\end{proof}

% Let's consider the previous sampled and the next sampling terms are $H_x$ and $H_y$, respectively. The distribution of $H_x$ is $\pi$. We denote the number of CNOT gates between the $Rz$ gates of the operator $e^{i\lambda tH_x/N}$ and $e^{i\lambda tH_y/N}$ as $Y_i$. Consequently, we can compute the expected value of $Y_i$.

% Let's denote the random variable $X_i$ as the index of the $i$-th sampling result ($i\in \{0,1,2,3,...,n-1\}$), and the number of CNOT gates between the operators $e^{i\lambda tH_x/N}$ and $e^{i\lambda tH_y/N}$ as CNOT\_num$(x,y)$. Then the mathematical expectation of $Y_i$ (the proof is in \ref{sub:EY}): 
% \begin{equation}
% \label{modelconvert}
%      \mathrm{E}[Y_i] = (\text{The cost flow of the graph $G$}).
% \end{equation}

%yhliu{It is evidence that the expectation of CNOT gates between the $R(z)$ gate of the first $e^{i\lambda tH_i/N}$ and the last $e^{i\lambda tH_i/N}$ circuit snippet, which equals to approximately all CNOT gates in the circuit ($N$ is relatively large, CNOT gates on the left and right most could be ignored), is minimized via the min-cost flow problem.}

\subsection{Transition Matrix Combination}\label{combination}
After the constraints in Theorem~\ref{theorem:flow-constraint} are satisfied, the transition matrix generated from the MCFP has already satisfied the stationary distribution preserving condition, but the strong connectivity constraint is not guaranteed. In \myCompilerName, this constraint is realized by combining multiple transition matrices.

%Once we find some transition matrices, we could combine them into a new transition matrix through linear combination, which still preserves Theorem~\ref{theorem:correctness}.

\begin{theorem}[Combining Transition Matrices]
\label{theorem:linear-combination}
Given an array of transition matrices $\mathbf{P}_1$, $\mathbf{P}_2$, $\dots$, $\mathbf{P}_k\}$, all of which satisfy the stationary distribution requirement ($\forall \mathbf{P}_i, \pi\mathbf{P}_i = \pi$) in Theorem~\ref{theorem:correctness}, and an array of parameters $\Theta_1$, $\Theta_2$, $\dots$, $\Theta_k$ where $0 \leq \Theta_i \leq 1$ for all $1\leq  i \leq k$ and $\sum_{i=1}^k \Theta_i = 1$, a new transition matrix $ \mathbf{P} = \sum_{i=1}^k\Theta_i\mathbf{P}_i$ also satisfies  the stationary distribution requirement in Theorem~\ref{theorem:correctness}.
\end{theorem}

\begin{proof}
    %\yhliu{Not Trivial, please provide a short proof.}
    $\pi\mathbf{P} = \pi\sum_{i=1}^k\Theta_i\mathbf{P}_i = \sum_{i=1}^k\Theta_i\pi\mathbf{P}_i = \sum_{i=1}^k\Theta_i\pi = \pi$
\end{proof}

Theorem~\ref{theorem:linear-combination} indicates that we can safely combine multiple transition matrices with normalized positive weights. 
This can enable reconciling multiple optimization opportunities provided by matrices together and guarantee strong connectivity.

\textbf{Strong Connectivity Guarantee}
Recall that the transition matrix for the vanilla randomized compilation qDrift algorithm $\mathbf{P_{qd}}$ is a valid matrix for the compilation Algorithm~\ref{alg:ARQSC_FT}. All entries in $\mathbf{P_{qd}}$ are strictly larger than $0$ since we must have $|h_j|>0$ in the Hamiltonian decomposition $\mathcal{H} = \sum_jh_jH_j$. For any transition matrix $\mathbf{P}$, we can always combine it and $\mathbf{P_{qd}}$ by Theorem~\ref{theorem:linear-combination} to make the new transition matrix all positive. Then the new transition matrix will represent a complete direct graph and naturally guarantee strong connectivity.
That ensures we can always find transition matrices that satisfy the two requirements to guarantee the overall correctness of the \myCompilerNameSpace compilation framework.

\begin{example}
For the Example~\ref{eg:example-hamiltonian}, we can obtain a new transition matrix $\mathbf{P} = 0.4\mathbf{P_{qd}} + 0.6\mathbf{P_{gc}}$ ($\mathbf{P_{qd}}$ and $\mathbf{P_{gc}}$ can be found in Example~\ref{eg:example-hamiltonian} and Example~\ref{exampleP1}, respectively) to combine the qDrift randomized compilation and the gate cancellation optimization in \myCompilerName. The new matrix is:
\begin{equation}
    \mathbf{P}
    =\begin{pmatrix}
        0.2&0.4&0.32&0.08\\
        0.8&0.1&0.08&0.02\\
        0.8&0.1&0.08&0.02\\
        0.8&0.1&0.08&0.02\\
    \end{pmatrix}
\end{equation}
\end{example}

\subsection{Reasoning about Convergence Speed via Spectrum Analysis}
\label{Converge}

% \Junyu{delete prop, add equation in appendix, add small example, arguing why spectra matters}

With the transition matrix combination technique above, we can construct many transition matrices that satisfy the conditions in Theorem~\ref{theorem:correctness}, and they have the same error bound. 
However, in practice, we observe that their actual convergence speed, indicated by the variance of the sampled circuit unitary, can still be different.
In this section, we show that we can reason about the convergence speed of the sampling process by analyzing the spectrum of the transition matrix. 
%In this section, we discuss the convergence rate given different transition matrices by analyzing the spectra of each transition matrix $\mathbf{P}$.

Usually, in a sampling process, the sampling result after the first few steps has a very large variance, and the sampling variance gradually decreases as the number of sampling steps increases. 
The specific sampling problem in this paper can be formulated in the following:
%Our problem is: 
\textit{
Suppose the initial distribution is $\pi_{0}$ and the transition matrix is $\mathbf{P}$. We are interested in how fast $\mathbf{P}^{k}\pi_{0}$ can converge to the stationary distribution $\pi$ since this rate will affect the speed of our sampled Hamiltonian simulation circuit converging to the ideal circuit.}

We use the power method to reason about this convergence speed.
Suppose $\mathbf{P}$ has the following Jordan decomposition, and the eigenvalues in each Jordan block are in decreasing order:
\begin{equation*}
    \mathbf{P} = V \cdot diag(J_{1},\ldots,J_{p}) \cdot V^{-1},
\end{equation*}
with $|\lambda_1|>\ldots>|\lambda_p|$, where $J_{i}$ which is $n_{i}$ dimensional Jordan matrix associate with eigenvalue $\lambda_i$. 
% \begin{equation*}
%     J_{i} = 
%     \begin{pmatrix}
%         \lambda_{i} & 1 & 0 & 0 \\
%         0 & \lambda_{i} & 1 & 0 \\
%         0 & 0 & \lambda_{i} & 1 \\
%         0 & 0 & 0 & \lambda_{i} \\
%     \end{pmatrix}
% \end{equation*}

Then we divide V according to columns, $V= (V_{1}, V_{2}, \ldots, V_{p} )$ with $rank(V_{i})=n_{i}$. 
Let $y = (y^{T}_{1}, y^{T}_{2}, \ldots, y^{T}_{p})^{T}=V^{-1}\pi_{0}$, then:
\begin{equation}
    \mathbf{P}^{k}\pi_{0} = \lambda_{1}^{k} [V_{1}y_{1}+V_{2}(\frac{J_{2}}{\lambda_{1}})^{k}y_{2}+\cdots] =   [V_{1}y_{1}+V_{2}(J_{2})^{k}y_{2}+\cdots] ,
    \label{eq:conv}
\end{equation}
Note that we always have $0 \leq |\lambda_{i}| \leq 1$, where $\lambda_{i}$ is the spectra of a transition matrix of a Markov chain. The largest eigenvalue $\lambda_{1} = 1$ is always associated with the stationary distribution $\pi$.
Here $V_{1}y_{1}$ is stationary distribution $\pi$ that will remain after many sampling steps. The following terms $V_{i}(\frac{J_{i}}{\lambda_{1}})^{k}y_{i}$ ($1<i\leq p$) will converge to zero, since the spectral radius: $\rho(\frac{J_{i}}{\lambda_{1}})=|\lambda_{i}|/|\lambda_{1}|<1$. 

% Thus the convergence rate can be captured by $(|\lambda_{2}|/|\lambda_{1}|)^{k}$. \textbf{To achieve a $\epsilon$ precision (in distribution), we need $k=\mathcal{O}(log_{\frac{1}{|\lambda_{2}|}}(\frac{1}{\epsilon}))$ steps.} We define the value $1-|\lambda_{2}|$ as "Spectra Gap".

%We have the following proposition.

% \begin{proposition}[Convergence Speed]\label{theorem:convergence-speed}
%     Suppose the transition matrix $\mathbf{P}$ has the eigenvalues $\{\lambda_i\}$ and  $|\lambda_1|> |\lambda_2| > \ldots$
% To let $\mathbf{P}^{k}\pi_{0}$ converge to the stationary distribution $\pi$ with precision $\epsilon$, the number of sampling steps $k$ required is:
% $$k=\mathcal{O}(log_{\frac{1}{|\lambda_{2}|}}(\frac{1}{\epsilon}))$$
% \end{proposition}

% \begin{proof}
%    \TODO{postponed to xx}
% \end{proof}

% The proposition above shows that the convergence speed is affected by the eigenvalues associated with the eigenspace orthogonal to the stationary distribution and is bounded by the second large eigenvalue of the transition matrix.
Equation (\ref{eq:conv}) shows that the convergence speed is affected by the transition matrix's spectra distribution.
Except for the large eigenvalue $\lambda_1=1$, more spectra with smaller values will result in a higher convergence rate.
In the following example, we will investigate the spectra of the two transition matrices obtained by different methods for the same Hamiltonian.

\begin{figure*}[h]
    \centering
    \includegraphics[width=0.75\textwidth]{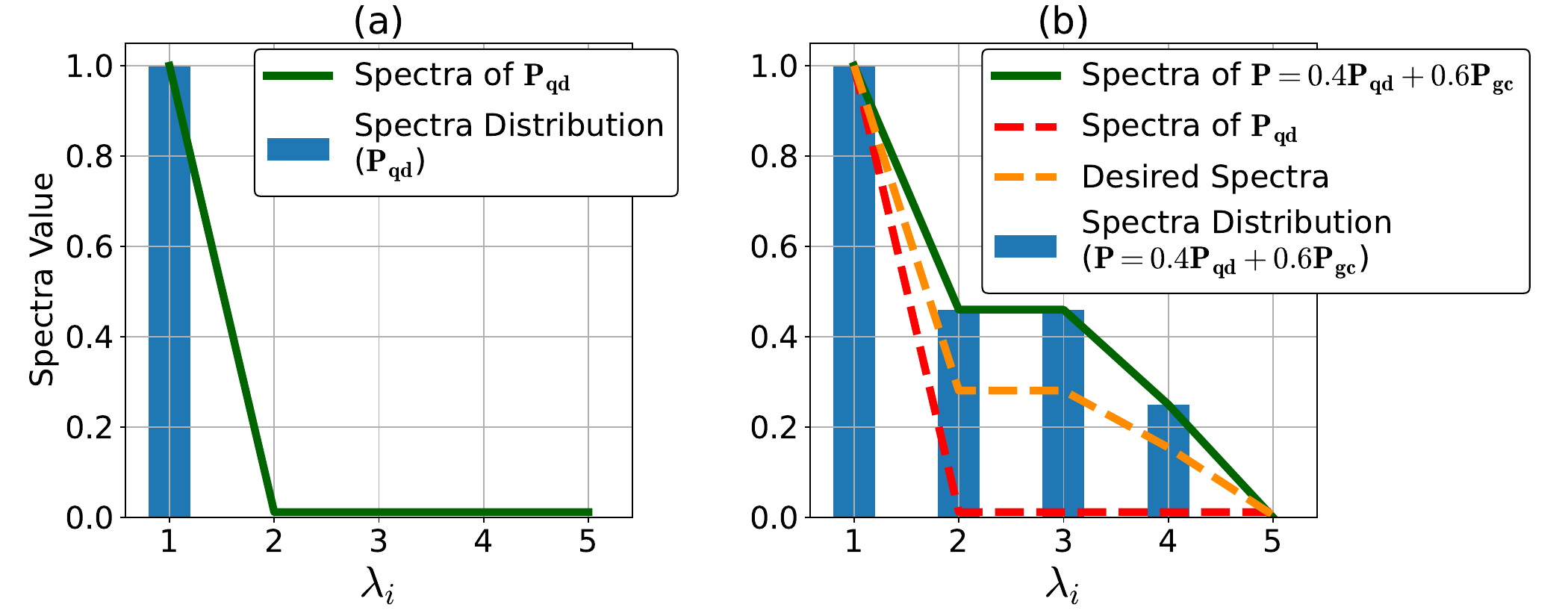}
  %  \text{Evolution Time: $\frac{\pi}{4}$s}
    \vspace{-15pt}
    \caption{Examples of transition matrix spectra. The histograms show the spectra distribution, and the lines capture the eigenvalue decreasing trend.}
%    \vspace{-15pt}
    \label{fig:spec}
\end{figure*}

\begin{example}

Given a Hamiltonian $\mathcal{H}=1.0\cdot IIIZY + 1.0\cdot XXIII + 0.7\cdot ZXZYI + 0.5\cdot IIZZX + 0.3\cdot XXYYZ$. We can have two transition matrices in \myCompilerNameSpace to compile this Hamiltonian: 
    \label{eg:spec}
    \begin{enumerate}
        \item  Consider the transition matrix $\mathbf{P_{qd}}$ for vanilla qDrift (see Corollary~\ref{corollary:qdrift} above). Its spectra are $\lambda_{1} = 1$ and $\lambda_{2} = \lambda_{3} = \lambda_{4} = \lambda_{5} = 0$ because it has rank 1 and the only one non-zero eigenvalue is 1 corresponding to the stationary distribution. The spectra distribution of $\mathbf{P_{qd}}$ is shown in Fig.~\ref{fig:spec} (a).
        \item  Consider the combined transition matrix $\mathbf{P} = 0.4\mathbf{P_{qd}} + 0.6\mathbf{P_{gc}}$ ($\mathbf{P_{gc}}$ is the transition matrix optimized for CNOT gate cancellation introduced in Section~\ref{sec:opt_for_CNOT}.). Its spectra are $\lambda_{1} = 1$, $\lambda_{2} = \lambda_{3} = 0.46$, $\lambda_{4} = 0.25$, and $\lambda_{5} = 0$. The spectra distribution of $\mathbf{P}$ is shown in Fig.~\ref{fig:spec} (b).
    \end{enumerate}

        To characterize the properties of the spectra, we plot the boundary of the spectra distribution to show the eigenvalue decreasing trend in Fig.~\ref{fig:spec}.

        % The combined transition matrix $\mathbf{P} = 0.4\mathbf{P_{qd}} + 0.6\mathbf{P_{gc}}$ in Example~\ref{} has rank 2 with the following spectrum: $\lambda_{1} = 1, \lambda_{2} = \TODO{??}, \lambda_{3} = \lambda_{4} = 0$.
\end{example}
The transition matrix $\mathbf{P_{qd}}$ always has rank 1 with $\lambda_{2} = \cdots = \lambda_{p} = 0$, so it has the smallest variance in the sampled circuit unitary. 
%According to Equation (\ref{eq:conv}), the transition matrix $\mathbf{P_{qd}}$ always has $\lambda_{2} = \cdots = \lambda_{p} = 0$, so it has the smallest variance in the sampled circuit unitary. 
In contrast, the combined transition matrix $\mathbf{P}$ has non-zero $\lambda_{i}$'s for $i>1$. Thus it has a larger variance and lower convergence speed. Such non-zero eigenvalues come from the $\mathbf{P_{gc}}$ component.
$\mathbf{P_{gc}}$ uses these additional spectra to increase the gate cancellation.

\subsection{
%\yhliu{Improving Transition Matrix by Perturbation}
Compilation via Random Perturbation
}
\label{sub: An Upgrade Method}

\iffalse
\begin{algorithm}[t]
\SetAlgoLined
\For{$i\in\{1,\dots,n\}$ }{
    \For{$j\in\{1,\dots,n\}$ }{
        \uIf{$i\ne j$}{
            $\epsilon\gets$ small random number \\
            $w(\edge{H^{prev}_i,H^{next}_j})\gets w(\edge{H^{prev}_i,H^{next}_j}) + \epsilon$
        }
    }
}
\caption{Improving $\mathbf{P_1}$ By Perturbing the Weight}
\label{alg:P_2}
\end{algorithm}
\fi
%While Algorithm-\ref{alg:P_1} provides a transition matrix that minimizes the expected number of CNOT gates, it still has some practical drawbacks. When applied, it often generates a sparse transition matrix $\mathbf{P_1}$ (i.e., a relatively large high proportion of the elements in the matrix are zero). Therefore, if we use $\mathbf{P_1}$ as the transition matrix, the distribution of the subsequent sampling will also be sparse. This means there's a high likelihood that the next sampling result will fall within a certain, relatively small set. Adding $\mathbf{P_0}$ with some weight to $\mathbf{P_1}$ can weaken the randomness of the sampling process. How might we enhance the randomness of the transition matrix while sacrificing maybe only a minimal number of expected CNOT gates? 

%Previous works~\cite{childs2019faster, poulin2011quantum, campbell2017shorter, hastings2016turning} on quantum simulation have shown that introducing randomness even in the static order can improve the overall compilation performance.
%This motivates us to introduce more randomness in \myCompilerNameSpace.

According to the transition matrix spectra analysis above, a desired transition matrix with faster convergence and smaller sampling variance would have a spectrum with smaller eigenvalues. That is, its eigenvalue decreasing trend should be lower and closer to the trend of  $\mathbf{P_{qd}}$. An example of such spectra is denoted by an orange dash in Fig.~\ref{fig:spec} (b).
Our objective is to generate transition matrices with such spectra while maintaining the gate cancellation capability.

%\color{blue}
%Our evaluation shows that randomness in the transition matrix can help push the spectra boundary of $\mathbf{P_{gc}}$ (See Fig. ~\ref{fig:spec} (b)) close to the spectra boundary of $\mathbf{P_{qd}}$ (See Fig. ~\ref{fig:spec} (a)), thus can stabilize the compilation process and drive the evolution unitary to the ideal one at a higher speed. Section ~\ref{sec:stdanalysis} will show how to take advantage of this randomness practically.
%\color{black}

The insight of our method is that if we add some randomness in the transition matrix, the eigenvectors or the eigenspace will be changed. After combining multiple such matrices, the randomly altered eigenvectors will amortize the eigenvalues to obtain a trend with smaller eigenvalues. Meanwhile, we need to ensure that the stationary distribution eigenvector is always preserved.

Based on this insight, we proposed the following procedure to randomly perturb the transition matrix. Directly adding randomness to the transition matrix may result in invalid transition matrices that do not satisfy the requirements in Theorem~\ref{theorem:correctness}.
Fortunately, the randomness can be introduced in the MCFP framework without affecting the overall correctness.
%Also, a small perturbation will not significantly increase the overall cost so that the gate cancellation is largely preserved.
%Specifically, 
We can introduce a small random perturbation $\epsilon$ to the weight function in the MCFP:
\begin{gather*}
     w(\edge{H^{prev}_i,H^{next}_j})\gets w(\edge{H^{prev}_i,H^{next}_j}) + \epsilon
\end{gather*}
Such a randomness-adding method has two benefits:
\begin{enumerate}
    \item The overall correctness is guarded by the capacity function and overall flow in the MCFP, which are not affected by the random perturbation.
    \item The gate cancellation functionality is only slightly affected since the cost function is still close to the $gate\_count$. So that we will not significantly lose the gate reduction benefits.
\end{enumerate}
We can generate multiple different randomly perturbed MCFPs, solve them to obtain the transition matrices, and finally average them.
The transition matrix obtained from the randomly perturbed MCFPs is denoted as $\mathbf{P_{rp}}$. 

Finally, all the transition matrices collected in this and the previous sections, $\mathbf{P_{qd}}$ (vanilla qDrift for strong connectivity), $\mathbf{P_{gc}}$ (for gate cancellation), and $\mathbf{P_{rp}}$ (for more randomness), can be combined into one new transition matrix to deliver compilation results that can not be achieved by any of them individually. 

%\color{blue} Next section will theoretically show how this randomness in combination of the transition matrix can help stabilize the compilation process. \color{black}

\section{Evaluation}

In this section, we evaluate the performance of \myCompilerName. We will first introduce our experiment setup, then evaluate the performance improvement of \myCompilerNameSpace optimization, and finally study the compilation time consumption.

%%%%%%%%%%%

%present an experiment that evaluates the performance of Markov Simulation by comparing it with the original random compiler (Algorithm-\ref{alg:ARQSC_FT}). The evaluation will focus on the accuracy of the simulation and the number of quantum gates used.

%To assess the accuracy of the simulation, we will perform direct numerical calculations on classical computers and compare the results obtained from Markov Simulation with the exact values. However, due to limitations in computing resources, we may not be able to test the accuracy for Hamiltonians with long Pauli strings (i.e., $\otimes_{i=1}^{n-1}\sigma_i$ with a large value of $n$).

% For short Pauli strings, we will evaluate both the accuracy and the number of quantum gates used by Markov Simulation. This will allow us to assess the trade-off between accuracy and circuit complexity.

% For long Pauli strings, where direct numerical calculations may not be feasible, we will focus on evaluating the number of quantum gates used by Markov Simulation. This will provide insights into the efficiency of the algorithm for simulating Hamiltonians with long Pauli strings.

%By conducting these evaluations, we aim to demonstrate the effectiveness of Markov Simulation in terms of accuracy and gate optimization, and provide insights into its performance for both short and long Pauli strings.

\subsection{Experiment Setup}

\textbf{Experimental Configuration} %In the  protocol, we introduced the transition matrix $\mathbf{P_1}$ that captured the information of CNOT numbers, as well as its improvement $\mathbf{P_1''}$ by adding some random information. 
We prepared three configurations to demonstrate the effect of different transition matrices in \myCompilerName.
\textbf{1. Baseline} is the qDrift algorithm~\cite{campbell2019random} with transition matrix be $\mathbf{P}=\mathbf{P_{qd}}$ (defined in Corollary~\ref{corollary:qdrift}) followed by applying gate cancellation~\cite{gui2020term} on the randomized sequence. 
\textbf{2. \myCompilerName-GC} includes the transition matrix $\mathbf{P_{gc}}$ for optimize CNOT gate cancellation with a certain weight: $\mathbf{P}=0.4\mathbf{P_{qd}}+0.6\mathbf{P_{gc}}$. 
\textbf{3. \myCompilerName-GC-RP} further considers adding the random perturbations, and the transition matrix is $\mathbf{P}=0.4\mathbf{P_{qd}}+0.3\mathbf{P_{gc}}+0.3\mathbf{P_{rp}}$.
Note that the weight perturbation $\epsilon$ increases 1 in $CNOT\_count$ between two Hamiltonian terms with probability 0.5 and $P_{rp}$ are averaged from 100 different random perturbations for each benchmark.% \Junyu{by randomly increasing 1 count in $CNOT\_count$ between two Hamiltonian terms with probability 0.5}.
%\Junyu{The repetition times to construct 
%$\mathbf{P_1''}$ mentioned in Section~\ref{sub: An Upgrade Method} is 100 in all experiments.}

% In this subsection, we show the details about the specific tasks that we will apply our method on. 

%\begin{wraptable}{l}{0.47\columnwidth}
\begin{table}[t]
  \centering
  % \vspace{-5pt}
  \caption{Benchmark Information}
  \vspace{-5pt}
       \resizebox{0.45\columnwidth}{!}{ 
    \begin{tabular}{|c|c|c|c|}\hline

    \small
%    \begin{tabular}{0.5\textwidth} { 
%  | >{\centering\arraybackslash}X 
%  | >{\centering\arraybackslash}X 
%  | >{\centering\arraybackslash}X
%  | >{\centering\arraybackslash}X | 
%|c|c|c|c|
%}

    %\hline
  %  \hline
          Benchmark & Qubit\# & Pauli String\# & Time \\
    \hline
     Na+ & 8 & 60 & $\pi$/4 \\ \hline
     Cl- & 8 & 60 & $\pi$/4 \\ \hline
     Ar & 8 & 60 & $\pi$/4 \\ \hline
     OH- & 10 & 275 & $\pi$/4 \\ \hline
     HF & 10 & 275 & $\pi$/4 \\ \hline
     LiH (froze) & 10 & 275 & $\pi$/4 \\ \hline
     $\rm BeH_2$ (froze) & 12 & 661 & $\pi$/4 \\ \hline
     LiH & 12 & 614 & $\pi$/4 \\ \hline
     $\rm H_2O$ & 12 & 550 & $\pi$/4 \\ \hline
     SYK model 1 & 8 & 210 & 0.15 \\ \hline
     SYK model 2 & 10 & 210 & 0.15 \\ \hline 
     $\rm BeH_2$  & 14 & 661 & 0.15 \\\hline
   % \hline
  %  \multirow{2}{*}{\parbox{\linewidth}{Evolution Time Study}} & Na+ & 8 & 60 \\ \cline{2-4}
  %  & OH- & 10 & 275 \\ 
  %  \hline
  %  \multirow{10}{*}{\parbox{\linewidth}{Scalability Study}} & BeH2 & 14 & 661 \\ \cline{2-4}
  %  & random1 & 10 & 100 \\ \cline{2-4}
  %  & random2 & 10 & 500 \\ \cline{2-4}
  %  & random3 & 10 & 1000 \\ \cline{2-4}
  %  & random4 & 20 & 100 \\ \cline{2-4}
  %  & random5 & 20 & 500 \\ \cline{2-4}
  %  & random6 & 20 & 1000 \\ \cline{2-4}
  %  & random7 & 30 & 100 \\ \cline{2-4}
  %  & random8 & 30 & 500 \\ \cline{2-4}
  %  & random9 & 30 & 1000 \\
  %  \hline
    \end{tabular}%
     }
    \vspace{-5pt}
  \label{tab:benchmarks}%
  \end{table}
%\end{wraptable}%

\textbf{Benchmarks} We generate the Hamiltonian for various molecule or ion physical systems 
%The real molecule or ion benchmarks used in our experiments are generated 
using the PySCF library \cite{PySCF} and Qiskit library \cite{Qiskit} (including Qiskit Nature \cite{https://doi.org/10.5281/zenodo.7828767}). 
We use the Jordan Wigner fermion-to-qubit transformation~\cite{Jordan1928} and the freezing core method to control the size of the Hamiltonian by only simulating the out-layer electrons.
We also include the SYK models~\cite{haldar2023numerical} from quantum field theory.
The benchmark information is summarized in Table~\ref{tab:benchmarks}.
%These benchmarks, which use Jordan Wigner's method \cite{Jordan1928} of fermion-to-qubit transformation, involve simulating the electronic structures of certain molecular systems. The freezing core method are also adopted to reduce the computation time. The qubit numbers for molecules or ions in our experiment are from 8 to 14, and the numbers of Pauli strings in the Hamiltonian are from 60 to 661. \Junyu{In addition, we also select some Hamiltonian that are using SYK mapping from Majorana operators, the qubit numbers are 8 and 10, and the number of the Pauli strings are both 210}. See Table~\ref{tab:benchmarks} for the overall benchmarks. To demonstrate the compilation time, we also construct random Hamiltonian with qubit numbers from 10 to 30, and numbers of Pauli string are from 100 to 1000. 

% \iffalse
% However, due to the limited computational resources described in subsubsection \ref{subsub:backend}, we are restricted to Hamiltonians with Pauli strings of length up to 12. Simulating Hamiltonians with very short Pauli strings can be efficiently done on classical computers, so we focus on Hamiltonians with Pauli strings of at least length 8. The evolution time for the simulations is chosen from the intersection of the interval $(0,2\pi]$ and the maximum time that our computational devices can handle effectively. By selecting suitable evolution times within this range, we can accurately assess the performance of our algorithm for various Hamiltonians.
% \fi

\textbf{Metrics}
We focus on the \textbf{CNOT gate count} as the optimization in our MCFP formulation. % is targeting reducing the number of CNOT gates. 
CNOT gates are usually much more expensive than single-qubit gates in the near-term noisy quantum computers~\cite{linke2017experimental, debnath2016demonstration, gambetta2019cramming}. 
Even in future fault-tolerant quantum computers with error correction, the overhead of CNOT gates is still considerable due to various constraints~\cite{maslov2016optimal,javadi2017optimized}.
We also calculate the \textbf{algorithmic accuracy}, i.e., how closely the unitary matrix $\mathbf{U}_{app}$ of the circuit generated by our compiler approximates the exact unitary evolution matrix $\mathbf{U} = e^{iHt}$.
This is indicated by the unitary fidelity defined as $tr(\mathbf{U}_{app}\cdot \mathbf{U}^{\dagger})/2^{n}$ ($n$ is the number of qubits). A fidelity closer to $1$ is better.

\textbf{Implementation} We implement \myCompilerNameSpace in Python with Numpy. We use the Python networkx package~\cite{hagberg2008exploring} to implement and solve the MCFP. 
Calculating the approximated unitary in the generated circuits $\mathbf{U}_{app}$ is costly. We accelerate this on an NVIDIA A100 GPU with PyTorch~\cite{NEURIPS2019_9015} with CUDA~\cite{cuda}. 
Due to the limitation of GPU memory, we can only evaluate up to 14 qubits.

\begin{wrapfigure}{R}{0.66\textwidth}
    \centering
    \vspace{-5pt}
    \includegraphics[width=0.65\textwidth]{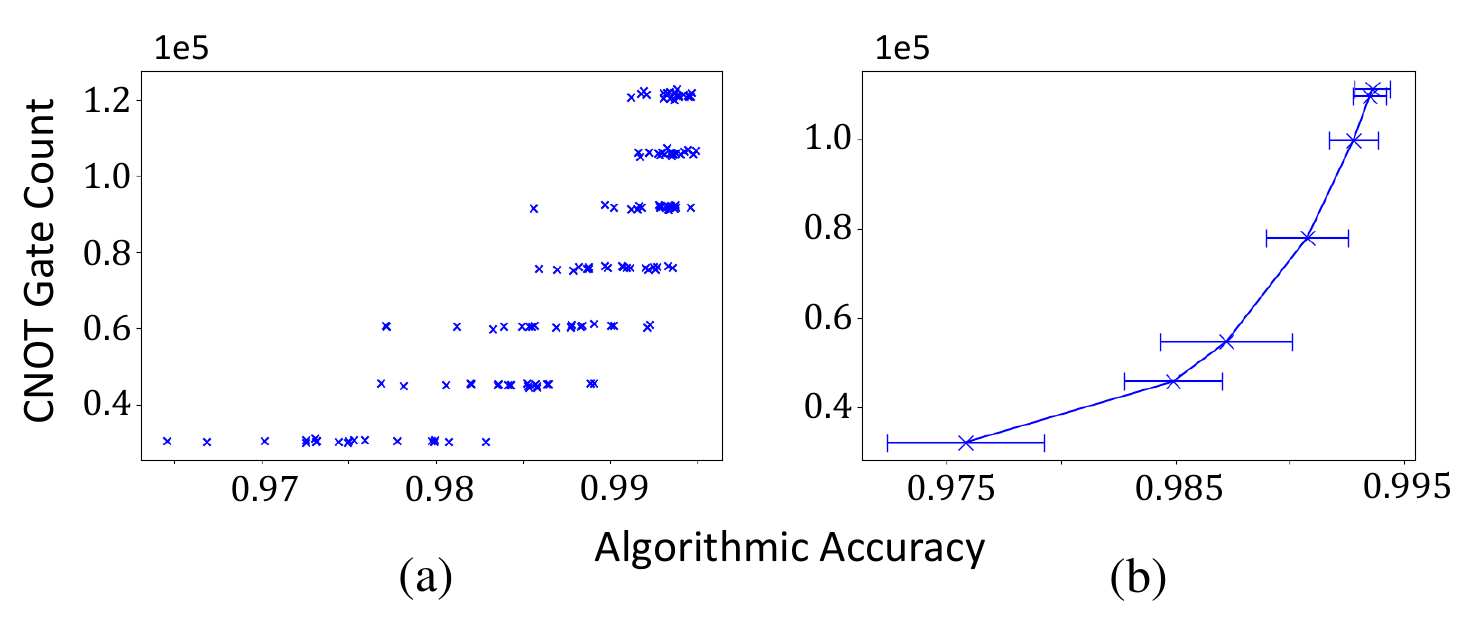}
    % \text{Left: Markov Chain Model}
    % \text{Right: Flow Model satisfies \eqref{equ:flow}}
%    \text{Evolution Time: $\frac{\pi}{4}$s}
 \vspace{-10pt}
        \caption{ Data processing with raw data from $\rm BeH_2$ (froze)
    }
    \label{fig:explanation}
\end{wrapfigure}

\textbf{Data Processing} 
The approximation error $\varepsilon$ is set to be $0.1$, $0.067$, $0.05$, $0.04$, $0.033$, $0.0286$, $0.025$ to make the distribution of the number of random sampling steps $N$ relatively uniform. %The evolution time $t$ is set to be $\frac{\pi}{4}$ to finish our numerical evaluation in an acceptable time (around dozens of hours for each benchmark).
%, which is suit for both large and small qubit numbers. 
For each experimental configuration, we repeat the compilation 20 times as our compilation is randomized (except for the largest benchmark $\rm BeH_2$, we repeat five times due to time limitation). %, \Junyu{except 5 times for BeH2 due to the computational limitations}. 
Fig.~\ref{fig:explanation} (a) shows an example of our raw data (from the $\rm BeH_2$(froze) benchmark). The X-axis is the approximation fidelity. The Y-axis is the number of CNOT gates in the compiled circuit. We can roughly control the number of gates with the approximation error, and it can be observed that the data points are clustered by the number of gates. However, we can hardly know the exact approximation fidelity before we obtain the compiled quantum circuit. 
To ensure a fair comparison, we first take the average for each cluster and then use $y=a+e^{bx+c}$ to fit the data points. 
This data processing method enables us to compare the number of gates under the same simulation accuracy between 0.992 and 0.994. 
%\Junyu{
The interpolated data points, as well as the standard deviation of the raw data, are plotted in Fig.~\ref{fig:explanation} (b).

\subsection{Overall Improvement} \label{sec:overall}

\begin{figure*}[t]
    \centering
    \includegraphics[width=1\textwidth]{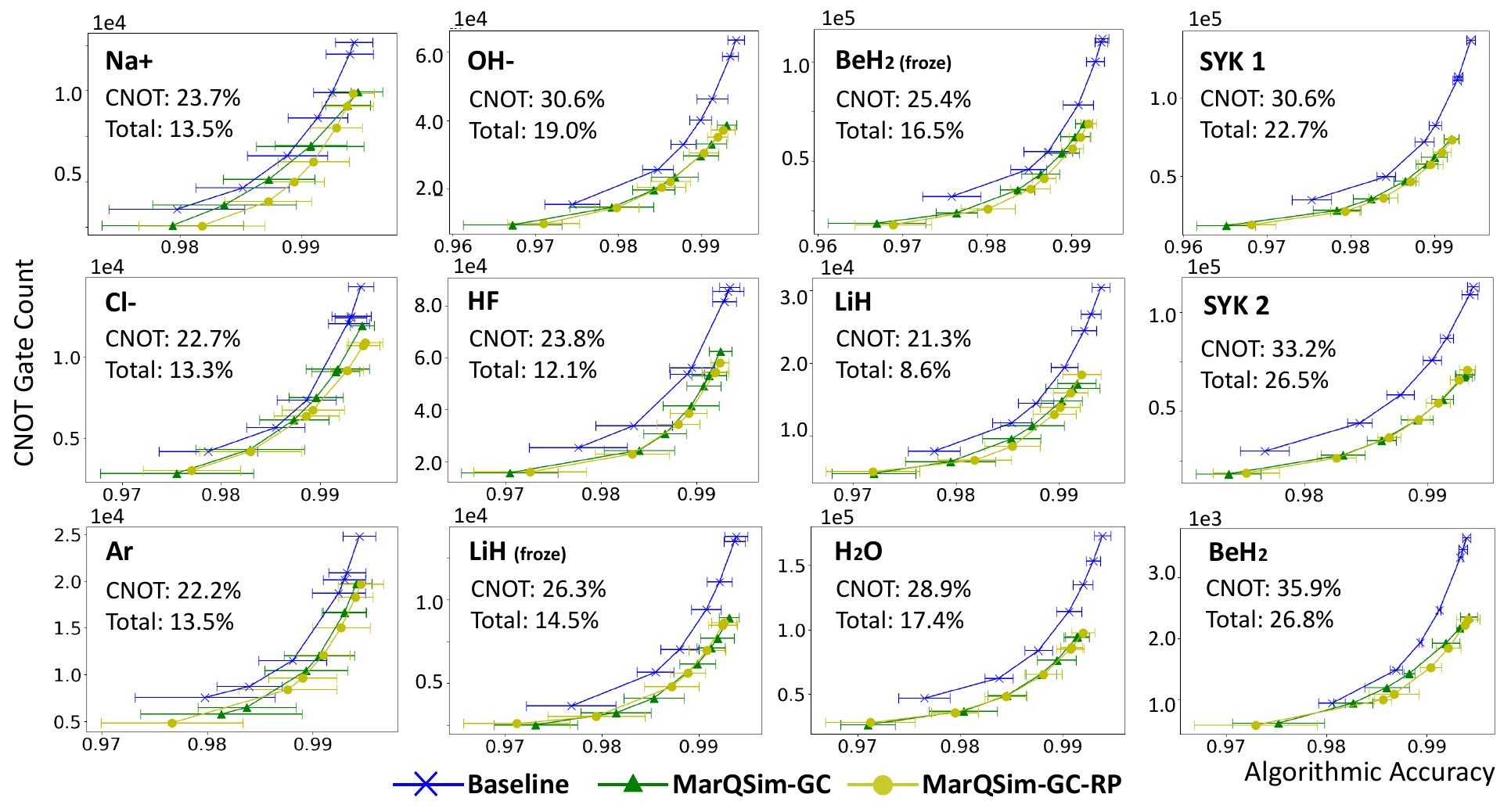}
  %  \text{Evolution Time: $\frac{\pi}{4}$s}
    \vspace{-20pt}
    \caption{Overall Improvement over all benchmarks}%.  The X-axis and Y-axis represent the algorithmic accuracy (higher is better) and CNOT gate count of the compiled circuit (lower is better), respectively.}
    \vspace{-10pt}
    \label{fig:overall}
\end{figure*}

Fig.~\ref{fig:overall} shows the results of all the benchmarks compiled with all three configurations. 
We also listed the CNOT gate reduction and total gate reduction for \textbf{\myCompilerName-GC-RP} in each sub-figure.
In summary, \textbf{\myCompilerName-GC} can reduce the number of CNOT gates by 25.1\% on average (up to 34.6\%).
We also observe that \textbf{\myCompilerName-GC} can reduce the number of single-qubit gates by 2.1\% even if the transition matrix $\mathbf{P_{gc}}$ is not tuned for removing single-qubit gates.
%\Junyu{
The total gate reduction for \textbf{\myCompilerName-GC} is 14.6\% on average. \textbf{\myCompilerName-GC-RP} can further achieve an average 27.0\% reduction on CNOT gates, 5.0\% reduction on single-qubit gates, and 17.0\% reduction for total gates. %, improve 1.9\%, 2.9\%, and 2.4\% reduction respectively compared with \textbf{\myCompilerName-GC}.%}.
In addition, we measure the standard deviation ($\sigma$) of the data point for each target accuracy setting. The result shows that \textbf{\myCompilerName-GC-RP} can reduce the standard deviation by 8.3\% on average compared with the \textbf{\myCompilerName-GC}.
This indicates that adding more randomness can make the compilation output more stable and reduce the probability of obtaining low-accuracy circuits.
%the general results. 
%The first column contains 8 qubit ions Na+, Cl-, and atom Ar. 
%The second column contains 10 qubit ion OH- and molecules HF and LiH. The third column contains 12 qubit molecules BeH2, LiH with unfreezing core and H2O respectively. We can see that the CNOT gates are relatively uniformly distributed, each data point cluster comes from one target accuracy.

%For \textbf{\myCompilerName-GC}, the average gate reduction is 17.08\%, 25.97\%, and 24.62\% for 8, 10, and 12 qubit systems respectively. For \textbf{\myCompilerName-GC-RP}, the reduction is 22.86\%, 26.88\%, and 25.18\%. 

%16.27\%, 3.86\%, and 7.14\% for 8, 10, and 12 qubit systems respectively, compared with the \textbf{\myCompilerName-GC}.

\subsection{Varying Transition Matrix Combination}
\label{sec:hyper}
% \Junyu{new section, compare spectra distribution}

In Section ~\ref{sec:overall}, the ratio of transition matrices ($\mathbf{P_{qd}}$, $\mathbf{P_{gc}}$, and $\mathbf{P_{rp}}$) is fixed. In this section, we adjust the ratio combining $\mathbf{P_{qd}}$ and $\mathbf{P_{gc}}$ to understand how it affects the CNOT gate reduction.
The results are shown in Fig.~\ref{fig:ratio}.
In each subgraph, we list the CNOT gate reduction corresponding to each combination ratio. Compared with the baseline that uses only qDrift, a higher combination ratio on $\mathbf{P_{gc}}$ will result in higher CNOT gate reduction since the transition matrix will be more tailored for gate reduction. For the eight benchmarks, the average CNOT reduction rates are $10.3\%$, $23.8\%$, $28.0\%$ for the transition matrix settings $\mathbf{P}=0.8\mathbf{P_{qd}}+0.2\mathbf{P_{gc}}$, $\mathbf{P}=0.4\mathbf{P_{qd}}+0.6\mathbf{P_{gc}}$, $\mathbf{P}=0.2\mathbf{P_{qd}}+0.8\mathbf{P_{gc}}$ respectively.
We also observe that there is an algorithmic accuracy loss as the $\mathbf{P_{gc}}$ component continues to increase.
As discussed in Section~\ref{Converge}, a larger $\mathbf{P_{gc}}$ component will lead to matrix spectra with larger eigenvalues, and thus the sampling will converge slower with larger variance.

\begin{figure*}[t]
    \centering
    \includegraphics[width=1\textwidth]{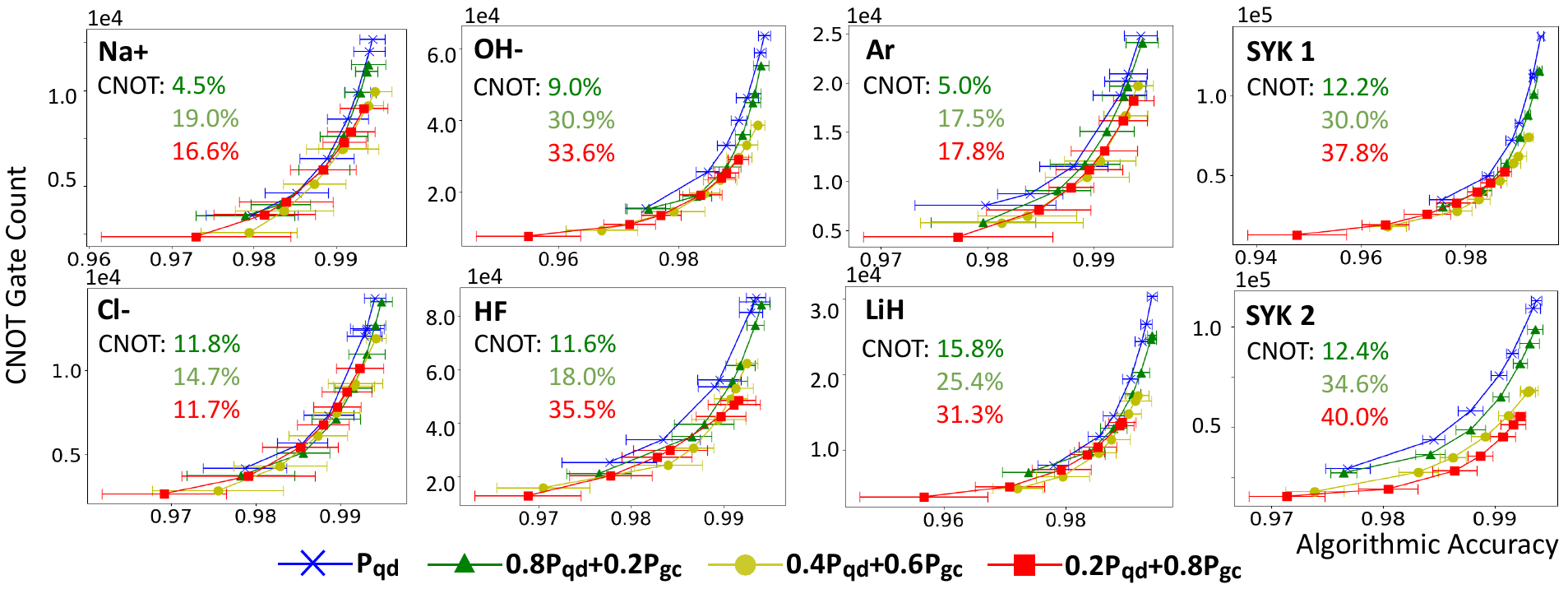}
  %  \text{Evolution Time: $\frac{\pi}{4}$s}
  %  \vspace{-20pt}
    \caption{Compilation results of varying ($\mathbf{P_{qd}}$, $\mathbf{P_{gc}}$) combination ratios}
 %   \vspace{-10pt}
    \label{fig:ratio}
\end{figure*}

%\color{blue}

\subsection{Matrix Spectra Change with Random Perturbation}
\label{sec:stdanalysis}

\begin{figure*}[t]
    \centering
    \includegraphics[width=0.8\textwidth]{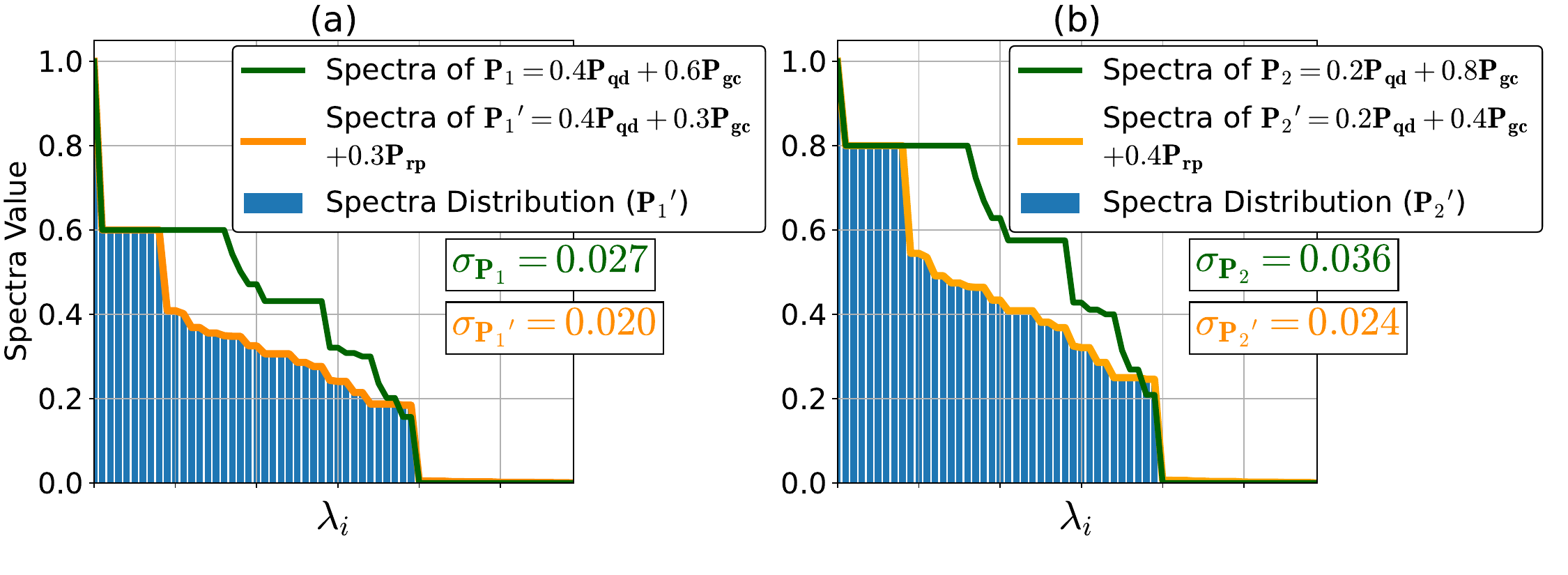}
  %  \text{Evolution Time: $\frac{\pi}{4}$s}
    \vspace{-15pt}
    \caption{Transition matrix spectra for Na+ benchmarks with different matrix combination configurations. The variance $\sigma$ of the sample circuit algorithmic accuracy for each configuration is also listed.}
 %   \vspace{-10pt}
    \label{fig:Expspec}
\end{figure*}

% \textbf{Second}, we explain how the randomness in the transition matrix can help stabilize the compilation and reduce the standard deviation ($\sigma$) of the sampled circuit, by including the different ratio of $\mathbf{P_{rp}}$.
In this section, we investigate the impact of randomness in the transition matrix experimentally. We select the combination $\mathbf{P}_{1}=0.4\mathbf{P_{qd}}+0.6\mathbf{P_{gc}}$ and $\mathbf{P}_{2}=0.2\mathbf{P_{qd}}+0.8\mathbf{P_{gc}}$ as baseline combination ratios, and introduce their randomly perturbed counterparts $\mathbf{P_1^\prime}=0.4\mathbf{P_{qd}}+0.3\mathbf{P_{gc}}+0.3\mathbf{P_{rp}}$ and $\mathbf{P_2^\prime}=0.2\mathbf{P_{qd}}+0.4\mathbf{P_{gc}}+0.4\mathbf{P_{rp}}$, respectively. This strategy can exclude the impact of different $\mathbf{P_{qd}}$ component ratios.
%to introduce randomness is to share a part of the ratio from $\mathbf{P_{gc}}$ to $\mathbf{P_{rp}}$, and make the ratio of $\mathbf{P_{qd}}$ unchanged.

% Section ~\ref{Converge} indicates combining transition matrix $\mathbf{P_{qd}}$ with $\mathbf{P_{gc}}$ will decrease the spectra gap and result in an unstable sampling process, finally lead to the large standard deviation in the sample points. We have observed a certain amount of increase in standard deviation when using the combined transition matrix: $\mathbf{P}=0.8\mathbf{P_{qd}}+0.2\mathbf{P_{gc}}$, $\mathbf{P}=0.4\mathbf{P_{qd}}+0.6\mathbf{P_{gc}}$, and $\mathbf{P}=0.2\mathbf{P_{qd}}+0.8\mathbf{P_{gc}}$, compared with baseline transition matrix $\mathbf{P_{qd}}$.

Fig.~\ref{fig:Expspec} shows the spectra of the transition matrices of all these configurations for the Na+ benchmark and the trends for other benchmarks are similar.
The green lines represent the eigenvalue decreasing trends of the transition matrices without perturbation, and the orange lines represent those trends for the transition matrices with random perturbation components.
As discussed in Section~\ref{Converge}, a transition matrix with a spectrum that has smaller eigenvalues will have faster convergence speed and smaller variance in the sampled circuits.
It can be easily observed that there are significant gaps between the spectra distribution w/ and w/o the random perturbation for both configurations plotted in (a) and (b).
This spectra difference leads to a substantial reduction in the standard deviation of the sampled quantum circuit. When $\mathbf{P_{qd}}$ contributes 0.4, the reduction of reduction of standard deviation is $26\%$ comparing $\sigma_{P_1^\prime}$ against $\sigma_{{P_1}}$. This reduction is $33\%$ comparing $\sigma_{P_2^\prime}$ against $\sigma_{{P_2}}$ when $\mathbf{P_{qd}}$ contributes 0.2.

\subsection{Impact of the Evolution Time}

\begin{figure*}[t]
    \centering
    \includegraphics[width=1\textwidth]{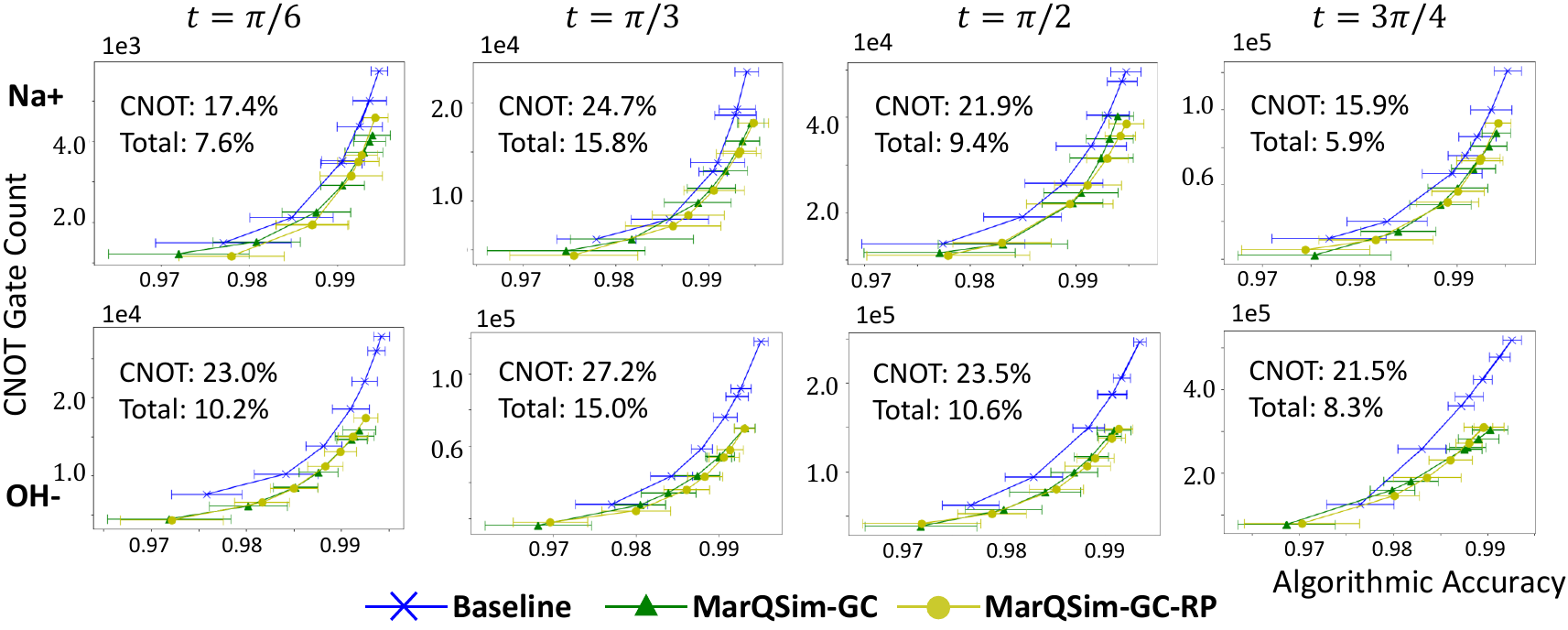}
    \vspace{-20pt}
    \caption{Compilation optimization effect with different evolution times}
    %\vspace{-10pt}
    \label{fig:evolution}
\end{figure*}

We also studied the impact of different evolution times. Two ions are chosen as the benchmarks: Na+ and OH- with qubit numbers equal to 8 and 10, respectively.
%The target accuracy $\varepsilon$ is set to be $0.1,0.067,0.05,0.04,0.033,0.0286,0.025$, and each data point is sampled 20 times. 
Four evolution times are selected: $t = \frac{\pi}{6}, \frac{\pi}{3}, \frac{\pi}{2}, \frac{3\pi}{4}$. 
The results are shown in Fig.~\ref{fig:evolution}. Similarly, CNOT gate reduction and total gate reduction for \textbf{\myCompilerName-GC-RP} are also listed in each sub-figure. % for the result.
For \textbf{\myCompilerName-GC}, the average CNOT gate reduction is 21.8\%, 24.7\%, 17.9\%, 24.8\% for $t = \frac{\pi}{6}, \frac{\pi}{3}, \frac{\pi}{2}, \frac{3\pi}{4}$ respectively. For \textbf{\myCompilerName-GC-RP}, the reduction is 20.2\%, 25.9\%, 22.7\%, 18.7\%. 
For the standard deviation, the result shows that \textbf{\myCompilerName-GC-RP} can reduce the standard deviation by 11.6\%, 8.5\%, 5.8\%, -3.4\% for $t = \frac{\pi}{6}, \frac{\pi}{3}, \frac{\pi}{2}, \frac{3\pi}{4}$ respectively, compared with the \textbf{\myCompilerName-GC}.
These results confirm that the benefit of \myCompilerNameSpace is not affected by longer time simulation.

\subsection{Compilation Time Analysis and Algorithm Complexity}

%\begin{wrapfigure}{R}{0.5\textwidth}
%    \centering
%    \includegraphics[width=0.50\textwidth]{fig/D-scalability.pdf}
    % \text{Left: Markov Chain Model}
    % \text{Right: Flow Model satisfies \eqref{equ:flow}}
%    \text{Evolution Time: 0.15s}
%        \caption{ Result for 14 qubit BeH2 
%    }
%    \label{fig:scalability}
%\end{wrapfigure}

%Due to the computational resources, it's difficult to present the result for 14 qubit molecule with standard parameter settings in Section~\ref{sec:overall}. Here we simulate BeH2 with an unfreezing core, for which the Hamiltonian contains 14 qubit Pauli string. The evolution time is set to be 0.15(s) and we only take 5 samples for each data point. See Figure~\ref{fig:scalability} for the result. The gate reduction is 38.3\%, 38.6\% for \textbf{\myCompilerName-GC} and \textbf{\myCompilerName-GC-RP} respectively. Since we only repeat every data point 5 times, we do not compare the standard deviation this time.

In addition, we studied the scalability of \myCompilerNameSpace by using the randomly generated Hamiltonian at different scales. 
We randomly generate Hamiltonians with 10, 20, and 30 qubits. For each qubit count, the number of Hamiltonian terms (Pauli strings) is selected to be 100, 500, and 1000. %are generated. 
We recorded the execution time to get the transition matrix and circuit, respectively. In this part, the evolution time is set to be $\frac{\pi}{4}$, and the target accuracy is set to be $0.05$. The compilation time results are shown in Table~\ref{tab:scale}.
The entire \myCompilerNameSpace compilation process has two major phases. The first phase is to generate the transition matrix via MCFP. 
The second phase is to sample from the Markov chain to assemble the final circuit.
As shown in Table~\ref{tab:scale}, the compilation time mainly depends on the number of Pauli strings.

\begin{table}[t]
  \centering
  % \vspace{-5pt}
  \caption{Compilation time analysis ($t=\frac{\pi}{4}$, $\epsilon = 0.05$)}
  \vspace{-5pt}
      \resizebox{0.98\columnwidth}{!}{ 
    \begin{tabular}{|c|c|c|c|c|c|c|c|}
    \hline
         \multirow{2}{*}{Qubit\#} & \multirow{2}{*}{Pauli String\#} & \multicolumn{3}{c|}{Transition Matrix Generation} & \multicolumn{3}{c|}{Circuit Generation}   \\
    \cline{3-8}
     &  & $ \enspace \enspace \mathbf{P_{qd}} \enspace \enspace$ & $\enspace \enspace \mathbf{P_{gc}} \enspace \enspace $ & $\mathbf{P_{rp}}$ & \textbf{Baseline} & \textbf{\myCompilerName-GC} & \textbf{\myCompilerName-GC-RP}\\
    \hline
    \multirow{3}{*}{10} & $100$ & 0.003 & 0.258 & 0.246 & 1.17 & 1.15 & 1.09 \\
\cline{2-8}
    & $500$ & 0.06 & 22.5 & 23.4 & 78.4 & 75.0 & 76.7 \\ \cline{2-8}
    & $1000$ & 0.225 & 169.3 & 168.4 & 585.0 & 584.5 & 580.7 \\
    \hline
    \multirow{3}{*}{20} &  $100$ & 0.002 & 0.271 & 0.240 & 1.12 & 1.11 & 1.10 \\
\cline{2-8}
    & $500$ & 0.059 & 22.3 & 22.4 & 70.2 & 71.7 & 71.5 \\ \cline{2-8}
    & $1000$ & 0.24 & 197.2 & 170.6 & 580.8 & 576.0 & 573.7 \\
    \hline
    \multirow{3}{*}{30} & $100$ & 0.003 & 0.276 & 0.268 & 1.18 & 1.18 & 1.16\\
\cline{2-8}
    & $500$ & 0.06 & 21.9 & 22.2 & 72.1 & 72.3 & 73.3 \\ \cline{2-8}
    & $1000$ & 0.232 & 179.4 & 173.8 & 597.3 & 565.0 & 581.1 \\
    \hline
    \end{tabular}%
    }%\vspace{5pt}
%    \text{Evolution Time: $\frac{\pi}{4}$s}
  \label{tab:scale}%
\end{table}%

%\Junyu{
%Here we also give the computational complexity for Markov compilation. 
We now estimate the computational complexity for the entire \myCompilerNameSpace compilation flow.
Suppose we have an $L$-qubit system with $n$ Pauli strings, $\lambda=\sum_{i=1}^n |h_i|$ is the summation of coefficient. The evolution time is $t$ and the target accuracy is $\epsilon$. First, we need to get the Markov transition matrix by solving an MCFP, in which $O(n^2L)$ time is needed to calculate the weight on each edge shown in Fig.~\ref{fig:flowgraph-problem}, and $O(n^3log(n)log(nL))$ is needed by using network simplex method solving the flow model~\cite{tarjan1997dynamic}. The second step is using the transition matrix to sample the circuit. The total number of sampling steps is $N=\left\lceil2{\lambda}^2t^2/\epsilon\right\rceil$. In each step, we need to sample from a discrete distribution of size $n$, which can take $log(n)$ time~\cite{bringmann2012efficient}.  Thus the overall computational complexity is $O(n^2L+n^3log(n)log(nL)+\frac{{(\lambda}t)^2}{\epsilon}log(n))$.
The entire compilation procedure can be accomplished in polynomial time, and the time-consuming sampling process can be easily paralleled.
\section{Discussion and Conclusion}

%\Junyu{}

The compiler infrastructures co-evolve with the new applications, languages, and hardware platforms. 
Today's classical computer hardware is highly sophisticated, and people usually write programs using high-level programming languages.
A classical compiler employs various intermediate representations at different levels and applies corresponding program transformation/optimizations.
We observe that the quantum computing area is in a transition with significant advancements in computation capability. 
Quantum devices with dozens or even hundreds of qubits of various technologies~\cite{moses2023race, IonQFort9:online, APreview62:online,chow2021ibm,wurtz2023aquila, ebadi2021quantum} are becoming publicly available. 
With hardware advancements, we are able to execute programs with larger state space and more operations more reliably, which is essential for quantum advantage.
Unfortunately, quantum software optimization and compilation techniques have not kept pace with hardware improvements. 
The low assembly-level quantum program compilation/optimizations for local gates/qubits~\cite{murali2019noise, li2019tackling, soeken2013white, nam2018automated, maslov2008quantum, miller2010lowering} have been well explored in the past decade while the high-level quantum compilation beyond the low-level gates and circuits is much less developed.

This paper aims to advance the area of high-level quantum compilation and proposes the \myCompilerNameSpace framework targeting the quantum Hamiltonian simulation subroutine.
Even though it is not universal, the quantum Hamiltonian simulation appears in quantum programs for many purposes far beyond simulating a physical system, and thus our framework can benefit a wide range of quantum applications.
More specifically, the proposed \myCompilerNameSpace framework reconciles the optimization opportunities from both deterministic Pauli string ordering and random sampling with new program intermediate representation and compilation optimization algorithms, exploring new program optimization design space that is, to the best of our knowledge, not reachable with previous approaches. 
We prove the correctness and error bound of our compilation algorithm and study the convergence speed and sampling variance via matrix spectra analysis, sharpening our understanding of the quantum simulation program compilation.
We anticipate that \myCompilerNameSpace will continue to benefit future quantum algorithms, as quantum simulation has been a longstanding principle in algorithm design over the past few decades.

%\todo{TWO EMPTY CITATION}
Looking forward, \myCompilerNameSpace is currently tuned to leverage the gate cancellation opportunities, while other optimizations in deterministic ordering may also be incorporated by combining more transition matrices. These opportunities can be grouping mutually commutative Hamiltonian terms together to reduce the approximation error~\cite{gui2020term}, or even further optimized by taking the underlying hardware architecture into consideration~\cite{li2022paulihedral}.
How to find transition matrices targeting these objectives and how to further modulate the spectrum of the transition matrices remain open problems.

Moreover, the high-level optimization principle can be applied to other quantum application domains or algorithmic subroutines.
There are several other significant techniques commonly used in quantum algorithm design, such as quantum phase estimation ~\cite{nielsen2010quantum}, amplitude amplification~\cite{brassard1997exact}, and quantum singular value transformation~\cite{gilyen2019quantum}, as well as promising application domains like machine learning~\cite{biamonte2017quantum} and optimization~\cite{abbas2023quantum}. Developing new high-level algorithmic quantum compiler optimizations for them remains an open problem that could be addressed in future work.

\bibliographystyle{ACM-Reference-Format}
\bibliography{ref}


\begin{thebibliography}{69}

%%% the Web address.

\ifx \showCODEN    \undefined \def \showCODEN     #1{\unskip}     \fi
\ifx \showDOI      \undefined \def \showDOI       #1{#1}\fi
\ifx \showISBNx    \undefined \def \showISBNx     #1{\unskip}     \fi
\ifx \showISBNxiii \undefined \def \showISBNxiii  #1{\unskip}     \fi
\ifx \showISSN     \undefined \def \showISSN      #1{\unskip}     \fi
\ifx \showLCCN     \undefined \def \showLCCN      #1{\unskip}     \fi
\ifx \shownote     \undefined \def \shownote      #1{#1}          \fi
\ifx \showarticletitle \undefined \def \showarticletitle #1{#1}   \fi
\ifx \showURL      \undefined \def \showURL       {\relax}        \fi
% invisible to TeX
\providecommand\bibfield[2]{#2}
\providecommand\bibinfo[2]{#2}
\providecommand\natexlab[1]{#1}
\providecommand\showeprint[2][]{arXiv:#2}

\bibitem[Abbas et~al\mbox{.}(2023)]%
        {abbas2023quantum}
\bibfield{author}{\bibinfo{person}{Amira Abbas}, \bibinfo{person}{Andris Ambainis}, \bibinfo{person}{Brandon Augustino}, \bibinfo{person}{Andreas Baertschi}, \bibinfo{person}{Harry Buhrman}, \bibinfo{person}{Carleton~James Coffrin}, \bibinfo{person}{Giorgio Cortiana}, \bibinfo{person}{Vedran Dunjko}, \bibinfo{person}{Daniel~J. Egger}, \bibinfo{person}{Bruce~G. Elmegreen}, \bibinfo{person}{Nicola Franco}, \bibinfo{person}{Filippo Fratini}, \bibinfo{person}{Bryce Fuller}, \bibinfo{person}{Julien Gacon}, \bibinfo{person}{Constantin Gonciulea}, \bibinfo{person}{Sander Gribling}, \bibinfo{person}{Swati Gupta}, \bibinfo{person}{Stuart Hadfield}, \bibinfo{person}{Raoul Heese}, \bibinfo{person}{Gerhard Kircher}, \bibinfo{person}{Thomas Kleinert}, \bibinfo{person}{Thorsten Koch}, \bibinfo{person}{Georgios Korpas}, \bibinfo{person}{Steve Lenk}, \bibinfo{person}{Jakub Marecek}, \bibinfo{person}{Vanio Markov}, \bibinfo{person}{Guglielmo Mazzola}, \bibinfo{person}{Stefano Mensa}, \bibinfo{person}{Naeimeh Mohseni},
  \bibinfo{person}{Giacomo Nannicini}, \bibinfo{person}{Corey O'Meara}, \bibinfo{person}{Elena Peña~Tapia}, \bibinfo{person}{Sebastian Pokutta}, \bibinfo{person}{Manuel Proissl}, \bibinfo{person}{Patrick Rebentrost}, \bibinfo{person}{Emre Sahin}, \bibinfo{person}{Benjamin C.~B. Symons}, \bibinfo{person}{Sabine Tornow}, \bibinfo{person}{Víctor Valls}, \bibinfo{person}{Stefan Woerner}, \bibinfo{person}{Mira~L. Wolf-Bauwens}, \bibinfo{person}{Jon Yard}, \bibinfo{person}{Sheir Yarkoni}, \bibinfo{person}{Dirk Zechiel}, \bibinfo{person}{Sergiy Zhuk}, {and} \bibinfo{person}{Christa Zoufal}.} \bibinfo{year}{2023}\natexlab{}.
\newblock \showarticletitle{Quantum Optimization: Potential, Challenges, and the Path Forward}.
\newblock  (\bibinfo{date}{12} \bibinfo{year}{2023}).
\newblock
\urldef\tempurl%
\url{https://doi.org/10.2172/2229681}
\showDOI{\tempurl}


\bibitem[Abrams and Lloyd(1999)]%
        {abrams1999quantum}
\bibfield{author}{\bibinfo{person}{Daniel~S. Abrams} {and} \bibinfo{person}{Seth Lloyd}.} \bibinfo{year}{1999}\natexlab{}.
\newblock \showarticletitle{Quantum Algorithm Providing Exponential Speed Increase for Finding Eigenvalues and Eigenvectors}.
\newblock \bibinfo{journal}{\emph{Phys. Rev. Lett.}}  \bibinfo{volume}{83} (\bibinfo{date}{Dec} \bibinfo{year}{1999}), \bibinfo{pages}{5162--5165}.
\newblock
Issue 24.
\urldef\tempurl%
\url{https://doi.org/10.1103/PhysRevLett.83.5162}
\showDOI{\tempurl}


\bibitem[Amy and Gheorghiu(2020)]%
        {amy2020staq}
\bibfield{author}{\bibinfo{person}{Matthew Amy} {and} \bibinfo{person}{Vlad Gheorghiu}.} \bibinfo{year}{2020}\natexlab{}.
\newblock \showarticletitle{staq{\textemdash}A full-stack quantum processing toolkit}.
\newblock \bibinfo{journal}{\emph{Quantum Science and Technology}} \bibinfo{volume}{5}, \bibinfo{number}{3} (\bibinfo{date}{jun} \bibinfo{year}{2020}), \bibinfo{pages}{034016}.
\newblock
\urldef\tempurl%
\url{https://doi.org/10.1088/2058-9565/ab9359}
\showDOI{\tempurl}


\bibitem[Berry et~al\mbox{.}(2006)]%
        {Berry2006}
\bibfield{author}{\bibinfo{person}{Dominic~W. Berry}, \bibinfo{person}{Graeme Ahokas}, \bibinfo{person}{Richard Cleve}, {and} \bibinfo{person}{Barry~C. Sanders}.} \bibinfo{year}{2006}\natexlab{}.
\newblock \showarticletitle{Efficient Quantum Algorithms for Simulating Sparse Hamiltonians}.
\newblock \bibinfo{journal}{\emph{Communications in Mathematical Physics}} \bibinfo{volume}{270}, \bibinfo{number}{2} (\bibinfo{date}{Dec.} \bibinfo{year}{2006}), \bibinfo{pages}{359--371}.
\newblock
\urldef\tempurl%
\url{https://doi.org/10.1007/s00220-006-0150-x}
\showDOI{\tempurl}


\bibitem[Biamonte et~al\mbox{.}(2017)]%
        {biamonte2017quantum}
\bibfield{author}{\bibinfo{person}{Jacob Biamonte}, \bibinfo{person}{Peter Wittek}, \bibinfo{person}{Nicola Pancotti}, \bibinfo{person}{Patrick Rebentrost}, \bibinfo{person}{Nathan Wiebe}, {and} \bibinfo{person}{Seth Lloyd}.} \bibinfo{year}{2017}\natexlab{}.
\newblock \showarticletitle{Quantum machine learning}.
\newblock \bibinfo{journal}{\emph{Nature}} \bibinfo{volume}{549}, \bibinfo{number}{7671} (\bibinfo{year}{2017}), \bibinfo{pages}{195--202}.
\newblock


\bibitem[Brassard and Hoyer(1997)]%
        {brassard1997exact}
\bibfield{author}{\bibinfo{person}{G. Brassard} {and} \bibinfo{person}{P. Hoyer}.} \bibinfo{year}{1997}\natexlab{}.
\newblock \showarticletitle{An exact quantum polynomial-time algorithm for Simon's problem}. In \bibinfo{booktitle}{\emph{Proceedings of the Fifth Israeli Symposium on Theory of Computing and Systems}}. \bibinfo{pages}{12--23}.
\newblock
\urldef\tempurl%
\url{https://doi.org/10.1109/ISTCS.1997.595153}
\showDOI{\tempurl}


\bibitem[Bringmann and Panagiotou(2012)]%
        {bringmann2012efficient}
\bibfield{author}{\bibinfo{person}{Karl Bringmann} {and} \bibinfo{person}{Konstantinos Panagiotou}.} \bibinfo{year}{2012}\natexlab{}.
\newblock \showarticletitle{Efficient sampling methods for discrete distributions}. In \bibinfo{booktitle}{\emph{Automata, Languages, and Programming: 39th International Colloquium, ICALP 2012, Warwick, UK, July 9-13, 2012, Proceedings, Part I 39}}. Springer, \bibinfo{pages}{133--144}.
\newblock


\bibitem[Campbell(2019)]%
        {campbell2019random}
\bibfield{author}{\bibinfo{person}{Earl Campbell}.} \bibinfo{year}{2019}\natexlab{}.
\newblock \showarticletitle{Random compiler for fast Hamiltonian simulation}.
\newblock \bibinfo{journal}{\emph{Physical review letters}} \bibinfo{volume}{123}, \bibinfo{number}{7} (\bibinfo{year}{2019}), \bibinfo{pages}{070503}.
\newblock


\bibitem[Childs et~al\mbox{.}(2019)]%
        {childs2019faster}
\bibfield{author}{\bibinfo{person}{Andrew~M Childs}, \bibinfo{person}{Aaron Ostrander}, {and} \bibinfo{person}{Yuan Su}.} \bibinfo{year}{2019}\natexlab{}.
\newblock \showarticletitle{Faster quantum simulation by randomization}.
\newblock \bibinfo{journal}{\emph{Quantum}}  \bibinfo{volume}{3} (\bibinfo{year}{2019}), \bibinfo{pages}{182}.
\newblock


\bibitem[Chow et~al\mbox{.}(2021)]%
        {chow2021ibm}
\bibfield{author}{\bibinfo{person}{Jerry Chow}, \bibinfo{person}{Oliver Dial}, {and} \bibinfo{person}{Jay Gambetta}.} \bibinfo{year}{2021}\natexlab{}.
\newblock \showarticletitle{IBM Quantum breaks the 100-qubit processor barrier}.
\newblock \bibinfo{journal}{\emph{IBM Research Blog}}  \bibinfo{volume}{2} (\bibinfo{year}{2021}).
\newblock


\bibitem[Cowtan et~al\mbox{.}(2020a)]%
        {cowtan2019phase}
\bibfield{author}{\bibinfo{person}{Alexander Cowtan}, \bibinfo{person}{Silas Dilkes}, \bibinfo{person}{Ross Duncan}, \bibinfo{person}{Will Simmons}, {and} \bibinfo{person}{Seyon Sivarajah}.} \bibinfo{year}{2020}\natexlab{a}.
\newblock \showarticletitle{Phase Gadget Synthesis for Shallow Circuits}.
\newblock \bibinfo{journal}{\emph{Electronic Proceedings in Theoretical Computer Science}}  \bibinfo{volume}{318} (\bibinfo{date}{May} \bibinfo{year}{2020}), \bibinfo{pages}{213–228}.
\newblock
\showISSN{2075-2180}
\urldef\tempurl%
\url{https://doi.org/10.4204/eptcs.318.13}
\showDOI{\tempurl}


\bibitem[Cowtan et~al\mbox{.}(2020b)]%
        {cowtan2020generic}
\bibfield{author}{\bibinfo{person}{Alexander Cowtan}, \bibinfo{person}{Will Simmons}, {and} \bibinfo{person}{Ross Duncan}.} \bibinfo{year}{2020}\natexlab{b}.
\newblock \showarticletitle{A Generic Compilation Strategy for the Unitary Coupled Cluster Ansatz}.
\newblock \bibinfo{journal}{\emph{arXiv preprint arXiv:2007.10515}} (\bibinfo{year}{2020}).
\newblock


\bibitem[Cross et~al\mbox{.}(2017)]%
        {cross2017open}
\bibfield{author}{\bibinfo{person}{Andrew~W Cross}, \bibinfo{person}{Lev~S Bishop}, \bibinfo{person}{John~A Smolin}, {and} \bibinfo{person}{Jay~M Gambetta}.} \bibinfo{year}{2017}\natexlab{}.
\newblock \showarticletitle{Open quantum assembly language}.
\newblock \bibinfo{journal}{\emph{arXiv preprint arXiv:1707.03429}} (\bibinfo{year}{2017}).
\newblock


\bibitem[Daley et~al\mbox{.}(2022)]%
        {daley2022practical}
\bibfield{author}{\bibinfo{person}{Andrew~J Daley}, \bibinfo{person}{Immanuel Bloch}, \bibinfo{person}{Christian Kokail}, \bibinfo{person}{Stuart Flannigan}, \bibinfo{person}{Natalie Pearson}, \bibinfo{person}{Matthias Troyer}, {and} \bibinfo{person}{Peter Zoller}.} \bibinfo{year}{2022}\natexlab{}.
\newblock \showarticletitle{Practical quantum advantage in quantum simulation}.
\newblock \bibinfo{journal}{\emph{Nature}} \bibinfo{volume}{607}, \bibinfo{number}{7920} (\bibinfo{year}{2022}), \bibinfo{pages}{667--676}.
\newblock


\bibitem[de~Griend and Duncan(2020)]%
        {de2020architecture}
\bibfield{author}{\bibinfo{person}{Arianne Meijer-van de Griend} {and} \bibinfo{person}{Ross Duncan}.} \bibinfo{year}{2020}\natexlab{}.
\newblock \showarticletitle{Architecture-aware synthesis of phase polynomials for NISQ devices}.
\newblock \bibinfo{journal}{\emph{arXiv preprint arXiv:2004.06052}} (\bibinfo{year}{2020}).
\newblock


\bibitem[Debnath et~al\mbox{.}(2016)]%
        {debnath2016demonstration}
\bibfield{author}{\bibinfo{person}{Shantanu Debnath}, \bibinfo{person}{Norbert~M Linke}, \bibinfo{person}{Caroline Figgatt}, \bibinfo{person}{Kevin~A Landsman}, \bibinfo{person}{Kevin Wright}, {and} \bibinfo{person}{Christopher Monroe}.} \bibinfo{year}{2016}\natexlab{}.
\newblock \showarticletitle{Demonstration of a small programmable quantum computer with atomic qubits}.
\newblock \bibinfo{journal}{\emph{Nature}} \bibinfo{volume}{536}, \bibinfo{number}{7614} (\bibinfo{year}{2016}), \bibinfo{pages}{63}.
\newblock


\bibitem[Developers(2023)]%
        {https://doi.org/10.5281/zenodo.7828767}
\bibfield{author}{\bibinfo{person}{The Qiskit~Nature Developers}.} \bibinfo{year}{2023}\natexlab{}.
\newblock \bibinfo{title}{Qiskit Nature 0.6.0}.
\newblock
\newblock
\urldef\tempurl%
\url{https://doi.org/10.5281/ZENODO.7828767}
\showDOI{\tempurl}


\bibitem[Ebadi et~al\mbox{.}(2021)]%
        {ebadi2021quantum}
\bibfield{author}{\bibinfo{person}{Sepehr Ebadi}, \bibinfo{person}{Tout~T. Wang}, \bibinfo{person}{Harry Levine}, \bibinfo{person}{Alexander Keesling}, \bibinfo{person}{Giulia Semeghini}, \bibinfo{person}{Ahmed Omran}, \bibinfo{person}{Dolev Bluvstein}, \bibinfo{person}{Rhine Samajdar}, \bibinfo{person}{Hannes Pichler}, \bibinfo{person}{Wen~Wei Ho}, \bibinfo{person}{Soonwon Choi}, \bibinfo{person}{Subir Sachdev}, \bibinfo{person}{Markus Greiner}, \bibinfo{person}{Vladan Vuleti{\'{c}}}, {and} \bibinfo{person}{Mikhail~D. Lukin}.} \bibinfo{year}{2021}\natexlab{}.
\newblock \showarticletitle{Quantum phases of matter on a 256-atom programmable quantum simulator}.
\newblock \bibinfo{journal}{\emph{Nature}} \bibinfo{volume}{595}, \bibinfo{number}{7866} (\bibinfo{date}{01 Jul} \bibinfo{year}{2021}), \bibinfo{pages}{227--232}.
\newblock
\showISSN{1476-4687}
\urldef\tempurl%
\url{https://doi.org/10.1038/s41586-021-03582-4}
\showDOI{\tempurl}


\bibitem[Feynman(1982)]%
        {feynman1982simulating}
\bibfield{author}{\bibinfo{person}{Richard~P Feynman}.} \bibinfo{year}{1982}\natexlab{}.
\newblock \showarticletitle{Simulating physics with computers}.
\newblock \bibinfo{journal}{\emph{Int. J. Theor. Phys}} \bibinfo{volume}{21}, \bibinfo{number}{6/7} (\bibinfo{year}{1982}).
\newblock


\bibitem[Gambetta and Sheldon(2019)]%
        {gambetta2019cramming}
\bibfield{author}{\bibinfo{person}{J Gambetta} {and} \bibinfo{person}{S Sheldon}.} \bibinfo{year}{2019}\natexlab{}.
\newblock \showarticletitle{Cramming more power into a quantum device}.
\newblock \bibinfo{journal}{\emph{IBM Research Blog}} (\bibinfo{year}{2019}).
\newblock
\urldef\tempurl%
\url{https://www.ibm.com/blogs/research/2019/03/power-quantum-device/}
\showURL{%
\tempurl}


\bibitem[Gily\'{e}n et~al\mbox{.}(2019)]%
        {gilyen2019quantum}
\bibfield{author}{\bibinfo{person}{Andr\'{a}s Gily\'{e}n}, \bibinfo{person}{Yuan Su}, \bibinfo{person}{Guang~Hao Low}, {and} \bibinfo{person}{Nathan Wiebe}.} \bibinfo{year}{2019}\natexlab{}.
\newblock \showarticletitle{Quantum singular value transformation and beyond: exponential improvements for quantum matrix arithmetics}. In \bibinfo{booktitle}{\emph{Proceedings of the 51st Annual ACM SIGACT Symposium on Theory of Computing}} (Phoenix, AZ, USA) \emph{(\bibinfo{series}{STOC 2019})}. \bibinfo{publisher}{Association for Computing Machinery}, \bibinfo{address}{New York, NY, USA}, \bibinfo{pages}{193–204}.
\newblock
\showISBNx{9781450367059}
\urldef\tempurl%
\url{https://doi.org/10.1145/3313276.3316366}
\showDOI{\tempurl}


\bibitem[Gui et~al\mbox{.}(2020)]%
        {gui2020term}
\bibfield{author}{\bibinfo{person}{Kaiwen Gui}, \bibinfo{person}{Teague Tomesh}, \bibinfo{person}{Pranav Gokhale}, \bibinfo{person}{Yunong Shi}, \bibinfo{person}{Frederic~T Chong}, \bibinfo{person}{Margaret Martonosi}, {and} \bibinfo{person}{Martin Suchara}.} \bibinfo{year}{2020}\natexlab{}.
\newblock \showarticletitle{Term grouping and travelling salesperson for digital quantum simulation}.
\newblock \bibinfo{journal}{\emph{arXiv preprint arXiv:2001.05983}} (\bibinfo{year}{2020}).
\newblock


\bibitem[Hagberg et~al\mbox{.}(2008)]%
        {hagberg2008exploring}
\bibfield{author}{\bibinfo{person}{Aric Hagberg}, \bibinfo{person}{Pieter Swart}, {and} \bibinfo{person}{Daniel S~Chult}.} \bibinfo{year}{2008}\natexlab{}.
\newblock \bibinfo{booktitle}{\emph{Exploring network structure, dynamics, and function using NetworkX}}.
\newblock \bibinfo{type}{{T}echnical {R}eport}. \bibinfo{institution}{Los Alamos National Lab.(LANL), Los Alamos, NM (United States)}.
\newblock


\bibitem[Haldar and Brady(2023)]%
        {haldar2023numerical}
\bibfield{author}{\bibinfo{person}{Stav Haldar} {and} \bibinfo{person}{Anthony~J Brady}.} \bibinfo{year}{2023}\natexlab{}.
\newblock \showarticletitle{A numerical study of measurement-induced phase transitions in the Sachdev-Ye-Kitaev model}.
\newblock \bibinfo{journal}{\emph{arXiv preprint arXiv:2301.05195}} (\bibinfo{year}{2023}).
\newblock


\bibitem[Harrow et~al\mbox{.}(2009)]%
        {harrow2009quantum}
\bibfield{author}{\bibinfo{person}{Aram~W. Harrow}, \bibinfo{person}{Avinatan Hassidim}, {and} \bibinfo{person}{Seth Lloyd}.} \bibinfo{year}{2009}\natexlab{}.
\newblock \showarticletitle{Quantum Algorithm for Linear Systems of Equations}.
\newblock \bibinfo{journal}{\emph{Phys. Rev. Lett.}}  \bibinfo{volume}{103} (\bibinfo{date}{Oct} \bibinfo{year}{2009}), \bibinfo{pages}{150502}.
\newblock
Issue 15.
\urldef\tempurl%
\url{https://doi.org/10.1103/PhysRevLett.103.150502}
\showDOI{\tempurl}


\bibitem[Hastings et~al\mbox{.}(2015)]%
        {hastings2015improving}
\bibfield{author}{\bibinfo{person}{Matthew~B. Hastings}, \bibinfo{person}{Dave Wecker}, \bibinfo{person}{Bela Bauer}, {and} \bibinfo{person}{Matthias Troyer}.} \bibinfo{year}{2015}\natexlab{}.
\newblock \showarticletitle{Improving Quantum Algorithms for Quantum Chemistry}.
\newblock \bibinfo{journal}{\emph{Quantum Info. Comput.}} \bibinfo{volume}{15}, \bibinfo{number}{1–2} (\bibinfo{date}{jan} \bibinfo{year}{2015}), \bibinfo{pages}{1–21}.
\newblock
\showISSN{1533-7146}
\urldef\tempurl%
\url{https://doi.org/10.5555/2685188.2685189}
\showDOI{\tempurl}


\bibitem[Hietala et~al\mbox{.}(2021)]%
        {hietala2023a}
\bibfield{author}{\bibinfo{person}{Kesha Hietala}, \bibinfo{person}{Robert Rand}, \bibinfo{person}{Shih-Han Hung}, \bibinfo{person}{Xiaodi Wu}, {and} \bibinfo{person}{Michael Hicks}.} \bibinfo{year}{2021}\natexlab{}.
\newblock \showarticletitle{A Verified Optimizer for Quantum Circuits}.
\newblock \bibinfo{journal}{\emph{Proc. ACM Program. Lang.}} \bibinfo{volume}{5}, \bibinfo{number}{POPL}, Article \bibinfo{articleno}{37} (\bibinfo{date}{jan} \bibinfo{year}{2021}), \bibinfo{numpages}{29}~pages.
\newblock
\urldef\tempurl%
\url{https://doi.org/10.1145/3434318}
\showDOI{\tempurl}


\bibitem[Javadi-Abhari et~al\mbox{.}(2017)]%
        {javadi2017optimized}
\bibfield{author}{\bibinfo{person}{Ali Javadi-Abhari}, \bibinfo{person}{Pranav Gokhale}, \bibinfo{person}{Adam Holmes}, \bibinfo{person}{Diana Franklin}, \bibinfo{person}{Kenneth~R. Brown}, \bibinfo{person}{Margaret Martonosi}, {and} \bibinfo{person}{Frederic~T. Chong}.} \bibinfo{year}{2017}\natexlab{}.
\newblock \showarticletitle{Optimized Surface Code Communication in Superconducting Quantum Computers}. In \bibinfo{booktitle}{\emph{Proceedings of the 50th Annual IEEE/ACM International Symposium on Microarchitecture}} (Cambridge, Massachusetts) \emph{(\bibinfo{series}{MICRO-50 '17})}. \bibinfo{publisher}{Association for Computing Machinery}, \bibinfo{address}{New York, NY, USA}, \bibinfo{pages}{692–705}.
\newblock
\showISBNx{9781450349529}
\urldef\tempurl%
\url{https://doi.org/10.1145/3123939.3123949}
\showDOI{\tempurl}


\bibitem[Javadi-Abhari et~al\mbox{.}(2024)]%
        {Qiskit}
\bibfield{author}{\bibinfo{person}{Ali Javadi-Abhari}, \bibinfo{person}{Matthew Treinish}, \bibinfo{person}{Kevin Krsulich}, \bibinfo{person}{Christopher~J Wood}, \bibinfo{person}{Jake Lishman}, \bibinfo{person}{Julien Gacon}, \bibinfo{person}{Simon Martiel}, \bibinfo{person}{Paul~D Nation}, \bibinfo{person}{Lev~S Bishop}, \bibinfo{person}{Andrew~W Cross}, \bibinfo{person}{Blake~R. Johnson}, {and} \bibinfo{person}{Jay~M. Gambetta}.} \bibinfo{year}{2024}\natexlab{}.
\newblock \showarticletitle{Quantum computing with Qiskit}.
\newblock \bibinfo{journal}{\emph{arXiv preprint arXiv:2405.08810}} (\bibinfo{year}{2024}).
\newblock


\bibitem[Jordan and Wigner(1928)]%
        {Jordan1928}
\bibfield{author}{\bibinfo{person}{P. Jordan} {and} \bibinfo{person}{E. Wigner}.} \bibinfo{year}{1928}\natexlab{}.
\newblock \showarticletitle{{\"U}ber das paulische {\"a}quivalenzverbot}.
\newblock \bibinfo{journal}{\emph{Zeitschrift f{\"u}r Physik}} \bibinfo{volume}{47}, \bibinfo{number}{9-10} (\bibinfo{date}{Sept.} \bibinfo{year}{1928}), \bibinfo{pages}{631--651}.
\newblock
\urldef\tempurl%
\url{https://doi.org/10.1007/bf01331938}
\showDOI{\tempurl}


\bibitem[Kelly({[n.\,d.]})]%
        {APreview62:online}
\bibfield{author}{\bibinfo{person}{Julian Kelly}.} \bibinfo{year}{[n.\,d.]}\natexlab{}.
\newblock \bibinfo{title}{A Preview of Bristlecone, Google’s New Quantum Processor – Google Research Blog}.
\newblock \bibinfo{howpublished}{\url{https://ai.googleblog.com/2018/03/a-preview-of-bristlecone-googles-new.html}}.
\newblock
\newblock
\shownote{(Accessed on 07/22/2023)}.


\bibitem[Khammassi et~al\mbox{.}(2021)]%
        {khammassi2020openql}
\bibfield{author}{\bibinfo{person}{N. Khammassi}, \bibinfo{person}{I. Ashraf}, \bibinfo{person}{J.~V. Someren}, \bibinfo{person}{R. Nane}, \bibinfo{person}{A.~M. Krol}, \bibinfo{person}{M.~A. Rol}, \bibinfo{person}{L. Lao}, \bibinfo{person}{K. Bertels}, {and} \bibinfo{person}{C.~G. Almudever}.} \bibinfo{year}{2021}\natexlab{}.
\newblock \showarticletitle{OpenQL: A Portable Quantum Programming Framework for Quantum Accelerators}.
\newblock \bibinfo{journal}{\emph{J. Emerg. Technol. Comput. Syst.}} \bibinfo{volume}{18}, \bibinfo{number}{1}, Article \bibinfo{articleno}{13} (\bibinfo{date}{dec} \bibinfo{year}{2021}), \bibinfo{numpages}{24}~pages.
\newblock
\showISSN{1550-4832}
\urldef\tempurl%
\url{https://doi.org/10.1145/3474222}
\showDOI{\tempurl}


\bibitem[Kir{\'a}ly and Kov{\'a}cs(2012)]%
        {kiraly2012efficient}
\bibfield{author}{\bibinfo{person}{Zolt{\'a}n Kir{\'a}ly} {and} \bibinfo{person}{P{\'e}ter Kov{\'a}cs}.} \bibinfo{year}{2012}\natexlab{}.
\newblock \showarticletitle{Efficient implementations of minimum-cost flow algorithms}.
\newblock \bibinfo{journal}{\emph{arXiv preprint arXiv:1207.6381}} (\bibinfo{year}{2012}).
\newblock


\bibitem[Kissinger and van~de Wetering(2020)]%
        {kissinger2019pyzx}
\bibfield{author}{\bibinfo{person}{Aleks Kissinger} {and} \bibinfo{person}{John van~de Wetering}.} \bibinfo{year}{2020}\natexlab{}.
\newblock \showarticletitle{PyZX: Large Scale Automated Diagrammatic Reasoning}.
\newblock \bibinfo{journal}{\emph{Electronic Proceedings in Theoretical Computer Science}}  \bibinfo{volume}{318} (\bibinfo{date}{May} \bibinfo{year}{2020}), \bibinfo{pages}{229–241}.
\newblock
\showISSN{2075-2180}
\urldef\tempurl%
\url{https://doi.org/10.4204/eptcs.318.14}
\showDOI{\tempurl}


\bibitem[Lehmann et~al\mbox{.}(2022)]%
        {lehmann2022vyzx}
\bibfield{author}{\bibinfo{person}{Adrian Lehmann}, \bibinfo{person}{Ben Caldwell}, {and} \bibinfo{person}{Robert Rand}.} \bibinfo{year}{2022}\natexlab{}.
\newblock \showarticletitle{VyZX: A Vision for Verifying the ZX Calculus}.
\newblock \bibinfo{journal}{\emph{arXiv preprint arXiv:2205.05781}} (\bibinfo{year}{2022}).
\newblock


\bibitem[Li et~al\mbox{.}(2019)]%
        {li2019tackling}
\bibfield{author}{\bibinfo{person}{Gushu Li}, \bibinfo{person}{Yufei Ding}, {and} \bibinfo{person}{Yuan Xie}.} \bibinfo{year}{2019}\natexlab{}.
\newblock \showarticletitle{Tackling the Qubit Mapping Problem for NISQ-Era Quantum Devices}. In \bibinfo{booktitle}{\emph{Proceedings of the Twenty-Fourth International Conference on Architectural Support for Programming Languages and Operating Systems}} (Providence, RI, USA) \emph{(\bibinfo{series}{ASPLOS '19})}. \bibinfo{publisher}{Association for Computing Machinery}, \bibinfo{address}{New York, NY, USA}, \bibinfo{pages}{1001–1014}.
\newblock
\showISBNx{9781450362405}
\urldef\tempurl%
\url{https://doi.org/10.1145/3297858.3304023}
\showDOI{\tempurl}


\bibitem[Li et~al\mbox{.}(2022)]%
        {li2022paulihedral}
\bibfield{author}{\bibinfo{person}{Gushu Li}, \bibinfo{person}{Anbang Wu}, \bibinfo{person}{Yunong Shi}, \bibinfo{person}{Ali Javadi-Abhari}, \bibinfo{person}{Yufei Ding}, {and} \bibinfo{person}{Yuan Xie}.} \bibinfo{year}{2022}\natexlab{}.
\newblock \showarticletitle{Paulihedral: A Generalized Block-Wise Compiler Optimization Framework for Quantum Simulation Kernels}. In \bibinfo{booktitle}{\emph{Proceedings of the 27th ACM International Conference on Architectural Support for Programming Languages and Operating Systems}} (Lausanne, Switzerland) \emph{(\bibinfo{series}{ASPLOS '22})}. \bibinfo{publisher}{Association for Computing Machinery}, \bibinfo{address}{New York, NY, USA}, \bibinfo{pages}{554–569}.
\newblock
\showISBNx{9781450392051}
\urldef\tempurl%
\url{https://doi.org/10.1145/3503222.3507715}
\showDOI{\tempurl}


\bibitem[Linke et~al\mbox{.}(2017)]%
        {linke2017experimental}
\bibfield{author}{\bibinfo{person}{Norbert~M Linke}, \bibinfo{person}{Dmitri Maslov}, \bibinfo{person}{Martin Roetteler}, \bibinfo{person}{Shantanu Debnath}, \bibinfo{person}{Caroline Figgatt}, \bibinfo{person}{Kevin~A Landsman}, \bibinfo{person}{Kenneth Wright}, {and} \bibinfo{person}{Christopher Monroe}.} \bibinfo{year}{2017}\natexlab{}.
\newblock \showarticletitle{Experimental comparison of two quantum computing architectures}.
\newblock \bibinfo{journal}{\emph{Proceedings of the National Academy of Sciences}} \bibinfo{volume}{114}, \bibinfo{number}{13} (\bibinfo{year}{2017}), \bibinfo{pages}{3305--3310}.
\newblock


\bibitem[Lloyd(1996)]%
        {lloyd1996universal}
\bibfield{author}{\bibinfo{person}{Seth Lloyd}.} \bibinfo{year}{1996}\natexlab{}.
\newblock \showarticletitle{Universal quantum simulators}.
\newblock \bibinfo{journal}{\emph{Science}} (\bibinfo{year}{1996}), \bibinfo{pages}{1073--1078}.
\newblock


\bibitem[Lloyd et~al\mbox{.}(2014)]%
        {lloyd2014quantum}
\bibfield{author}{\bibinfo{person}{Seth Lloyd}, \bibinfo{person}{Masoud Mohseni}, {and} \bibinfo{person}{Patrick Rebentrost}.} \bibinfo{year}{2014}\natexlab{}.
\newblock \showarticletitle{Quantum principal component analysis}.
\newblock \bibinfo{journal}{\emph{Nature Physics}} \bibinfo{volume}{10}, \bibinfo{number}{9} (\bibinfo{year}{2014}), \bibinfo{pages}{631--633}.
\newblock


\bibitem[Maslov(2016)]%
        {maslov2016optimal}
\bibfield{author}{\bibinfo{person}{Dmitri Maslov}.} \bibinfo{year}{2016}\natexlab{}.
\newblock \showarticletitle{Optimal and asymptotically optimal NCT reversible circuits by the gate types}.
\newblock \bibinfo{journal}{\emph{arXiv preprint arXiv:1602.02627}} (\bibinfo{year}{2016}).
\newblock


\bibitem[Maslov et~al\mbox{.}(2008)]%
        {maslov2008quantum}
\bibfield{author}{\bibinfo{person}{Dmitri Maslov}, \bibinfo{person}{Gerhard~W Dueck}, \bibinfo{person}{D~Michael Miller}, {and} \bibinfo{person}{Camille Negrevergne}.} \bibinfo{year}{2008}\natexlab{}.
\newblock \showarticletitle{Quantum circuit simplification and level compaction}.
\newblock \bibinfo{journal}{\emph{IEEE Transactions on Computer-Aided Design of Integrated Circuits and Systems}} \bibinfo{volume}{27}, \bibinfo{number}{3} (\bibinfo{year}{2008}), \bibinfo{pages}{436--444}.
\newblock


\bibitem[McCaskey and Nguyen(2021)]%
        {mccaskey2021mlir}
\bibfield{author}{\bibinfo{person}{A. McCaskey} {and} \bibinfo{person}{T. Nguyen}.} \bibinfo{year}{2021}\natexlab{}.
\newblock \showarticletitle{A MLIR Dialect for Quantum Assembly Languages}. In \bibinfo{booktitle}{\emph{2021 IEEE International Conference on Quantum Computing and Engineering (QCE)}}. \bibinfo{publisher}{IEEE Computer Society}, \bibinfo{address}{Los Alamitos, CA, USA}, \bibinfo{pages}{255--264}.
\newblock
\urldef\tempurl%
\url{https://doi.org/10.1109/QCE52317.2021.00043}
\showDOI{\tempurl}


\bibitem[Miller and Sasanian(2010)]%
        {miller2010lowering}
\bibfield{author}{\bibinfo{person}{D~Michael Miller} {and} \bibinfo{person}{Zahra Sasanian}.} \bibinfo{year}{2010}\natexlab{}.
\newblock \showarticletitle{Lowering the quantum gate cost of reversible circuits}. In \bibinfo{booktitle}{\emph{2010 53rd IEEE International Midwest Symposium on Circuits and Systems}}. IEEE, \bibinfo{pages}{260--263}.
\newblock


\bibitem[Moses et~al\mbox{.}(2023)]%
        {moses2023race}
\bibfield{author}{\bibinfo{person}{S.~A. Moses}, \bibinfo{person}{C.~H. Baldwin}, \bibinfo{person}{M.~S. Allman}, \bibinfo{person}{R. Ancona}, \bibinfo{person}{L. Ascarrunz}, \bibinfo{person}{C. Barnes}, \bibinfo{person}{J. Bartolotta}, \bibinfo{person}{B. Bjork}, \bibinfo{person}{P. Blanchard}, \bibinfo{person}{M. Bohn}, \bibinfo{person}{J.~G. Bohnet}, \bibinfo{person}{N.~C. Brown}, \bibinfo{person}{N.~Q. Burdick}, \bibinfo{person}{W.~C. Burton}, \bibinfo{person}{S.~L. Campbell}, \bibinfo{person}{J.~P. Campora~III au2}, \bibinfo{person}{C. Carron}, \bibinfo{person}{J. Chambers}, \bibinfo{person}{J.~W. Chan}, \bibinfo{person}{Y.~H. Chen}, \bibinfo{person}{A. Chernoguzov}, \bibinfo{person}{E. Chertkov}, \bibinfo{person}{J. Colina}, \bibinfo{person}{J.~P. Curtis}, \bibinfo{person}{R. Daniel}, \bibinfo{person}{M. DeCross}, \bibinfo{person}{D. Deen}, \bibinfo{person}{C. Delaney}, \bibinfo{person}{J.~M. Dreiling}, \bibinfo{person}{C.~T. Ertsgaard}, \bibinfo{person}{J. Esposito}, \bibinfo{person}{B. Estey},
  \bibinfo{person}{M. Fabrikant}, \bibinfo{person}{C. Figgatt}, \bibinfo{person}{C. Foltz}, \bibinfo{person}{M. Foss-Feig}, \bibinfo{person}{D. Francois}, \bibinfo{person}{J.~P. Gaebler}, \bibinfo{person}{T.~M. Gatterman}, \bibinfo{person}{C.~N. Gilbreth}, \bibinfo{person}{J. Giles}, \bibinfo{person}{E. Glynn}, \bibinfo{person}{A. Hall}, \bibinfo{person}{A.~M. Hankin}, \bibinfo{person}{A. Hansen}, \bibinfo{person}{D. Hayes}, \bibinfo{person}{B. Higashi}, \bibinfo{person}{I.~M. Hoffman}, \bibinfo{person}{B. Horning}, \bibinfo{person}{J.~J. Hout}, \bibinfo{person}{R. Jacobs}, \bibinfo{person}{J. Johansen}, \bibinfo{person}{L. Jones}, \bibinfo{person}{J. Karcz}, \bibinfo{person}{T. Klein}, \bibinfo{person}{P. Lauria}, \bibinfo{person}{P. Lee}, \bibinfo{person}{D. Liefer}, \bibinfo{person}{C. Lytle}, \bibinfo{person}{S.~T. Lu}, \bibinfo{person}{D. Lucchetti}, \bibinfo{person}{A. Malm}, \bibinfo{person}{M. Matheny}, \bibinfo{person}{B. Mathewson}, \bibinfo{person}{K. Mayer}, \bibinfo{person}{D.~B. Miller},
  \bibinfo{person}{M. Mills}, \bibinfo{person}{B. Neyenhuis}, \bibinfo{person}{L. Nugent}, \bibinfo{person}{S. Olson}, \bibinfo{person}{J. Parks}, \bibinfo{person}{G.~N. Price}, \bibinfo{person}{Z. Price}, \bibinfo{person}{M. Pugh}, \bibinfo{person}{A. Ransford}, \bibinfo{person}{A.~P. Reed}, \bibinfo{person}{C. Roman}, \bibinfo{person}{M. Rowe}, \bibinfo{person}{C. Ryan-Anderson}, \bibinfo{person}{S. Sanders}, \bibinfo{person}{J. Sedlacek}, \bibinfo{person}{P. Shevchuk}, \bibinfo{person}{P. Siegfried}, \bibinfo{person}{T. Skripka}, \bibinfo{person}{B. Spaun}, \bibinfo{person}{R.~T. Sprenkle}, \bibinfo{person}{R.~P. Stutz}, \bibinfo{person}{M. Swallows}, \bibinfo{person}{R.~I. Tobey}, \bibinfo{person}{A. Tran}, \bibinfo{person}{T. Tran}, \bibinfo{person}{E. Vogt}, \bibinfo{person}{C. Volin}, \bibinfo{person}{J. Walker}, \bibinfo{person}{A.~M. Zolot}, {and} \bibinfo{person}{J.~M. Pino}.} \bibinfo{year}{2023}\natexlab{}.
\newblock \showarticletitle{A race track trapped-ion quantum processor}.
\newblock \bibinfo{journal}{\emph{arXiv preprint arXiv:2305.03828}} (\bibinfo{year}{2023}).
\newblock


\bibitem[Murali et~al\mbox{.}(2019)]%
        {murali2019noise}
\bibfield{author}{\bibinfo{person}{Prakash Murali}, \bibinfo{person}{Jonathan~M. Baker}, \bibinfo{person}{Ali Javadi-Abhari}, \bibinfo{person}{Frederic~T. Chong}, {and} \bibinfo{person}{Margaret Martonosi}.} \bibinfo{year}{2019}\natexlab{}.
\newblock \showarticletitle{Noise-Adaptive Compiler Mappings for Noisy Intermediate-Scale Quantum Computers}. In \bibinfo{booktitle}{\emph{Proceedings of the Twenty-Fourth International Conference on Architectural Support for Programming Languages and Operating Systems}} (Providence, RI, USA) \emph{(\bibinfo{series}{ASPLOS '19})}. \bibinfo{publisher}{Association for Computing Machinery}, \bibinfo{address}{New York, NY, USA}, \bibinfo{pages}{1015–1029}.
\newblock
\showISBNx{9781450362405}
\urldef\tempurl%
\url{https://doi.org/10.1145/3297858.3304075}
\showDOI{\tempurl}


\bibitem[Nam et~al\mbox{.}(2018)]%
        {nam2018automated}
\bibfield{author}{\bibinfo{person}{Yunseong Nam}, \bibinfo{person}{Neil~J Ross}, \bibinfo{person}{Yuan Su}, \bibinfo{person}{Andrew~M Childs}, {and} \bibinfo{person}{Dmitri Maslov}.} \bibinfo{year}{2018}\natexlab{}.
\newblock \showarticletitle{Automated optimization of large quantum circuits with continuous parameters}.
\newblock \bibinfo{journal}{\emph{npj Quantum Information}} \bibinfo{volume}{4}, \bibinfo{number}{1} (\bibinfo{year}{2018}), \bibinfo{pages}{1--12}.
\newblock


\bibitem[Nielsen and Chuang(2010)]%
        {nielsen2010quantum}
\bibfield{author}{\bibinfo{person}{M.A. Nielsen} {and} \bibinfo{person}{I.L. Chuang}.} \bibinfo{year}{2010}\natexlab{}.
\newblock \bibinfo{booktitle}{\emph{Quantum computation and quantum information}}.
\newblock \bibinfo{publisher}{Cambridge university press}.
\newblock


\bibitem[Norris(1998)]%
        {norris1998markov}
\bibfield{author}{\bibinfo{person}{James~R Norris}.} \bibinfo{year}{1998}\natexlab{}.
\newblock \bibinfo{booktitle}{\emph{Markov chains}}.
\newblock Number~2. \bibinfo{publisher}{Cambridge university press}.
\newblock


\bibitem[NVIDIA et~al\mbox{.}(2020)]%
        {cuda}
\bibfield{author}{\bibinfo{person}{NVIDIA}, \bibinfo{person}{P{\'e}ter Vingelmann}, {and} \bibinfo{person}{Frank~H.P. Fitzek}.} \bibinfo{year}{2020}\natexlab{}.
\newblock \bibinfo{title}{CUDA, release: 10.2.89}.
\newblock
\newblock
\urldef\tempurl%
\url{https://developer.nvidia.com/cuda-toolkit}
\showURL{%
\tempurl}


\bibitem[Ouyang et~al\mbox{.}(2020)]%
        {ouyang2020compilation}
\bibfield{author}{\bibinfo{person}{Yingkai Ouyang}, \bibinfo{person}{David~R White}, {and} \bibinfo{person}{Earl~T Campbell}.} \bibinfo{year}{2020}\natexlab{}.
\newblock \showarticletitle{Compilation by stochastic Hamiltonian sparsification}.
\newblock \bibinfo{journal}{\emph{Quantum}}  \bibinfo{volume}{4} (\bibinfo{year}{2020}), \bibinfo{pages}{235}.
\newblock


\bibitem[Paszke et~al\mbox{.}(2019)]%
        {NEURIPS2019_9015}
\bibfield{author}{\bibinfo{person}{Adam Paszke}, \bibinfo{person}{Sam Gross}, \bibinfo{person}{Francisco Massa}, \bibinfo{person}{Adam Lerer}, \bibinfo{person}{James Bradbury}, \bibinfo{person}{Gregory Chanan}, \bibinfo{person}{Trevor Killeen}, \bibinfo{person}{Zeming Lin}, \bibinfo{person}{Natalia Gimelshein}, \bibinfo{person}{Luca Antiga}, \bibinfo{person}{Alban Desmaison}, \bibinfo{person}{Andreas Kopf}, \bibinfo{person}{Edward Yang}, \bibinfo{person}{Zachary DeVito}, \bibinfo{person}{Martin Raison}, \bibinfo{person}{Alykhan Tejani}, \bibinfo{person}{Sasank Chilamkurthy}, \bibinfo{person}{Benoit Steiner}, \bibinfo{person}{Lu Fang}, \bibinfo{person}{Junjie Bai}, {and} \bibinfo{person}{Soumith Chintala}.} \bibinfo{year}{2019}\natexlab{}.
\newblock \showarticletitle{PyTorch: An Imperative Style, High-Performance Deep Learning Library}.
\newblock In \bibinfo{booktitle}{\emph{Advances in Neural Information Processing Systems 32}}. \bibinfo{publisher}{Curran Associates, Inc.}, \bibinfo{pages}{8024--8035}.
\newblock
\urldef\tempurl%
\url{http://papers.neurips.cc/paper/9015-pytorch-an-imperative-style-high-performance-deep-learning-library.pdf}
\showURL{%
\tempurl}


\bibitem[Peng et~al\mbox{.}(2023)]%
        {peng2023simuq}
\bibfield{author}{\bibinfo{person}{Yuxiang Peng}, \bibinfo{person}{Jacob Young}, \bibinfo{person}{Pengyu Liu}, {and} \bibinfo{person}{Xiaodi Wu}.} \bibinfo{year}{2023}\natexlab{}.
\newblock \showarticletitle{SimuQ: A Domain-Specific Language for Quantum Simulation with Analog Compilation}.
\newblock \bibinfo{journal}{\emph{arXiv preprint arXiv:2303.02775}} (\bibinfo{date}{March} \bibinfo{year}{2023}).
\newblock


\bibitem[Rebentrost et~al\mbox{.}(2014)]%
        {rebentrost2014quantum}
\bibfield{author}{\bibinfo{person}{Patrick Rebentrost}, \bibinfo{person}{Masoud Mohseni}, {and} \bibinfo{person}{Seth Lloyd}.} \bibinfo{year}{2014}\natexlab{}.
\newblock \showarticletitle{Quantum Support Vector Machine for Big Data Classification}.
\newblock \bibinfo{journal}{\emph{Phys. Rev. Lett.}}  \bibinfo{volume}{113} (\bibinfo{date}{Sep} \bibinfo{year}{2014}), \bibinfo{pages}{130503}.
\newblock
Issue 13.
\urldef\tempurl%
\url{https://doi.org/10.1103/PhysRevLett.113.130503}
\showDOI{\tempurl}


\bibitem[Samuel~Karlin(1974)]%
        {StochasticProcess}
\bibfield{author}{\bibinfo{person}{Howard M.~Taylor Samuel~Karlin}.} \bibinfo{year}{1974}\natexlab{}.
\newblock \bibinfo{booktitle}{\emph{A First Course in Stochastic Process}}.
\newblock \bibinfo{publisher}{ACADEMIC PRESS, INC. (LONDON) LTD.}, \bibinfo{address}{London}.
\newblock


\bibitem[Sivarajah et~al\mbox{.}(2020)]%
        {sivarajah2020t}
\bibfield{author}{\bibinfo{person}{Seyon Sivarajah}, \bibinfo{person}{Silas Dilkes}, \bibinfo{person}{Alexander Cowtan}, \bibinfo{person}{Will Simmons}, \bibinfo{person}{Alec Edgington}, {and} \bibinfo{person}{Ross Duncan}.} \bibinfo{year}{2020}\natexlab{}.
\newblock \showarticletitle{$t \vert ket\rangle$: a retargetable compiler for {NISQ} devices}.
\newblock \bibinfo{journal}{\emph{Quantum Science and Technology}} \bibinfo{volume}{6}, \bibinfo{number}{1} (\bibinfo{date}{nov} \bibinfo{year}{2020}), \bibinfo{pages}{014003}.
\newblock
\urldef\tempurl%
\url{https://doi.org/10.1088/2058-9565/ab8e92}
\showDOI{\tempurl}


\bibitem[Smith et~al\mbox{.}(2016)]%
        {smith2016practical}
\bibfield{author}{\bibinfo{person}{Robert~S Smith}, \bibinfo{person}{Michael~J Curtis}, {and} \bibinfo{person}{William~J Zeng}.} \bibinfo{year}{2016}\natexlab{}.
\newblock \showarticletitle{A practical quantum instruction set architecture}.
\newblock \bibinfo{journal}{\emph{arXiv preprint arXiv:1608.03355}} (\bibinfo{year}{2016}).
\newblock


\bibitem[Soeken and Thomsen(2013)]%
        {soeken2013white}
\bibfield{author}{\bibinfo{person}{Mathias Soeken} {and} \bibinfo{person}{Michael~Kirkedal Thomsen}.} \bibinfo{year}{2013}\natexlab{}.
\newblock \showarticletitle{White Dots Do Matter: Rewriting Reversible Logic Circuits}. In \bibinfo{booktitle}{\emph{Proceedings of the 5th International Conference on Reversible Computation}} (Victoria, BC, Canada) \emph{(\bibinfo{series}{RC'13})}. \bibinfo{publisher}{Springer-Verlag}, \bibinfo{address}{Berlin, Heidelberg}, \bibinfo{pages}{196–208}.
\newblock
\showISBNx{9783642389856}
\urldef\tempurl%
\url{https://doi.org/10.1007/978-3-642-38986-3_16}
\showDOI{\tempurl}


\bibitem[Staff(2023)]%
        {IonQFort9:online}
\bibfield{author}{\bibinfo{person}{IonQ Staff}.} \bibinfo{year}{2023}\natexlab{}.
\newblock \bibinfo{title}{IonQ Forte: The First Software-Configurable Quantum Computer}.
\newblock \bibinfo{howpublished}{\url{https://ionq.com/resources/ionq-forte-first-configurable-quantum-computer}}.
\newblock
\newblock
\shownote{(Accessed on 07/22/2023)}.


\bibitem[Sun et~al\mbox{.}(2018)]%
        {PySCF}
\bibfield{author}{\bibinfo{person}{Qiming Sun}, \bibinfo{person}{Timothy~C. Berkelbach}, \bibinfo{person}{Nick~S. Blunt}, \bibinfo{person}{George~H. Booth}, \bibinfo{person}{Sheng Guo}, \bibinfo{person}{Zhendong Li}, \bibinfo{person}{Junzi Liu}, \bibinfo{person}{James~D. McClain}, \bibinfo{person}{Elvira~R. Sayfutyarova}, \bibinfo{person}{Sandeep Sharma}, \bibinfo{person}{Sebastian Wouters}, {and} \bibinfo{person}{Garnet Kin-Lic Chan}.} \bibinfo{year}{2018}\natexlab{}.
\newblock \showarticletitle{PySCF: the Python-based simulations of chemistry framework}.
\newblock \bibinfo{journal}{\emph{WIREs Computational Molecular Science}} \bibinfo{volume}{8}, \bibinfo{number}{1} (\bibinfo{year}{2018}), \bibinfo{pages}{e1340}.
\newblock
\urldef\tempurl%
\url{https://doi.org/10.1002/wcms.1340}
\showDOI{\tempurl}
\showeprint{https://wires.onlinelibrary.wiley.com/doi/pdf/10.1002/wcms.1340}


\bibitem[Suzuki(1990)]%
        {suzuki1990fractal}
\bibfield{author}{\bibinfo{person}{Masuo Suzuki}.} \bibinfo{year}{1990}\natexlab{}.
\newblock \showarticletitle{Fractal decomposition of exponential operators with applications to many-body theories and {Monte Carlo} simulations}.
\newblock \bibinfo{journal}{\emph{Physics Letters A}} \bibinfo{volume}{146}, \bibinfo{number}{6} (\bibinfo{year}{1990}), \bibinfo{pages}{319 -- 323}.
\newblock
\showISSN{0375-9601}
\urldef\tempurl%
\url{https://doi.org/10.1016/0375-9601(90)90962-N}
\showDOI{\tempurl}


\bibitem[Suzuki(1991)]%
        {Suzuki1991}
\bibfield{author}{\bibinfo{person}{Masuo Suzuki}.} \bibinfo{year}{1991}\natexlab{}.
\newblock \showarticletitle{General theory of fractal path integrals with applications to many-body theories and statistical physics}.
\newblock \bibinfo{journal}{\emph{J. Math. Phys.}} \bibinfo{volume}{32}, \bibinfo{number}{2} (\bibinfo{date}{Feb.} \bibinfo{year}{1991}), \bibinfo{pages}{400--407}.
\newblock
\urldef\tempurl%
\url{https://doi.org/10.1063/1.529425}
\showDOI{\tempurl}


\bibitem[Tao et~al\mbox{.}(2022)]%
        {tao2022giallar}
\bibfield{author}{\bibinfo{person}{Runzhou Tao}, \bibinfo{person}{Yunong Shi}, \bibinfo{person}{Jianan Yao}, \bibinfo{person}{Xupeng Li}, \bibinfo{person}{Ali Javadi-Abhari}, \bibinfo{person}{Andrew~W. Cross}, \bibinfo{person}{Frederic~T. Chong}, {and} \bibinfo{person}{Ronghui Gu}.} \bibinfo{year}{2022}\natexlab{}.
\newblock \showarticletitle{Giallar: Push-Button Verification for the Qiskit Quantum Compiler}. In \bibinfo{booktitle}{\emph{Proceedings of the 43rd ACM SIGPLAN International Conference on Programming Language Design and Implementation}} (San Diego, CA, USA) \emph{(\bibinfo{series}{PLDI 2022})}. \bibinfo{publisher}{Association for Computing Machinery}, \bibinfo{address}{New York, NY, USA}, \bibinfo{pages}{641–656}.
\newblock
\showISBNx{9781450392655}
\urldef\tempurl%
\url{https://doi.org/10.1145/3519939.3523431}
\showDOI{\tempurl}


\bibitem[Tarjan(1997)]%
        {tarjan1997dynamic}
\bibfield{author}{\bibinfo{person}{Robert~E Tarjan}.} \bibinfo{year}{1997}\natexlab{}.
\newblock \showarticletitle{Dynamic trees as search trees via euler tours, applied to the network simplex algorithm}.
\newblock \bibinfo{journal}{\emph{Mathematical Programming}} \bibinfo{volume}{78}, \bibinfo{number}{2} (\bibinfo{year}{1997}), \bibinfo{pages}{169--177}.
\newblock


\bibitem[Trotter(1959)]%
        {trotter1959product}
\bibfield{author}{\bibinfo{person}{Hale~F Trotter}.} \bibinfo{year}{1959}\natexlab{}.
\newblock \showarticletitle{On the product of semi-groups of operators}.
\newblock \bibinfo{journal}{\emph{Proc. Amer. Math. Soc.}} \bibinfo{volume}{10}, \bibinfo{number}{4} (\bibinfo{year}{1959}), \bibinfo{pages}{545--551}.
\newblock


\bibitem[Van Den~Berg and Temme(2020)]%
        {van2020circuit}
\bibfield{author}{\bibinfo{person}{Ewout Van Den~Berg} {and} \bibinfo{person}{Kristan Temme}.} \bibinfo{year}{2020}\natexlab{}.
\newblock \showarticletitle{Circuit optimization of Hamiltonian simulation by simultaneous diagonalization of Pauli clusters}.
\newblock \bibinfo{journal}{\emph{Quantum}}  \bibinfo{volume}{4} (\bibinfo{year}{2020}), \bibinfo{pages}{322}.
\newblock


\bibitem[Whitfield et~al\mbox{.}(2011)]%
        {Whitfield_2011}
\bibfield{author}{\bibinfo{person}{James~D. Whitfield}, \bibinfo{person}{Jacob Biamonte}, {and} \bibinfo{person}{Al{{\'a} }n Aspuru-Guzik}.} \bibinfo{year}{2011}\natexlab{}.
\newblock \showarticletitle{Simulation of electronic structure Hamiltonians using quantum computers}.
\newblock \bibinfo{journal}{\emph{Molecular Physics}} \bibinfo{volume}{109}, \bibinfo{number}{5} (\bibinfo{date}{mar} \bibinfo{year}{2011}), \bibinfo{pages}{735--750}.
\newblock
\urldef\tempurl%
\url{https://doi.org/10.1080/00268976.2011.552441}
\showDOI{\tempurl}


\bibitem[{Wikipedia contributors}(2023)]%
        {enwiki:1184915760}
\bibfield{author}{\bibinfo{person}{{Wikipedia contributors}}.} \bibinfo{year}{2023}\natexlab{}.
\newblock \bibinfo{title}{Minimum-cost flow problem --- {Wikipedia}{,} The Free Encyclopedia}.
\newblock
\newblock
\urldef\tempurl%
\url{https://en.wikipedia.org/w/index.php?title=Minimum-cost_flow_problem&oldid=1184915760}
\showURL{%
\tempurl}
\newblock
\shownote{[Online; accessed 16-November-2023]}.


\bibitem[Wurtz et~al\mbox{.}(2023)]%
        {wurtz2023aquila}
\bibfield{author}{\bibinfo{person}{Jonathan Wurtz}, \bibinfo{person}{Alexei Bylinskii}, \bibinfo{person}{Boris Braverman}, \bibinfo{person}{Jesse Amato-Grill}, \bibinfo{person}{Sergio~H. Cantu}, \bibinfo{person}{Florian Huber}, \bibinfo{person}{Alexander Lukin}, \bibinfo{person}{Fangli Liu}, \bibinfo{person}{Phillip Weinberg}, \bibinfo{person}{John Long}, \bibinfo{person}{Sheng-Tao Wang}, \bibinfo{person}{Nathan Gemelke}, {and} \bibinfo{person}{Alexander Keesling}.} \bibinfo{year}{2023}\natexlab{}.
\newblock \showarticletitle{Aquila: QuEra's 256-qubit neutral-atom quantum computer}.
\newblock \bibinfo{journal}{\emph{arXiv preprint arXiv:2306.11727}} (\bibinfo{year}{2023}).
\newblock


\end{thebibliography}

\newpage
\appendix
\section{Appendix}

\subsection{Pauli Matrices}
\label{sec:pauli-matrix}

The identity gate $I$ is the identity matrix $\mathbb{I}$. The matrix representation of the Pauli matrices is:
\begin{equation*}
    X=\mqty(0 & 1 \\ 1 & 0)\quad Y=\mqty(0 & -i \\ i & 0)\quad Z=\mqty(1 & 0 \\ 0 & -1)
\end{equation*}

\subsection{Proof of Theorem~\ref{theorem:correctness}}
\label{proof_lem: MarQSim}

\textbf{[Correctness and Approximation Error Bound]}
Given a Hamiltonian $\mathcal{H}=\sum_{j=1}^n h_j H_j$, the quantum circuit compiled by Algorithm~\ref{alg:ARQSC_FT} correctly approximates the operator $e^{i\mathcal{H}t}$ if the HTT graph and the corresponding transition matrix $\ \mathbf{P}$ satisfies the following two conditions:
\begin{enumerate}
    \item \textbf{Strong Connectivity}: The HTT graph is a strongly connected state transition graph, meaning that only one unique recurrence class containing all possible states exists. 
    \item \textbf{Stationary Distribution Preservation}: The distribution $\pi_i=|h_i|/(\sum_{j}|h_j|)$ is stationary under the transition matrix $\mathbf{P}$ (this is the initial distribution in our algorithm):
    $$
    \pi = \pi\mathbf{P} \mathrm{\quad where\quad} \pi =  \begin{pmatrix}
    \frac{|h_1|}{\sum_{j}|h_j|} & \frac{|h_2|}{\sum_{j}|h_j|} & \cdots & \frac{|h_n|}{\sum_{j}|h_j|} 
    \end{pmatrix}
    $$
  %  This make the 
\end{enumerate}

The approximation error $\epsilon$ is bounded by $\epsilon\lesssim2\lambda^2 t^2/N$, where $\lambda$ is the sum of the absolute values of the Hamiltonian term coefficients, $t$ is the simulated evolution time, and $N$ is the number of sampling steps. (defined in Algorithm~ \ref{alg:ARQSC_FT}, line 2). 

\begin{proof}

The first condition asserts the stationary distribution mentioned in condition (2) is unique, and every Hamiltonian can be reached for a given initial state. This can be proved by contradiction. Suppose the HTT graph is not strongly connected. Then there exists at least one pair of vertices $v_i$ and $v_j$ where we can not reach $v_j$ from $v_i$.
It is possible that we sample the term $H_i$ in the first step and then the term $H_j$ cannot be included in the entire sampling process. 
This will change the overall Hamiltonian we are simulating and lead to wrong results.
By contradiction, the HTT graph needs to be strongly connected.

Next, we construct the evolution channel for every step according to our transition matrix. We denote $\tau=\lambda t/N$.
\iffalse
The only difference between the random compiler and MarQSim is the relation between each sample is no longer independent.
\fi
For step $i$, the evolution channel can be constructed as:
\begin{align}
\nonumber
\mathcal{E}_{i}(\rho) & = \sum_{j=1}^{n} p^{(i)}_j e^{ i \tau H_j } \rho e^{ - i \tau H_j } 
\end{align}
where $p_j^{(i)}$ stands for the probability of getting the Hamiltonian term $H_j$ in the i-th sampling step. 

Note that our initial state probability $p^{(1)}$ in Algorithm~\ref{alg:ARQSC_FT} is exactly the stationary distribution $\pi$ of transition matrix, thus $p^{(i)} = (p_1^{(i)}, p_2^{(i)}, \ldots, p_n^{(i)})$ can be calculated recurrently:
\begin{equation*}
    p^{(i)} = p^{(i-1)}\mathbf{P}=\ldots=p^{(1)}\mathbf{P}^{i-1}=\pi
\end{equation*}
indicates the evolution channel in each step is identical:
\begin{align}
\nonumber
\mathcal{E}_{i}(\rho) = \sum_{j=1}^{n} p^{(i)}_j e^{ i \tau H_j } \rho e^{ - i \tau H_j } = \sum_{j=1}^{n} p_j e^{ i \tau H_j } \rho e^{ - i \tau H_j } = \mathcal{E}(\rho)
\end{align}
Using Taylor series expansions and $p_j=h_j/\lambda$, we have that to leading order in $\tau$,
\begin{align}
\label{RandChannelRough}
\nonumber
\mathcal{E}(\rho) & = \rho +i \sum_{j}^{n} \frac{h_j \tau}{ \lambda} (H_j \rho - \rho H_j ) + O(\tau^2) .
\end{align}

Consider the channel $\mathcal{U}_N$, which is $\frac{1}{N}$ of the full dynamics of the Hamiltonian simulation:

\begin{align}
\nonumber
\mathcal{U}_N( \rho ) & = e^{ i t H/N } \rho e^{-i t H/N} \\ 
\nonumber
& = \rho + i \frac{t}{N}( H \rho - \rho H ) + O \left( \frac{t^2 }{ N^2 } \right) ,
\end{align}
where we have expanded out to the leading order in $t/N$. Using $H= \sum_j h_j H_j$, we find:

\begin{align}
\nonumber
\mathcal{U}_N( \rho ) & = \rho + i \sum_{j}^{L} \frac{t h_j}{N} ( H_j \rho - \rho H_j ) + O \left( \frac{t^2 }{ N^2 } \right) .
\end{align}

Comparing $\mathcal{E}$ and $\mathcal{U}_N$, we see that the zeroth and first-order terms match.
The higher-order terms typically do not match, and a more rigorous analysis (see \cite{campbell2019random} for the details) shows that the channels $\mathcal{E}$ and $\mathcal{U}_N$ differ by an amount bounded by
\begin{equation}
\label{deltaBound}
\nonumber
\delta \leq \frac{ 2 \lambda^2 t^2 }{ N^2 } e^{ 2 \lambda t / N } \approx \frac{ 2 \lambda^2 t^2 }{ N^2 } .
\end{equation}

Thus the error $\Delta$ of $N$ repetitions $\mathcal{E}_{1} \mathcal{E}_{2}\cdots\mathcal{E}_{N}=\mathcal{E}^{N}$ relative to the target unitary $U$ is
\begin{equation}
\nonumber
\Delta = N \delta \lesssim \frac{ 2 \lambda^2 t^2}{N} .
\end{equation}

\end{proof}

\subsection{Proof of Theorem~\ref{theorem:flow-constraint}}
\label{proof_lem:tranflow}

\textbf{[Preserving Stationary Distribution in the Flow Network]}
The transition matrix $\mathbf{P}_1$ extracted from a flow network mentioned above satisfies the second condition in Theorem~\ref{theorem:correctness} once the following equations hold:
\begin{equation*}
f(\edge{S,H^{prev}_i})=f(\edge{H^{next}_i,T})=\pi_i 
%\left\{\begin{array}{lr}
%     f(\edge{S,H^{prev}_i})=\pi_i \\
%     f(\edge{H^{next}_j,T})=\pi_j 
%\end{array}\right.
% \label{TH5.1eq}
\end{equation*}
for all $1\leq i\leq n$.

\begin{proof}
Given $\pi=(\pi_1, \pi_2, \ldots, \pi_n)$ defined in Theorem~\ref{theorem:correctness}, the product of $\pi$ and the transition matrix $\mathbf{P}$ can be calculated as:
\begin{align}    
         (\mathbf{\pi} \mathbf{P})(i) &=\sum_{j} \mathbf{\pi_j} \mathbf{P}(j,i) \notag\\ 
        & =\sum_{j} \mathbf{\pi_j} \frac{f(\edge{H^{prev}_j,H^{next}_i})}{ f(\edge{S, H^{prev}_j})}  \notag\\ 
        & = \sum_{j} f(\edge{H^{prev}_j,H^{next}_i}) \notag\\ 
        & = f(\edge{H^{next}_i,T}) \notag\\ 
        & = \mathbf{\pi_i} \notag
\end{align}
which indicates that $\pi$ is the stationary distribution of $\mathbf{P}$.

A special case is that there exists $i$, s.t. $\pi_i>0.5$. 
In this case, a flow satisfies Equation~\eqref{TH5.1eq} does not exist.
This issue can be mitigated by further decomposing $H$ into $\sum_{j\ne i}h_jH_j+0.5\cdot h_iH_i+0.5\cdot h_iH_i$.
Then there exists no $\pi_i > 0.5$. % and find the flow.
%Specifically, not all cases exist a flow satisfies Equation~\eqref{TH5.1eq}, which could be broken when $\exists i,\mathbf{\pi_i}>0.5$. 

\iffalse
\begin{align}    
         (\mathbf{\pi} \mathbf{P})(i) &=\sum_{j} \mathbf{\pi_j} \mathbf{P}(j,i) \notag\\ 
        & =\sum_{j} \mathbf{\pi_j} \frac{f(\langle v_j,v_i\rangle)}{\mathbf{\pi_j}}  \notag\\ 
        & = \sum_{j} f(\langle v_j,v_i\rangle) \notag\\ 
        & = f(\langle v_i,v_i\rangle) + \sum_{j \in \{j|j\ne i\}} f(\langle v_j,v_i\rangle) \notag\\ 
        & = \sum_{j \ne i} f(\langle v_j,v_i\rangle) \notag\\ 
        & = \mathbf{\pi_i} \notag
\end{align}
\fi

\end{proof}

\subsection{Proof of Proposition~\ref{theorem:cost-is-expectation}}
\label{sec:proof-cost-is-expectation}

\textbf{[Cost is the Expectation of CNOT Gate Count]}
Given the MCFP framework in Fig.~\ref{fig:flowgraph-problem}, when the cost function $w(\edge{i,j})$ in the flow network can represent the number of CNOT gates required when transiting between the Hamiltonian terms $e^{i\lambda tH_i/N}$ and $e^{i\lambda tH_j/N}$,
%Given $G_{\mathbf{P_1}}$, $C(\edge{i,j})$, $w(\edge{i,j})$, any flow $f(\edge{i,j})$ that satisfies Theorem~\ref{theorem:flow-constraint}, and the generated transition matrix $\mathbf{P_{gc}}$, 
the total cost $\mathcal{W}=\sum_{\edge{i,j}\in E}f(\edge{i,j})w(\edge{i,j})$ is the mathematical expectation number of CNOT gates between circuit snippets of $H^{prev}$ and $H^{next}$ if the circuit is compiled via $\mathbf{P_{gc}}$.

\iffalse
Given $G_{\mathbf{P_1}}$, $C(\edge{i,j})$, $w(\edge{i,j})$, any flow $f(\edge{i,j})$ that satisfies Theorem~\ref{theorem:flow-constraint}, and the generated transition matrix $\mathbf{P}_1$, the cost $\mathcal{W}=\sum_{\edge{i,j}\in E}f(\edge{i,j})w(\edge{i,j})$ is the mathematical expectation number of CNOT gates between circuit snippets of $H^{prev}$ and $H^{next}$ if the circuit is compiled via $\mathbf{P}_1$.
\fi

\begin{proof}
    
%\yhliu{
We denote the random variable $\#G$ as the CNOT gate count between $H^{prev}$ and $H^{next}$. To calculate the expectation, we could decompose into each case $\#G_{ij}$ of the transition (sampling) from $H^{prev}_i$ to $H^{next}_j$. Note that, based on the definition, there is always $\#G_{ij}=CNOT\_count(H^{prev}_i,H^{next}_j)$.%}

\begin{align}
    \mathrm{E}[\#G]&=\sum_{j}\Pr[H^{prev}=H^{prev}_i]\cdot\mathrm{E}[\#G_{ij}|H^{prev}=H^{prev}_i] \notag\\
    &=\sum_i\pi_i\sum_j\Pr[H^{next}=H^{next}_j|H^{prev}=H^{prev}_i]\cdot CNOT\_count(H^{prev}_i,H^{next}_j) \notag\\
    &=\sum_i\pi_i\cdot\sum_j\mathbf{P_{gc}}_{ij}\cdot CNOT\_count(H^{prev}_i,H^{next}_j)) \notag\\
    &=\sum_i\sum_j f(\edge{H^{prev}_i,H^{next}_j})\cdot CNOT\_count(H^{prev}_i,H^{next}_j)\notag \\
    &= \sum_i\sum_j f(\edge{H^{prev}_i,H^{next}_j})\cdot w(H^{prev}_i,H^{next}_j)\notag \\
    &= \sum_{\edge{i,j}\in E} f(\edge{i,j})\cdot w(\edge{i,j})
\end{align}

%\yhliu{
The last step holds since the cost defined for edges $\edge{S,H^{prev}}$ and $\edge{H^{next},T}$ is $0$.%}
\end{proof}

\end{document}